\documentclass[11pt]{amsart}
\usepackage{mathrsfs}
\usepackage{amssymb}
\usepackage{amsmath}
\usepackage{amscd}
\newtheorem{theorem}{Theorem}[subsection]
\newtheorem{prop}[theorem]{Proposition}
\newtheorem{defn}[theorem]{Definition}
\newtheorem{lemma}[theorem]{Lemma}
\newtheorem{coro}[theorem]{Corollary}
\newtheorem{prop-def}{Proposition-Definition}[section]

\newtheorem{remark}[theorem]{Remark}

\newtheorem{exam}[theorem]{Example}

\begin{document}
\setlength{\oddsidemargin}{0cm} \setlength{\evensidemargin}{0cm}

\title{Double constructions of Frobenius algebras, Connes cocycles and their duality}

\author{Chengming Bai}

\address{Chern Institute of Mathematics \& LPMC, Nankai University,
Tianjin 300071, P.R. China} \email{baicm@nankai.edu.cn}

\def\shorttitle{Double Constructions of Frobenius Algebras and Connes cocycles}

\begin{abstract}

We construct an associative algebra with a decomposition into the
direct sum of the underlying vector spaces of another associative
algebra and its dual space such that both of them are subalgebras
and the natural symmetric bilinear form is invariant or the natural
antisymmetric bilinear form is a Connes cocycle. The former is
called a double construction of Frobenius algebra and the latter is
called a double construction of Connes cocycle which is interpreted
in terms of dendriform algebras. Both of them are equivalent to a
kind of bialgebras, namely, antisymmetric infinitesimal bialgebras
and dendriform D-bialgebras  respectively. In the coboundary cases,
our study leads to what we call associative Yang-Baxter equation in
an associative algebra and $D$-equation in a dendriform algebra
respectively, which are analogues of the classical Yang-Baxter
equation in a Lie algebra. We show that an antisymmetric solution of
associative Yang-Baxter equation corresponds to the antisymmetric
part of a certain operator called ${\mathcal O}$-operator which
gives a double construction of Frobenius algebra, whereas a
symmetric solution of $D$-equation corresponds to the symmetric part
of an ${\mathcal O}$-operator which gives a double construction of
Connes cocycle. By comparing antisymmetric infinitesimal bialgebras
and dendriform D-bialgebras, we observe that there is a clear
analogy between them. Due to the correspondences between certain
symmetries and antisymmetries appearing in the analogy, we regard it
as a kind of duality.

\end{abstract}



\maketitle

\tableofcontents \setcounter{section}{0}

\baselineskip=18pt

\section{Introduction}

Throughout this paper, an associative algebra is a nonunital
associative algebra. There are two important (nondegenerate)
bilinear forms on an associative algebra given as follows.

\begin{defn}{\rm
A bilinear form ${\mathcal B}(\;,\;)$ on an associative algebra $A$
is {\it invariant} if
$${\mathcal B}(xy,z)={\mathcal B}(x,yz),\;\;\forall\; x,y,z\in A.\eqno
(1.0.1)$$}
\end{defn}
\begin{defn}
{\rm An antisymmetric bilinear form $\omega(\;,\;)$ on an
associative algebra $A$ is a {\it  cyclic 1-cocycle in the sense of
Connes} if
$$\omega (xy,z)+\omega (yz,x)+\omega(zx,y)=0,\;\;\forall\; x,y,z\in
A.\eqno (1.0.2)$$ We also call $\omega$ a {\it Connes cocycle} for
abbreviation.}
\end{defn}

\subsection{Frobenius algebras}

A Frobenius algebra $(A,{\mathcal B})$ is an associative algebra
$A$ with a nondegenerate invariant bilinear form ${\mathcal
B}(\;,\;)$. It was first studied by Frobenius (\cite{Fro}) in 1903
and then named by Brauer and Nesbitt (\cite{BrN}). In fact,
Frobenius algebras appear in many fields in mathematics and
mathematical physics, such as (modular) representations of finite
groups (\cite{Kap}), Hopf algebras (\cite{LS}), statistical models
over 2-dimensional graphs (\cite{BFN}), Yang-Baxter equation
(\cite{St}), Poisson brackets of hydrodynamic type (\cite{BaN})
and so on. In particular, they play a key role in the study of
topological quantum field theory (\cite{Ko}, \cite{RFFS}, etc.).
There are a lot of references on the study of Frobenius algebras
(for example, see \cite{Kap} or \cite{Y} and the references
therein).

A Frobenius algebra $(A,{\mathcal B})$ is symmetric if $\mathcal B$
is symmetric. In this paper, we mainly consider a class of symmetric
Frobenius algebras $(A,{\mathcal B})$ satisfying the following
conditions:

(1) $A=A_1\oplus A_1^*$ as the direct sum of vector spaces;

(2) $A_1$ and $A_1^*$ are associative subalgebras of $A$;

(3) ${\mathcal B}$ is the natural symmetric bilinear form on
$A_1\oplus A_1^*$ given by
$${\mathcal B}(x+a^*,y+b^*)=\langle  x,b^*\rangle   +\langle  a^*,y\rangle   ,\;\;\forall x,y\in
A_1,\;\;a^*,b^*\in A_1^*,\eqno (1.1.1)$$ where $\langle  ,\rangle
$ is the natural pair between the vector space $A_1$ and its dual
space $A_1^*$. We call it a double construction of Frobenius
algebra.

Such a double construction of Frobenius algebra is quite different
from the ``double extension construction" of a Lie algebra with a
nondegenerate invariant bilinear form (\cite{Kac}, [MR1-2], etc.) or
the ``$T^*$-extension'' of Frobenius algebra given by Bordemann in
\cite{Bo}.

Moreover, the above double constructions of Frobenius algebras were
also considered by Zhelyabin in \cite{Z} and Aguiar in \cite{A3}
(under the name of ``balanced Drinfeld double $D_b(A)$") with
different motivations and approaches respectively. They are closely
related to Lie bialgebras. Lie bialgebras were introduced by
Drinfeld (\cite{D}) and they play a crucial role in symplectic
geometry and quantum groups. They are equivalent to Manin triples
(see \cite{CP} and the references therein or subsection 5.2).

It is easy to show that the commutator of a Frobenius algebra from
the above double construction gives a Manin triple (hence a Lie
bialgebra). Furthermore, such a double construction has many
properties similar to a Lie bialgebra. It is equivalent to an
antisymmetric infinitesimal bialgebra (which is the same structure
under the names of ``associative D-algebra" in \cite{Z} and
``balanced infinitesimal bialgebra" in the sense of the opposite
algebra in \cite{A3}) and under a ``coboundary" condition, it leads
to an analogue of the classical Yang-Baxter equation (\cite{Se}) in
an associative algebra $A_1$
$$r_{12}r_{13}+r_{13}r_{23}-r_{23}r_{12}=0,\eqno (1.1.2)$$
where $r=\sum\limits_i x_i\otimes y_i \in A_1\otimes  A_1$ and
$$r_{12}r_{13}=\sum_{i,j} x_ix_j\otimes  y_i \otimes y_j,\;\;r_{13}r_{23}=\sum_{i,J}
x_i\otimes
 x_j\otimes  y_iy_j,\;\;r_{23}r_{12}=\sum_{i,j}  x_j\otimes  x_iy_j\otimes y_i.\eqno
(1.1.3)$$ In particular, an antisymmetric solution of the above
equation in $A_1$ gives a double construction of Frobenius algebra
$(A=A_1\oplus A_1^*, {\mathcal B})$.

On the other hand, we introduce the new notion of antisymmetric
infinitesimal bialgebra in order to express explicitly its relation
with the known notion of infinitesimal bialgebra, although there are
certain notions for the same or similar structures. An infinitesimal
bialgebra is a triple $(A,m,\Delta)$, where $(A,m)$ is an
associative algebra, $(A,\Delta)$ is a coassociative algebra and
$$\Delta(ab)=\sum ab_1\otimes  b_2+\sum a_1\otimes  a_2b,\;\;\forall a,b\in A.\eqno (1.1.4)$$
It was introduced by Join and Rota (\cite{JR}) in order to provide
an algebraic framework for the calculus of divided difference.
Furthermore, Aguiar studied the cases of principal derivations and
gave the associative Yang-Baxter equation (\cite{A1})
$$r_{13}r_{12}-r_{12}r_{23}+r_{23}r_{13}=0.\eqno (1.1.5)$$
Note that equation (1.1.2) is equation (1.1.5) in the opposite
algebra and, when $r$ is antisymmetric, equation (1.1.5) is just
equation (1.1.2) under the operation $\sigma_{13} (x\otimes y\otimes
z)=z\otimes y\otimes x$.

We would like to point out that although there have been many
results on the double constructions of Frobenius algebras, there has
not been a complete and explicit interpretation yet. In fact, most
of these results were given in a scattered way with different
motivations. For example,  Zhelyabin in \cite{Z} introduced the
notion of associative D-algebra as an important step to develop a
bialgebra theory of Jordan algebras (there was not an explicit study
of coboundary cases for the associative algebras themselves). In
\cite{A3}, Aguiar introduced the notion of balanced infinitesimal
bialgebra and then studied the antisymmetric solution of equation
(1.1.5) in order to compare them with Lie bialgebras and the
classical Yang-Baxter equation in a Lie algebra respectively, and
the balanced Drinfeld double $D_b(A)$ appears as an important
consequence. We will formulate the known results by a different and
systematic approach (for example, the ``invariant" antisymmetry
appears naturally). Moreover such an approach is useful and
convenient for the whole study in this paper.

\subsection{${\mathcal O}$-operators and dendriform algebras}

When $r$ is antisymmetric, besides the standard tensor form (1.1.2)
or (1.1.5), the associative Yang-Baxter equation has an equivalent
operator form, that is, a special case of a certain operator called
${\mathcal O}$-operator.  An ${\mathcal O}$-operator associated to a
bimodule $(l,r,V)$ of an associative algebra $A$ is a linear map
$T:V\rightarrow A$ satisfying
$$T(u)\cdot T(v)=T(l(T(u))v+r(T(v)u)),\;\;\forall\; u,v\in V.\eqno
(1.2.1)$$ In fact, an antisymmetric solution of associative
Yang-Baxter equation is an ${\mathcal O}$-operator associated to the
bimodule $(R^*,L^*)$. The notion of ${\mathcal O}$-operator was
introduced in \cite{BGN1} (such a structure appeared independently in
\cite{U} under the name of generalized Rota-Baxter operator) which
is an analogue of the ${\mathcal O}$-operator defined by Kupershmidt
as a natural generalization of the operator form of the classical
Yang-Baxter equation (\cite{Ku3} and a further study in
\cite{Bai1}). Conversely, the antisymmetric part of an ${\mathcal
O}$-operator satisfies the associative Yang-Baxter equation in a
larger associative algebra.

From an ${\mathcal O}$-operator, one can get a dendriform algebra.
Dendriform algebras are equipped with an associative product which
can be written as a linear combination of nonassociative
compositions. They were introduced by Loday (\cite{Lo1}) with
motivation from algebraic $K$-theory and have been studied quite
extensively with connections to several areas in mathematics and
physics, including operads (\cite{Lo3}), homology ([Fra1-2]), Hopf
algebras (\cite{Cha2}, [H1-2], \cite{Ron}, \cite{LR2}), Lie and
Leibnitz algebras (\cite{Fra2}), combinatorics (\cite{LR1}),
arithmetic(\cite{Lo2}) and quantum field theory (\cite{F1}) and so
on (see \cite{EMP} and the references therein).

Furthermore, there is a compatible dendriform algebra structure on
an associative algebra $A$ if and only if there exists an invertible
${\mathcal O}$-operator  of $A$, or equivalently, there exists an
invertible (usual) 1-cocycle (see equation (3.1.6)) associated to
certain suitable bimodule of $A$ (\cite{BGN2}). Thus a close relation
between the associative Yang-Baxter equation (hence the
antisymmetric infinitesimal bialgebras and the double constructions
of Frobenius algebras) and dendriform algebras is obviously given
(see also \cite{A3}, [E1-2]).

\subsection{Connes cocycles}

 Note that a Connes cocycle given by equation (1.0.2) is in fact a Hochschild 2-cocycle
 which satisfies antisymmetry. It corresponds
to the original definition of cyclic cohomology by Connes ([C]). Also note that in cyclic cohomology a cyclic
$n$-cocycle in the sense of Connes is an ${n+1}$ linear form,
although a Connes cocycle was called a cyclic 2-cocycle in some
references (like \cite{A3}) from some different viewpoints.
Moreover, although Connes used it in the unital framework and
in the nonunital framework cyclic homology has a very different
behavior, we still use the terminology ``Connes cocycle" in this
paper.

In this paper, we will see that, from a nondegenerate Connes cocycle
on an associative algebra $A$, one can get a compatible dendriform
algebra structure on $A$. Moreover, the dendriform algebra
structures play a key role in the following constructions of
nondegenerate Connes cocycles, which is one of the main contents in
this paper. We call $(A,\omega)$ a double construction of Connes
cocycle if it satisfies the following conditions:

(1) $A=A_1\oplus A_1^*$ as the direct sum of vector spaces;

(2) $A$ is an associative algebra and $A_1$ and $A_1^*$ are
associative subalgebras of $A$;

(3) $\omega$ is the natural antisymmetric bilinear form on
$A_1\oplus A_1^*$ given by
$$\omega(x+a^*,y+b^*)=-\langle  x,b^*\rangle   +\langle  a^*,y\rangle   ,\;\;\forall x,y\in
A_1,\;\;a^*,b^*\in A_1^*,\eqno (1.4.1)$$ and $\omega$ is a Connes
cocycle on $A$.

In this paper, the double construction of Connes cocycle is
interpreted in terms of dendriform algebras. We find that such a
structure is quite similar to a double construction of Frobenius
algebra or a Lie bialgebra. Briefly speaking, a double construction
of Connes cocycle is equivalent to a certain bialgebra structure,
namely, a dendriform D-bialgebra  structure. Both antisymmetric
infinitesimal bialgebras and dendriform D-bialgebras have many
similar properties as Lie bialgebras. In particular, there are the
so-called coboundary dendriform D-bialgebras which lead to another
analogue ($D$-equation in a dendriform algebra) of the classical
Yang-Baxter equation. A symmetric solution of the $D$-equation
corresponds to the symmetric part of an ${\mathcal O}$-operator,
which gives a double construction of Connes cocycle.

\subsection{Duality between bialgebras}

By comparing antisymmetric infinitesimal bialgebras and  dendriform
D-bialgebras, we observe that there is a clear analogy between them.
Moreover, due to the correspondences between certain symmetries and
antisymmetries appearing in the analogy, we regard it as a kind of
duality.

There is a similar study in the version of Lie algebras (\cite{CP},
\cite{Bai2}). In fact, there is also a double construction of a Lie
algebra with a nondegenerate invariant bilinear form (Manin triple
or Lie bialgebra) or with a nondegenerate 2-cocycle of Lie algebra
(parak\"ahler Lie algebra or pre-Lie bialgebra). There are the
${\mathcal O}$-operators and a kind of algebras called pre-Lie
algebras (Lie-admissible algebras whose left multiplication
operators form a Lie algebra) which play the same roles of the
${\mathcal O}$-operators and dendriform algebras. And there is a
similar duality between Lie bialgebras and pre-Lie bialgebras.

Moreover, due to Chapoton (\cite{Cha1}), there is a close
relationship among the Lie algebras, associative algebras, pre-Lie
algebras and dendriform algebras as follows (in the sense of
commutative diagram of categories).
$$\begin{matrix} {\rm dendriform\quad algebras} & \longrightarrow & \mbox{pre-Lie algebras} \cr \downarrow & &\downarrow\cr {\rm
associative\quad algebras} & \longrightarrow & {\rm Lie\quad
algebras} \cr\end{matrix}$$ We will extend the above relationship at
the level of bialgebras with the dualities in a commutative diagram.
In particular, the relation between antisymmetric infinitesimal
bialgebras (the special case of infinitesimal Hopf algebras) and Lie
bialgebras have been mentioned in \cite{A3}. Furthermore, these
types of bialgebras fit into the general framework of ``generalized
bialgebras" as introduced by  Loday in \cite{Lo4}.

The paper is organized as follows. In section 2, we give an
explicit and systematic study on the double constructions of
Frobenius algebras and then get the associative Yang-Baxter
equation naturally. In section 3, we introduce the close relations
between ${\mathcal O}$-operators and dendriform algebras. In
section 4, we study the double constructions of Connes cocycles in
terms of dendriform algebras. In section 5, we give the clear
analogy between antisymmetric infinitesimal bialgebras and
dendriform D-bialgebras, which we regard it as a kind of duality.
After recalling a similar duality between Lie bialgebras and
pre-Lie bialgebras, we express a close relationship among
associative algebras, Lie algebras, pre-Lie algebras and
dendriform algebras at the level of bialgebras.

Throughout this paper, all algebras are finite-dimensional, although
many results still hold in the infinite-dimensional case.

\section{Double constructions of Frobenius algebras and another
approach to associative Yang-Baxter equation}

\subsection{Bimodules and matched pairs of associative algebras}

\begin{defn} {\rm Let $A$ be an associative algebra and $V$ be a vector space. Let
$l, r:A\rightarrow \frak g\frak l (V)$ be two linear maps. $V$ (or the pair
$(l,r)$, or $(l,r,V)$) is called a {\it bimodule} of $A$ if}
$$l(xy)v=l(x)l(y)v,\;r(xy)v=r(y)r(x)v,\;l(x)r(y)v=r(y)l(x)v,\;\forall\; x,y\in A, v\in V.\eqno (2.1.1)$$
\end{defn}

In fact, according to \cite{Sc}, $(l,r, V)$ is a bimodule of an
associative algebra $A$ if and only if the direct sum $A\oplus V$ of
vector spaces is turned into an associative algebra ( the semidirect
sum) by defining multiplication in $A\oplus V$ by
$$(x_1+v_1)*(x_2+v_2)=x_1\cdot x_2+(l(x_1)v_2+r(x_2)v_1),\;\;\forall x_1,x_2\in A, v_1,v_2\in
V.
\eqno(2.1.2)$$ We denote it by $A\ltimes_{l,r}V$ or simply $A\ltimes
V$.

The following conclusion is obvious.

\begin{lemma}
Let $(l,r,V)$ be a bimodule of an associative algebra $A$.

{\rm (1)} Let $l^*,r^*:A\rightarrow \frak g\frak l (V^*)$ be the linear maps
given by
$$\langle  l^*(x)u^*,v\rangle   =\langle  l(x)v,u^*\rangle   ,\;\langle  r^*(x)u^*,v\rangle   =\langle  r(x)v,u^*\rangle   ,\;\forall x\in A, u^*\in V^*, v\in V.\eqno (2.1.3)$$
Then $(r^*,l^*,V^*)$ is a bimodule of $A$.

{\rm (2)} $(l,0,V)$, $(0,r,V)$, $(r^*,0,V^*)$ and $(0,l^*,V^*)$ are
bimodules of $A$.
\end{lemma}

\begin{exam}
{\rm Let $A$ be an associative algebra. Let $L(x)$ and $R(x)$ denote
the left and right multiplication operator respectively, that is,
$L(x)(y)=xy$, $R(x)(y)=yx$ for any $x,y\in A$. Let $L:A\rightarrow
\frak g\frak l (A)$ with $x\rightarrow L(x)$ and $R: A\rightarrow \frak g\frak l (A)$  with
$x\rightarrow R(x)$ (for every $x\in A$) be two linear maps. Then
$(L,0)$, $(0, R)$ and $(L,R)$ are bimodules of $A$. On the other
hand, $(R^*,0)$, $(0,L^*)$ and $(R^*,L^*)$ are bimodules of $A$,
too.}
\end{exam}

\begin{theorem}  Let $(A,\cdot)$ and $(B,\circ)$ be two associative algebras.
Suppose that there are linear maps $l_A,r_A:A\rightarrow \frak g\frak l (B)$ and
$l_B,r_B:B\rightarrow \frak g\frak l (A)$ such that $(l_A,r_A)$ is a bimodule of
$A$ and $(l_B,r_B)$ is a bimodule of $B$ and they satisfy the
following conditions:
$$l_A(x)(a\circ b)=l_A(r_B(a)x)b+(l_A(x)a)\circ b;\eqno (2.1.4)$$
$$r_A(x)(a\circ b)=r_A(l_B(b)x)a+a\circ (r_A(x)b);\eqno (2.1.5)$$
$$l_B(a)(x\cdot y)=l_B(r_A(x)a)y+(l_B(a)x)\cdot y;\eqno (2.1.6)$$
$$r_B(a)(x\cdot y)=r_B(l_A(y)a)x+x\cdot (r_B(a)y);\eqno (2.1.7)$$
$$l_A(l_B(a)x)b+(r_A(x)a)\circ b-r_A(r_B(b)x)a-a\circ
(l_A(x)b)=0;\eqno (2.1.8)$$
$$l_B(l_A(x)a)y+(r_B(a)x)\cdot y-r_B(r_A(y)a)x-x\cdot
(l_B(a)y)=0,\eqno (2.1.9)$$ for any $x,y\in A,a,b\in B$. Then there
is an associative algebra structure on the direct sum $A\oplus B$ of
the underlying vector spaces of $A$ and $B$ given by
$$(x+a)*(y+b)=(x\cdot y+l_B(a)y+r_B(b)x)+(a\circ b+l_A(x)b+r_A(y)a),\;\;\forall x,y\in A,a,b\in B.\eqno (2.1.10)$$
We denote this associative algebra by
$A\bowtie^{l_A,r_A}_{l_B,r_B}B$ or simply $A\bowtie B$.  On the
other hand, every associative algebra with a decomposition into the
direct sum of the underlying vector spaces of two subalgebras can be
obtained from the above way.
\end{theorem}

\begin{proof} It is straightforward.
\end{proof}

\begin{defn}{\rm
Let $(A,\cdot)$ and $(B,\circ)$ be two associative algebras. Suppose
that there are linear maps $l_A,r_A:A\rightarrow \frak g\frak l (B)$ and
$l_B,r_B:B\rightarrow \frak g\frak l (A)$ such that $(l_A,r_A)$ is a bimodule of
$A$ and $(l_B,r_B)$ is a bimodule of $B$. If equations
(2.1.4)-(2.1.9) are satisfied, then $(A,B,l_A,r_A,l_B,r_B)$ is
called a {\it matched pair of associative algebras}.}
\end{defn}

\begin{remark}{\rm
Obviously $B$ is an ideal of $A\bowtie B$ if and only if
$l_B=r_B=0$. If $B$ is a trivial (that is, all the products of $B$
are zero) ideal, then $A\bowtie^{l_A,r_A}_{0,0} B\cong
A\ltimes_{l_A,r_A} B$. Moreover, some other special cases of Theorem
2.1.4 have already been studied. For example, the case that $A$ is
left $B$-module and $B$ is a right $A$-module was considered in
\cite{A1}, that is, $l_A=0$ and $r_B=0$.}
\end{remark}

\subsection{Double constructions of Frobenius algebras and antisymmetric
infinitesimal bialgebras} Recall that a (symmetric) Frobenius
algebra is an associative algebra $A$ with a nondegenerate
(symmetric) invariant bilinear form. Let $(A,\cdot)$ be an
associative algebra. Suppose that there is an associative algebra
structure ``$\circ$" on its dual space $A^*$. We construct an
associative algebra structure on the direct sum $A\oplus A^*$ of the
underlying vector spaces of $A$ and $A^*$ such that $(A,\cdot)$ and
$(A^*,\circ)$ are subalgebras and the symmetric bilinear form on
$A\oplus A^*$ given by equation (1.1.1) is invariant. That is,
$(A\oplus A^*,{\mathcal B})$ is a symmetric Frobenius algebra. Such
a construction is called a double construction of Frobenius algebra
associated to $(A,\cdot)$ and $(A^*,\circ)$ and we denote it by
$(A\bowtie A^*, {\mathcal B})$.

\begin{theorem}
Let $(A,\cdot)$ be an associative algebra. Suppose that there is an
associative algebra structure ``$\circ$" on its dual space $A^*$.
Then there is a double construction of Frobenius algebra associated
to $(A, \cdot)$ and $(A^*,\circ)$ if and only if $(A,A^*,
R_\cdot^*,L_\cdot^*,R_\circ^*,L_\circ^*)$ is a matched pair of
associative algebras.
\end{theorem}

\begin{proof}
If $(A,A^*, R_\cdot^*,L_\cdot^*,R_\circ^*,L_\circ^*)$ is a matched
pair of associative algebras, then it is straightforward to show
that the bilinear form (1.1.1) is invariant on the associative
algebra $A\bowtie^{R_\cdot^*,L_\cdot^*}_{R_\circ^*,L_\circ^*} A^*$
given by equation (2.1.10). Conversely, set
$$x*a^*=l_A(x)a^*+r_{A^*}(a^*)x,\;\;a^**x=l_{A^*}(a^*)x+r_A(x)a^*,\;\;\forall\;
x\in A, a^*\in A^*.$$ Then $(A,A^*,l_A,r_A,l_{A^*},r_{A^*})$ is a
matched pair of associative algebras. Note that
$$\langle  l_A(x)a^*,y\rangle   =\langle  r_A(y)a^*,x\rangle   =\langle  y\cdot x,
a^*\rangle   ,\langle  l_{A^*}(b^*)x,a^*\rangle   =\langle
r_{A^*}(a^*)x,b^*\rangle   =\langle  a^*\circ b^*,x\rangle,$$ where
$x,y\in A, a^*,b^*\in A^*$. Hence,
$l_A=R_\cdot^*,\;r_A=L_\cdot^*,\;l_{A^*}=R_\circ^*,\;r_{A^*}=L_\circ^*$.
\end{proof}

\begin{prop}
Let $(A,\cdot)$ be an associative algebra. Suppose that there is an
associative algebra structure ``$\circ$" on its dual space $A^*$.
Then $(A,A^*, R_\cdot^*,L_\cdot^*,R_\circ^*,L_\circ^*)$ is a matched
pair of associative algebras if and only if for any $x\in
A^*\;a^*,b^*\in A^*$,
$$R_\cdot^*(x)(a^*\circ b^*)=R_\cdot^*(L^*_\circ(a^*)x)b^*+(R_\cdot^*(x)a^*)\circ b^*;\eqno (2.2.1)$$
$$R_\cdot^*(R_\circ^*(a^*)x)b^*+L_\cdot^*(x)a^*\circ b^*=L_\cdot^*(L_\circ^*(b^*)x)a^*+a^*\circ (R_\cdot^*(x)b^*).\eqno (2.2.2)$$
\end{prop}

\begin{proof} Obviously, equation (2.2.1) is just equation (2.1.4) and
equation (2.2.2) is just equation (2.1.8) in the case
$l_A=R^*_\cdot, r_A=L^*_\cdot$,
$l_B=l_{A^*}=R^*_{\circ},r_B=r_{A^*}=L^*_\circ$. By equation
(2.1.3), it is easy to show that in this situation,
$${\rm equation}\;\; (2.1.4) \;\; \Longleftrightarrow {\rm equation}\;\; (2.1.5)\;\;
\Longleftrightarrow\;\; {\rm equation}\;\;
(2.1.6)\Longleftrightarrow {\rm equation}\;\; (2.1.7);$$
$${\rm equation}\;\; (2.1.8) \;\; \Longleftrightarrow {\rm
equation}\;\; (2.1.9).$$ Therefore the conclusion holds.\end{proof}

Before the next study, we give some notations as follows. Let $A$ be
an associative algebra. Let $\sigma:A\otimes A\rightarrow A\otimes
A$ be the exchange operator defined as
$$\sigma(x\otimes y)=y\otimes x,\;\;\forall x,y\in A.\eqno (2.2.3)$$

There are several ways to make $A\otimes A$ into a bimodule of $A$.
For example, let $id$ be the identity map on $A$. Then $(id \otimes
L, R\otimes id)$ given by (for any $x,a,b\in A$)
$$( id \otimes L)(x)(a\otimes b)=( id \otimes L(x))(a\otimes
b)=a\otimes xb,\;(R\otimes  id )(x)(a\otimes b)=(R(x)\otimes
 id )(a\otimes b)= ax\otimes b,\eqno (2.2.4)$$ is a bimodule of $A$.
Similarly, $(L\otimes  id ,  id \otimes R)$ is also a bimodule of
$A$. In fact, equation (1.1.4) given in the introduction can be
rewritten as
$$\Delta(ab)=(L(a)\otimes  id )\Delta(b)+( id \otimes R(b))\Delta(a),\eqno
(2.2.5)$$ which gives the notion of infinitesimal bialgebra
(\cite{JR}).

For a linear map $\phi:V_1\rightarrow V_2$, we denote the dual
(linear) map by $\phi^*:V_2^*\rightarrow V_1^*$ given by
$$\langle  v,\phi^*(u^*)\rangle   =\langle  \phi(v),u^*\rangle   ,\;\;\forall v\in V_1,u^*\in V_2.\eqno (2.2.6)$$

\begin{theorem}
Let $(A,\cdot)$ be an associative algebra. Suppose there is an
associative algebra structure $``\circ"$ on its dual space $A^*$
given by a linear map $\Delta^*: A^*\otimes A^*\rightarrow A^*$.
Then $(A,A^*, R_\cdot^*,L_\cdot^*,R_\circ^*,L_\circ^*)$ is a matched
pair of associative algebras if and only if $\Delta:A\rightarrow
A\otimes A $ satisfies the following two conditions:
$$\Delta(x\cdot y)=( id \otimes L_\cdot (x))\Delta
(y)+(R_\cdot(y)\otimes  id )\Delta (x);\eqno (2.2.7)$$
$$(L_\cdot (y)\otimes  id - id \otimes R_\cdot (y))\Delta(x)+\sigma [(L_\cdot (x)\otimes  id - id \otimes R_\cdot
(x))\Delta(y)]=0,\;\;\forall x,y\in A.\eqno (2.2.8)$$
\end{theorem}

\begin{proof} Let $\{e_1,\cdots,e_n\}$ be a basis of $A$ and $\{
e_1^*,\cdots, e_n^*\}$ be its dual basis. Set $e_i\cdot
e_j=\sum\limits_{k=1}^nc_{ij}^ke_k$ and $e_i^*\circ
e_j^*=\sum\limits_{k=1}^nf_{ij}^k e_k^*$. Therefore, we have $\Delta
(e_k)=\sum\limits_{i,j=1}^nf_{ij}^ke_i\otimes e_j$ and
$$R^*_\cdot (e_i)e_j^*=\sum_{k=1}^nc_{ki}^je_k^*,\;
L^*_\cdot (e_i)e_j^*=\sum_{k=1}^nc_{ik}^je_k^*,\;\; R^*_\circ
(e_i^*)e_j=\sum_{k=1}^nf_{ki}^je_k,\; L^*_\circ
(e_i^*)e_j=\sum_{k=1}^nf_{ik}^je_k.$$ Hence the coefficient of
$e_j\otimes e_k$ in
$$\Delta(e_i\cdot e_m)=( id \otimes L_\cdot (e_i))\Delta
(e_m)+(R_\cdot(e_m)\otimes  id )\Delta (e_i)$$ gives the following
relation (for any $i,j,k,m$)
$$\sum_{l=1}^n c_{mi}^lf_{jk}^l=\sum_{l=1}^n(
c_{ml}^kf_{jl}^i+c_{li}^jf_{lk}^m)$$ which is just the relation
given by the coefficient of $e_m^*$ in
$$R_\cdot^*(e_i)(e_j^*\circ e_k^*)=R_\cdot^*(L^*_\circ(e_j^*)e_i)e_k^*+(R_\cdot^*(e_i)e_j^*)\circ e_k^*.$$
Similarly, equation (2.2.8) corresponds to equation (2.2.2).
\end{proof}

\begin{remark} {\rm From the symmetry of the associative algebras $(A,\cdot)$ and
$(A^*,\circ)$ appearing in the double construction, we also can
consider the operation $\beta: A^*\rightarrow A^*\otimes A^*$ such
that $\beta^*: A\otimes A \rightarrow A$ gives an associative
algebra structure on $A$. It is easy to show that $\Delta$ satisfies
equations (2.2.7) and (2.2.8) if and only if $\beta$ satisfies
$$\beta(a^*\circ b^*)=( id \otimes L_\circ (a^*))\beta
(b^*)+(R_\circ(b^*)\otimes  id )\beta (a^*);\eqno (2.2.9)$$
$$(L_\circ (b^*)\otimes  id - id \otimes R_\circ (b^*))\beta(a^*)+\sigma [(L_\circ (a^*)\otimes  id - id \otimes R_\circ
(a^*))\beta(b^*)]=0,\;\;\forall a^*,b^*\in A.\eqno (2.2.10)$$}
\end{remark}

\begin{defn}
{\rm  Let $A$ be an associative algebra. An {\it antisymmetric
infinitesimal bialgebra} structure on $A$ is a linear map
$\Delta:A\rightarrow A\otimes A$ such that

(a) $\Delta^*:A^*\otimes A^*\rightarrow A^*$ defines an associative
algebra structure on $A^*$;

(b) $\Delta$ satisfies equations (2.2.7) and (2.2.8).

\noindent We denote it by $(A,\Delta)$ or $(A,A^*)$.}
\end{defn}

\begin{coro}
Let $(A,\cdot)$ and $(A^*,\circ)$ be two associative algebras. Then
the following conditions are equivalent.

{\rm (1)} There is a double construction of Frobenius algebra
associated to $(A,\cdot)$ and $(A^*,\circ)$;

{\rm (2)} $(A,A^*, R_\cdot^*,L_\cdot^*,R_\circ^*,L_\circ^*)$ is a
matched pair of associative algebras;

{\rm (3)} $(A,A^*)$ is an antisymmetric infinitesimal bialgebra.
\end{coro}

\begin{proof}
It follows from Theorems 2.2.1 and 2.2.3.
\end{proof}

\begin{remark}{\rm As we have pointed out in the introduction, an antisymmetric
infinitesimal bialgebra is exactly an associative D-algebra in
\cite{Z} where the above equivalence between (1) and (3) was given
and a balanced infinitesimal bialgebra in the sense of the opposite
algebra in \cite{A3} where the corresponding double construction of
Frobenius algebra was called a balanced Drinfeld double as an
important consequence. On the other hand, the notion of
antisymmetric infinitesimal bialgebra is due to the fact that
equation (2.2.7) (in the sense of the opposite algebra) corresponds
to equation (2.2.5) which gives the notion of infinitesimal
bialgebra and equation (2.2.8) expresses certain antisymmetry.}
\end{remark}

\begin{defn}{\rm Let $(A,\Delta_A)$ and $(B,\Delta_B)$ be two antisymmetric
infinitesimal bialgebras.  A {\it homomorphism of antisymmetric
infinitesimal bialgebras} $\varphi:A\rightarrow B$ is a homomorphism
of associative algebras such that
$$(\varphi\otimes \varphi) \Delta_A (x)=\Delta_B(\varphi(x)),\;\;
\forall x\in A.\eqno (2.2.11)$$  An {\it isomorphism of
antisymmetric infinitesimal bialgebras} is an invertible
homomorphism of antisymmetric infinitesimal bialgebras.}\end{defn}

\begin{defn}{\rm
 Let $(A_1\bowtie A_1^*,{\mathcal B}_1)$
and $(A_2\bowtie A_2^*,{\mathcal B}_2)$ be two double constructions
of Frobenius algebras. They are {\it isomorphic} if and only if
there exists an isomorphism of associative algebras $\varphi:
A_1\bowtie A_1^*\rightarrow A_2\bowtie A_2^*$ such that
$$\varphi(A_1)=A_2,\;\varphi (A_1^*)=A_2^*,\;{\mathcal B}_1(x,y)=\varphi^*{\mathcal B}_2(x,y)
={\mathcal B}_2(\varphi(x),\varphi (y)),\;\forall x,y\in A_1\bowtie
A_1^*.\eqno (2.2.12)$$}
\end{defn}

\begin{prop}
Two double constructions of Frobenius algebras are isomorphic if and
only if their corresponding antisymmetric infinitesimal bialgebras
are isomorphic.\end{prop}

\begin{proof}
Let $(A_1\bowtie A_1^*,{\mathcal B}_1)$ and
$(A_2\bowtie A_2^*,{\mathcal B}_2)$ be two double constructions of
Frobenius algebras. Let $\{e_1,\cdots,e_n\}$ be a basis of $A_1$ and
$\{e_1^*,\cdots, e_n^*\}$ be its dual basis. If $\varphi: A_1\bowtie
A_1^*\rightarrow A_2\bowtie A_2^*$ is an isomorphism of double
constructions of Frobenius algebras, then
$\varphi|_{A_1}:A_1\rightarrow A_2$ and
$\varphi|_{A_1^*}:A_1^*\rightarrow A_2^*$ are isomorphisms of
associative algebras. Moreover,
$\varphi|_{A_1^*}={(\varphi|_{A_1})^*}^{-1}$ since
\begin{eqnarray*}
\langle  \varphi|_{A_1^*}(e_i^*), \varphi (e_j)\rangle &=&{\mathcal
B}_2(\varphi|_{A_1^*}(e_i^*),\varphi(e_j)) ={\mathcal
B}_1(e_i^*,e_j)=\delta_{ij}=\langle  e_i^*,e_j\rangle\\
&=&\langle
\varphi^*{(\varphi|_{A_1})^*}^{-1}(e_i^*),e_j\rangle=\langle
{(\varphi|_{A_1})^*}^{-1}(e_i^*),\varphi(e_j)\rangle.
\end{eqnarray*}
Hence $(A_1,A_1^*)$ and $(A_2,A_2^*)$ are isomorphic as
antisymmetric infinitesimal bialgebras. Conversely, let
$\varphi':A_1\rightarrow A_2$ be an isomorphism between two
antisymmetric infinitesimal bialgebras $(A_1,A_1^*)$ and
$(A_2,A_2^*)$. Set $\varphi:A_1\oplus A_1^*\rightarrow A_2\oplus
A_2^*$ be a linear map given by
$$\varphi(x)=\varphi'(x),\varphi(a^*)=(\varphi'^*)^{-1}(a^*),\;\;\forall x\in A_1,a^*\in A_1^*.$$
Then it is easy to show that $\varphi$ is an isomorphism of double
constructions of Frobenius algebras between $(A_1\bowtie
A_1^*,{\mathcal B}_1)$ and $(A_2\bowtie A_2^*,{\mathcal
B}_2)$.\end{proof}

\begin{exam} {\rm Let $(A,\Delta)$ be an antisymmetric
infinitesimal bialgebra. Then its dual $(A^*,\beta)$ given in Remark
2.2.4 is also an antisymmetric infinitesimal bialgebra.}
\end{exam}

\begin{exam} {\rm Let $A$ be an associative algebra. If the
associative algebra structure on $A^*$ is trivial, then either
$(A,0)$ or $(A,A^*)$ is an antisymmetric infinitesimal bialgebra.
Moreover, its corresponding Frobenius algebra is given by the
semidirect sum $A\ltimes_{R^*, L^*}A^*$ with the natural invariant
bilinear form ${\mathcal  B}$  given by equation (1.1.1). Dually, if
$A$ is a trivial associative algebra, then the antisymmetric
infinitesimal bialgebra structures on $A$ are in one-to-one
correspondence with the associative algebra structures on $A^*$.}
\end{exam}

\begin{exam} {\rm Let $(A,A^*)$ be an antisymmetric
infinitesimal bialgebra. In the next subsection, we will prove that
there exists a canonical antisymmetric infinitesimal bialgebra
structure on the direct sum $A\oplus A^*$ of the underlying vector
spaces of $A$ and $A^*$.}
\end{exam}

\subsection{Coboundary (principal) antisymmetric infinitesimal bialgebras}

In fact, for an associative algebra $A$, $\Delta:A\rightarrow
A\otimes A$ satisfying equation (2.2.7) is a 1-cocycle or a
derivation of $A$ associated to the bimodule $( id \otimes L,
R\otimes  id )$. So it is natural to consider the special case that
$\Delta$ is a 1-coboundary or a principal derivation.

\begin{defn} {\rm An antisymmetric infinitesimal bialgebra $(A,\Delta)$ is
called {\it coboundary} if there exists a $r\in A\otimes A$ such
that
$$\Delta (x)=( id \otimes L(x)-R(x)\otimes  id )r,\;\;\forall x\in A.\eqno
(2.3.1)$$}\end{defn}

Let $A$ be an associative algebra and $r\in A\otimes A$. If
$\Delta:A\rightarrow A\otimes A$ is given by equation (2.3.1), then
it is obvious that $\Delta$ satisfies equation (2.2.7). Therefore,
$(A,\Delta)$ is an antisymmetric infinitesimal bialgebra if and only
if the following two conditions are satisfied:

(1) $\Delta^*:A^*\otimes A^*\rightarrow A^*$ defines an associative
algebra structure on $A^*$.

(2) $\Delta$ satisfies equation (2.2.8).


\begin{lemma} {\rm ([A1, Proposition 5.1])}\quad Let $A$ be an associative
algebra and $r\in A\otimes A$. Define $\Delta:A\rightarrow A\otimes
A$ by
$$\Delta(a)=[L(x)\otimes  id - id \otimes R(x)] r,\;\;\forall x\in A.\eqno (2.3.2)$$
Then $\Delta^*:A^*\otimes A^*\rightarrow A^*$ defines an associative
algebra structure on $A^*$ if and only if
$$(L(x)\otimes  id \otimes  id - id \otimes  id \otimes
R(x)) (r_{13}r_{12}+r_{23}r_{13}-r_{12}r_{23})=0,\;\;\forall x\in
A,\eqno (2.3.3)$$ where the notations
$r_{13}r_{12},r_{23}r_{13},r_{12}r_{23}$ are given similarly as
equation {\rm (1.1.3)}.
\end{lemma}

Therefore for (1), we use a similar discussion to get the following
conclusion.

\begin{prop}
Let $A$ be an associative algebra and $r\in A\otimes A$. Define
$\Delta:A\rightarrow A\otimes A$ by equation {\rm (2.3.1)}. Then
$\Delta^*:A^*\otimes A^*\rightarrow A^*$ defines an associative
algebra structure on $A^*$ if and only if
$$( id \otimes  id \otimes L(x)-R(x)\otimes  id \otimes
 id )(r_{12}r_{13}+r_{13}r_{23}-r_{23}r_{12})=0.\;\;\forall x\in A.\eqno (2.3.4)$$
\end{prop}

\begin{prop}
 Let $A$ be an associative algebra and
$r\in A\otimes A$. Define $\Delta:A\rightarrow A\otimes A$ by
equation {\rm (2.3.1)}. Then $\Delta$ satisfies equation {\rm
(2.2.8)} if and only if $r$ satisfies
$$ [L(x)\otimes  id - id \otimes R(x)][ id \otimes L(y)-R(y)\otimes
 id ](r+\sigma(r))=0,\;\;\forall x,y\in A.\eqno (2.3.5)$$
\end{prop}

\begin{proof}
It is straightforward.\end{proof}

Combining Proposition 2.3.3 and Proposition 2.3.4, we have the
following conclusion.

\begin{theorem} Let $A$ be an associative algebra and $r\in A\otimes A$. Then
the linear map $\Delta$ defined by equation {\rm (2.3.1)} induces an
associative algebra structure on $A^*$ such that $(A,A^*)$ is an
antisymmetric infinitesimal bialgebra if and only if equations {\rm
(2.3.4)} and {\rm (2.3.5)} are satisfied.
\end{theorem}

\begin{theorem} Let $(A,\Delta_A)$ be an antisymmetric
infinitesimal bialgebra. Then there is a canonical antisymmetric
infinitesimal bialgebra structure on the direct sum $A\oplus A^*$ of
the underlying vector spaces of $A$ and $A^*$ such that both the
inclusions $i_1:A\rightarrow A\oplus A^*$ and $i_2:A^*\rightarrow
A\oplus A^*$ into the two summands are homomorphisms of
antisymmetric infinitesimal bialgebras.  Here the antisymmetric
infinitesimal bialgebra structure on $A^*$ is $(A^*,-\beta_{A^*})$,
where $\beta_{A^*}:A^*\rightarrow A^*\otimes A^*$ is given in {\rm
Remark 2.2.4.}\end{theorem}

\begin{proof} Let $r\in A\otimes A^*\subset (A\oplus A^*)\otimes (A\oplus
A^*)$ correspond to the identity map $id:A\rightarrow A$. Let $\{
e_1,\cdots,e_n\}$ be a basis of $A$ and $\{ e_1^*,\cdots, e_n^*\}$
be its dual basis. Then $r=\sum\limits_{i=1}^n e_i\otimes e_i^*$.
Suppose that the associative algebra structure ``$*$'' on $A\oplus
A^*$ is given by ${\mathcal A}{\mathcal
D}(A)=A\bowtie^{R_\cdot^*,L_\cdot^*}_{R_\circ^*,L^*_\circ}A^*$. Then
by Theorem 2.1.4, we have (for any $x,y\in A, a^*,b^*\in A^*$)
$$x*y=x\cdot y,\;\;a^**b^*=a^*\circ b^*,\;x*a^*=R_\cdot^*(x)a^*+L_\circ^*(a^*)x,\;
a^**x=R_\circ (a^*)x+L^*_\cdot(x)a^*.$$  If $r$ satisfies equations
(2.3.4) and (2.3.5), then
$$\Delta_{{\mathcal A}{\mathcal D}}(u)=( id \otimes L(u)-R(u)\otimes  id )r,\;\;\forall\; u\in {\mathcal A}{\mathcal D}(A),$$
induces an antisymmetric infinitesimal bialgebra structure on
${\mathcal A}{\mathcal D}(A)$.

In fact, for equation (2.3.5), we prove a little stronger conclusion
(for any $\mu\in {\mathcal A}{\mathcal D}(A)$) :
$$( id \otimes L(\mu)-R(\mu)\otimes  id )(r+\sigma(r))=\sum_i (e_i\otimes \mu*e_i^*+e_i\otimes \mu*e_i-e_i*\mu\otimes
e_i^*-e_i*\mu\otimes e_i)=0. \eqno (2.3.6)$$ If $\mu=e_j$, then
\begin{eqnarray*}
\sum_i e_i\otimes e_j*e_i^*&=&\sum_{m}e_m\cdot e_j\otimes
e_m^*+\sum_{i,m}\langle  e_i^*\circ e_m^*,e_j\rangle   e_i\otimes
e_m;\;
\sum_{i} e_i^*\otimes e_j*e_i=\sum_i e_i^*\otimes e_j\cdot e_i;\\
\sum_i e_i*e_j\otimes e_i^*&=&\sum_i e_i\cdot e_j\otimes e_i^*;\;
e_i^**e_j\otimes e_i=\sum_{i,m} \langle  e_j,e_m^*\circ e_i^*\rangle
e_m\otimes e_i+\sum_{m} e_m^*\otimes e_j\cdot e_m.
\end{eqnarray*}
Hence equation (2.3.6) holds for $\mu=e_j$ by exchanging some
indices. Similarly, equation (2.3.6) holds for $\mu=e_j^*$.
Therefore equation (2.3.5) holds.  Furthermore,
$$r_{12}r_{13}+r_{13}r_{23}-r_{23}r_{12}= \sum_{i,j}\{e_j\otimes
e_i*e_j^*\otimes e_i^*-e_j\cdot e_i\otimes e_j^*\otimes
e_i^*-e_i\otimes e_j\otimes e_i^*\circ e_j^*\}.$$ Since
$e_i*e_j^*=\sum\limits_m (\langle  e_j^*,e_m\cdot e_i\rangle
e_m^*+\langle e_j^*\circ e_m^*,e_i\rangle   e_m),$ we show that
$r_{12}r_{13}+r_{13}r_{23}-r_{23}r_{12}=0$. So ${\mathcal
A}{\mathcal D}(A)$ is an antisymmetric infinitesimal bialgebra.

For $e_i\in A$, we have
\begin{eqnarray*}
\Delta_{{\mathcal A}{\mathcal D}}(e_i)
&=&\sum_{m,k}\{  \langle  e_m^*,e_k\cdot e_i\rangle   e_m\otimes e_k^* + \langle  e_m^*\circ
e_k^*,e_i\rangle   e_m\otimes e_k -\langle  e_m^*.e_k\cdot e_i\rangle   e_m\otimes e_k^*\}\\
&=&\sum_{m,k}\langle  e_m^*\circ e_k^*,e_i\rangle   e_m\otimes
e_k=\Delta_A (e_i).
\end{eqnarray*}
Therefore the inclusion $i_1:A\rightarrow A\oplus A^*$ is a
homomorphism of antisymmetric infinitesimal bialgebras. Similarly,
the inclusion $i_2:A^*\rightarrow A\oplus A^*$ is also a
homomorphism of antisymmetric infinitesimal bialgebras since
$\Delta_{{\mathcal A}{\mathcal D}}(e_i^*)=-\beta_{A^*}(e_i^*)$,
where $\beta_{A^*}$ is given in Remark 2.2.4.\end{proof}

\begin{defn}
{\rm Let $(A,A^*)$ be an antisymmetric infinitesimal bialgebra. With
the antisymmetric infinitesimal bialgebra structure given in Theorem
2.3.6, $A\oplus A^*$ is called  an {\it associative double} of $A$.
We denote it by ${\mathcal A}{\mathcal D}(A)$.}
\end{defn}

\begin{remark}{\rm If we use the opposite algebra, then Theorem 2.3.6 and its proof
overlap [A3, Theorem 5.9 and Proposition 5.10] partly. Moreover, the
associative double ${\mathcal A}{\mathcal D}(A)$ is a balanced
Drinfeld double which was denoted by $D_b(A)$ in \cite{A3}. }
\end{remark}

\begin{coro}
Let $(A,A^*)$ be an antisymmetric infinitesimal bialgebra. Then the
associative double ${\mathcal A}{\mathcal D}(A)$ of $A$ is an
antisymmetric infinitesimal bialgebra and it is a symmetric
Frobenius algebra with the bilinear form given by equation {\rm
(1.1.1)}.\end{coro}

\subsection{Associative Yang-Baxter equation and its properties}

\begin{coro} Let $A$ be an associative algebra and $r\in A\otimes A$.
Suppose that $r$ is antisymmetric. Then the map $\Delta$ defined by
equation {\rm (2.3.1)} induces an associative algebra structure on
$A^*$ such that $(A,A^*)$ is an antisymmetric infinitesimal
bialgebra if
$$r_{12}r_{13}+r_{13}r_{23}-r_{23}r_{12}=0.\eqno (2.4.1)$$
\end{coro}

\begin{defn}{\rm Let $A$ be an associative algebra and $r\in A\otimes A$.
Equation (2.4.1) is called {\it associative Yang-Baxter equation in
$A$}.}\end{defn}

\begin{remark}{\rm In \cite{A1} and \cite{A3}, the associative Yang-Baxter
equation is given as
$$r_{13}r_{12}+r_{23}r_{13}-r_{12}r_{23}=0.\eqno (2.4.2)$$
Note that equation (2.4.1) is equation (2.4.2) in the opposite
algebra. Moreover, if $r$ satisfies $(L(x)\otimes id \otimes  id -
id \otimes  id \otimes R(x))(r_{12}+r_{21})=0$, then ([A3, Lemma
3.4])
$$\sigma_{13}
(r_{12}r_{13}+r_{13}r_{23}-r_{23}r_{12})=r_{13}r_{12}+r_{23}r_{13}-r_{12}r_{23},\eqno
(2.4.3)$$ where the linear map $\sigma_{13}:A\otimes A\otimes
A\rightarrow A\otimes A\otimes A$ is given by $\sigma_{13}(x\otimes
y\otimes z)=z\otimes y\otimes x$ for any $x,y,z\in A$. In
particular, when $r$ is antisymmetric, the above two associative
Yang-Baxter equations are equivalent.}\end{remark}

In order to be self-contained, in the following we give some
properties of associative Yang-Baxter equation from the point of
view of Frobenius algebras, although some of them have already been
given in \cite{A3}. Let $A$ be a vector space. For any $r\in
A\otimes A$, $r$ can be regarded as a map from $A^*$ to $A$ in the
following way:
$$\langle  u^*\otimes v^*,r\rangle   =\langle  u^*,r(v^*)\rangle   ,\;\;\forall\; u^*,v^*\in A^*.\eqno (2.4.4)$$

\begin{prop} Let $(A,\cdot)$ be an associative algebra and $r\in A\otimes
A$ be an antisymmetric solution of associative Yang-Baxter equation
in $A$. Then the associative algebra structure on the associative
double ${\mathcal A}{\mathcal D}(A)$ is given from the products in
$A$ as follows.

{\rm (a)} $a^**b^*=a^*\circ b^*=R^*_\cdot
(r(a^*))b^*+L_\cdot^*(r(b^*))a^*$, for any $a^*,b^*\in A^*$;\hfill
{\rm (2.4.5)}

{\rm (b)} $x*a^*=x\cdot r(a^*)-r(R_\cdot^*(x)a^*) + R_\cdot^*
(x)a^*$, for any $x\in A$, $a^*\in A^*$;\hfill {\rm (2.4.6)}

{\rm (c)} $a^**x=r(a^*)\cdot x-r(L_\cdot^* (x)
a^*)+L_\cdot^*(x)a^*$, for any $x\in A$, $a^*\in A^*$.\hfill {\rm
(2.4.7)}
\end{prop}

\begin{proof} Let $\{ e_1,\cdots,e_n\}$ be a basis of $A$ and $\{
e_1^*,\cdots, e_n^*\}$ be its dual basis. Suppose that $e_i\cdot
e_j=\sum_{k}c_{ij}^ke_k$ and $r=\sum_{i,j}a_{ij}e_i\otimes e_j$,
where $a_{ij}=-a_{ji}$. Then for any $i$, we have
$$\Delta(e_i)=\sum_{\alpha,\beta,l}a_{\alpha\beta}(c_{i\beta}^le_\alpha\otimes
e_l-c_{\alpha i}^l e_l\otimes
e_\beta)=\sum_{\alpha,\beta}\sum_l(a_{\alpha
l}c_{il}^\beta-a_{l\beta}c_{li}^\alpha)e_\alpha\otimes e_\beta.$$
Therefore we show that (for any $i,j$)
\begin{eqnarray*}
e_i^*\circ e_j^*&=&\sum_{l,t}(a_{il}c_{tl}^j-a_{lj}c_{lt}^i)e_t^*=
\sum_{l,t}( a_{il}\langle  e_t\cdot e_l,e_j^*\rangle   -a_{lj}\langle  e_l\cdot e_t,e_i^*\rangle   )e_t^*\\
&=& \sum_{t}(\langle  e_t\cdot r(e_i^*),e_j^*\rangle   +\langle
r(e_j^*)\cdot e_t,e_i^*\rangle   )e_t^*=R^*_\cdot
(r(e_i^*))e_j^*+L_\cdot^*(r(e_j^*))e_i^*.
\end{eqnarray*}
Similarly, equations (2.4.6) and (2.4.7) hold.
\end{proof}

\begin{theorem} {\rm ([A3, Proposition 2.1])}\quad Let $A$ be an associative
algebra and $r\in A\otimes A$. Suppose that $r$ is antisymmetric and
nondegenerate. Then $r$ is a solution of associative Yang-Baxter
equation in $A$ if and only if the inverse of the isomorphism
$A^*\rightarrow A$ induced by $r$, regarded as a bilinear form
$\omega$ on $A$ (that is, $\omega(x,y)=\langle  r^{-1}x,y\rangle$
for any $x,y\in A$), is a Connes cocycle.\end{theorem}

\begin{coro} Let $(A,\cdot)$ be an associative algebra and $r\in A\otimes A$
be a nondegenerate antisymmetric solution of associative Yang-Baxter
equation in $A$. Suppose the associative algebra structure
$``\circ"$ on $A^*$ is induced by $r$ from equation {\rm (2.4.5)}.
Then we have
$$a^*\circ b^*=r^{-1}(r(a^*)\cdot r(b^*)),\;\;\forall a^*,b^*\in A^*.\eqno (2.4.8)$$
Therefore $r:A^*\rightarrow A$ is an isomorphism of associative
algebras.
\end{coro}

\begin{proof} Set $\omega(x,y)=\langle
r^{-1}(x),y\rangle $ for any $x,y\in A$. Then $\omega$ is a Connes
cocycle of $A$. Hence
\begin{eqnarray*}
\langle  a^*\circ b^*,x\rangle   &=&\langle  r(b^*)\cdot x,
a^*\rangle   +\langle  x\cdot r(a^*),b^*\rangle
=\omega(r(a^*), r(b^*)\cdot x)+\omega(r(b^*),x\cdot r(a^*))\\
&=&-\omega(x,r(a^*)\cdot r(b^*))= \langle  r^{-1}(r(a^*)\cdot
r(b^*)),x\rangle,\;\;\forall\;  a^*,b^*\in A^*, x\in A.
\end{eqnarray*}
So equation (2.4.8) holds. Therefore $r$ is an isomorphism of
associative algebras. \end{proof}

Next we turn to the general antisymmetric solutions of associative
Yang-Baxter equation.

\begin{theorem} Let $(A,\cdot)$ be an associative algebra and $r\in A\otimes
A$ be antisymmetric. Then $r$ is a solution of associative
Yang-Baxter equation in $A$ if and only if $r$ satisfies
$$r(a^*)\cdot r(b^*)=r(R_\cdot^*(r(a^*))b^*+L_\cdot^*(r(b^*))a^*),\;\;\forall a^*,b^*\in A^*.\eqno (2.4.9)$$
\end{theorem}

\begin{proof} Let $\{ e_1,\cdots,e_n\}$ be a basis of $A$ and $\{
e_1^*,\cdots,e_n^*\}$ be its dual basis. Suppose that $e_i\cdot
e_j=\sum_k c_{ij}^k e_k$ and $r=\sum_{i,j}a_{ij}e_i\otimes e_j$,
$a_{ij}=-a_{ji}$. Hence $r(e_i^*)=\sum_k a_{ki}e_k$. Then $r$ is a
solution of  associative Yang-Baxter equation in $A$ if and only if
(for any $i,j,k$)
$$\sum_{m,l}\{ c_{kl}^ma_{ik}a_{jl}-c_{lk}^ia_{jl}a_{km}-c_{lk}^ja_{lm}a_{ik}\}=0.$$
The left-hand side of the above equation is just the coefficient of
$e_m$ in
$$r(e_i^*)\cdot r(e_j^*)-r(R_\cdot^*(r(e_i^*))e_j^*+L_\cdot^*(r(e_j^*))e_i^*).$$ Therefore the conclusion
follows.\end{proof}

Combining Proposition 2.4.4 and Theorem 2.4.7, we have the following
conclusion which extends Corollary 2.4.6.

\begin{coro} Let $(A,\cdot)$ be an associative algebra and $r\in A\otimes
A$ be an antisymmetric solution of associative Yang-Baxter equation
in $A$. Suppose the associative algebra structure $``\circ"$ on
$A^*$ is induced by $r$ from equation {\rm (2.4.5)}. Then we have
$$r(a^*\circ b^*)=r(a^*)\cdot r(b^*),\;\;\forall a^*,b^*\in A^*.\eqno (2.4.10)$$
Therefore $r:A^*\rightarrow A$ is an homomorphism of associative
algebras.
\end{coro}

Recall that two Frobenius algebras $(A_1,{\mathcal B}_1)$ and
$(A_2,{\mathcal B}_2)$ are isomorphic if and only if there exists an
isomorphism of associative algebras $\varphi: A_1\rightarrow A_2$
such that
$${\mathcal B}_1(x,y)=\varphi^*{\mathcal B}_2(x,y)
={\mathcal B}_2(\varphi(x),\varphi (y)),\;\;\forall x,y\in A_1.\eqno
(2.4.11)$$

\begin{theorem} Let $(A,\cdot)$ be an associative algebra. Then as
Frobenius algebras, the Frobenius algebra
$(A\bowtie^{R^*_\cdot,L^*_\cdot}_{R^*_\circ, L_\circ^*}
A^*,{\mathcal B})$ given by an antisymmetric solution $r$ of
associative Yang-Baxter equation in $A$ is isomorphic to the
Frobenius algebra $(A\ltimes_{R_\cdot^*,L_\cdot^*}A^*,{\mathcal
B})$, where ${\mathcal B}$ is given by equation {\rm (1.1.1)}.
However, in general, they are not isomorphic as the double
constructions of Frobenius algebras (or equivalently, as
antisymmetric infinitesimal bialgebras).
\end{theorem}

\begin{proof} Let $r$ be an antisymmetric solution of associative
Yang-Baxter equation in $A$. Define a linear map $\varphi:
A\ltimes_{R^*,L^*}A^*\rightarrow
A\bowtie^{R^*_\cdot,L^*_\cdot}_{R^*_\circ, L_\circ^*} A^*$
satisfying
$$\varphi(x)=x,\;\;\varphi(a^*)=-r(a^*)+a^*,\;\;\forall x\in A,
a^*\in A^*.$$ It is straightforward to show that $\varphi$ is an
isomorphism of associative algebras. Moreover,
$$\varphi^*{\mathcal
B}(x+a^*,y+b^*)=\langle  a^*,-r(b^*)+y\rangle   +\langle
x-r(a^*),b^*\rangle
 =\langle  a^*,y\rangle   +\langle  x,b^*\rangle={\mathcal B}(x+a^*,y+b^*)   .$$
Therefore $\varphi$ is an isomorphism of Frobenius algebras. However
in general, as antisymmetric infinitesimal bialgebras, they are not
isomorphic. In fact, if $\psi$ is an isomorphism of antisymmetric
infinitesimal bialgebras between $A\ltimes_{R^*,L^*}A^*$ and
$A\bowtie^{R^*_\cdot,L^*_\cdot}_{R^*_\circ, L_\circ^*} A^*$, then
for any $u^*,v^*\in A^*$, there exist $a^*,b^*\in A^*$ such that
$\psi(a^*)=u^*,\psi(b^*)=v^*$. However, $\psi(a^*\circ b^*)=0$ and
$\psi(a^*)*\psi(b^*)=u^**v^*=R^*(r(a^*))b^*+L^*(r(b^*))a^*$ is not
zero in general, which is an contradiction.\end{proof}

\begin{coro}  Let $(A,\cdot)$ be an associative algebra. Then as Frobenius
algebras, the Frobenius algebras
$(A\bowtie^{R^*_\cdot,L^*_\cdot}_{R^*_\circ, L_\circ^*}
A^*,{\mathcal B})$ given by all antisymmetric solutions of
associative Yang-Baxter equation in $A$ are isomorphic to the
Frobenius algebra $(A\ltimes_{R_\cdot^*,L_\cdot^*}A^*,{\mathcal B})$
given by the zero solution.\end{coro}



\subsection{Associative Yang-Baxter equation and ${\mathcal
O}$-operators}


\begin{defn} {\rm Let $(A,\cdot)$ be an associative algebra and $(l,r, V)$ be
a bimodule. A linear map $T:V\rightarrow A$ is called  an {\it
${\mathcal O}$-operator associated to $(l,r,V)$} if $T$ satisfies
$$T(u)\cdot T(v)=T(l(T(u))v+r(T(v)u)),\;\;\forall u,v\in V.\eqno
(2.5.1)$$}
\end{defn}

\begin{exam} {\rm Let $(A,\cdot)$ be an associative
algebra.
 Then the identity map $id$ is an ${\mathcal O}$-operator
associated to the bimodule $(L,0)$ or $(0,R)$.}\end{exam}

\begin{exam} {\rm Let $(A,\cdot)$ be an associative algebra. A linear map $R:
A\rightarrow A$ is called a Rota-Baxter operator on $A$ of weight
zero  (\cite{Bax}, \cite{Rot}) if $R$ satisfies
$$R(x)\cdot R(y)=R(R(x)\cdot y+x\cdot R(y)),\;\;\forall\; x,y\in A.\eqno (2.5.2)$$
In fact, a Rota-Baxter operator on $A$ is just an ${\mathcal
O}$-operator associated to the bimodule $(L,R)$.}\end{exam}

\begin{exam} {\rm Let $(A,\cdot)$ be an associative algebra and $r\in A\otimes
A$ be antisymmetric. Then $r$ is a solution of associative
Yang-Baxter equation in $A$ if and only if $r$ is an ${\mathcal
O}$-operator associated to the bimodule $(R^*,L^*)$.}\end{exam}

\begin{theorem} {\rm (\cite{BGN1})}\quad Let $(A,\cdot)$ be an associative
algebra and $(l,r,V)$ be a bimodule. Let $(r^*,l^*,V^*)$ be the
bimodule of $A$ given by {\rm Lemma 2.1.2}. Let $T:V\rightarrow A$
be a linear map which is identified as an element in
$(A\ltimes_{r^*,l^*}V^*)\otimes (A\ltimes_{r^*,l^*}V^*)$. Then
$r=T-\sigma(T)$ is an antisymmetric solution of the associative
Yang-Baxter equation in $A\ltimes_{r^*,l^*} V^*$ if and only if $T$
is an ${\mathcal O}$-operator associated to the bimodule $(l,r,V)$.
\end{theorem}

\begin{coro} {\rm (cf. Corollary 3.1.5)}\quad Let $(A,\cdot)$ be an associative
algebra. Then
$$r=\sum_{i}^n (e_i\otimes e_i^*-e_i^*\otimes e_i)\eqno (2.5.3)$$
is a solution of the associative Yang-Baxter equation in $A
\ltimes_{R^*,0} A^*$ or $A\ltimes_{0,L^*}A^*$, where $\{e_1,\cdots,
e_n\}$ is a basis of $A$ and $\{e_1^*,\cdots, e_n^*\}$ is its dual
basis. Moreover there is a natural Connes cocycle $\omega$  on
$A\ltimes_{R^*,0} {A}^*$ or $A\ltimes_{0,L^*}$ induced by $r^{-1}:
A\oplus A^*\rightarrow (A\oplus A^*)^*$, which is given by equation
{\rm (1.4.1).}
\end{coro}

\begin{proof} Note that ${id}$ is an ${\mathcal O}$-operator associated to the
bimodule $(L,0,A)$ or $(0,R,A)$. Then the conclusion follows from
Theorems 2.5.5 and 2.4.5.
\end{proof}

\section{Dendriform algebras}

\subsection{${\mathcal O}$-operators and dendriform algebras}

\mbox{}

There are close relations between ${\mathcal O}$-operators and a
class of algebras, namely, dendriform algebras, which are given in
\cite{BGN2}. In order to be self-contained, we list them in this
subsection.

\begin{defn} {\rm (\cite{Lo1})\quad Let $A$ be a vector space over a field {\bf F} with two
bilinear products denoted by $\prec$ and $\succ $. $(A, \prec,
\succ)$ is called  a {\it dendriform algebra} if for any $x,y,z\in
A$,
$$(x\prec y)\prec z=x\prec (y*z),\;\;(x\succ y)\prec z=x\succ (y\prec
z),\;\;x\succ (y\succ z)=(x*y)\succ z,\eqno (3.1.1)$$ where
$x*y=x\prec y+x\succ y$.}\end{defn}

Let $(A, \prec, \succ)$ be a dendriform algebra. For any $x\in A$,
let $L_\succ(x), R_\succ(x)$ and $L_\prec (x), R_\prec(x)$ denote
the left and right multiplication operators of $(A,\prec)$ and
$(A,\succ)$ respectively, that is, $$L_\succ (x)(y)=x\succ y,
\;\;R_\succ(x)y=y\succ x,\;\;L_\prec(x)y= x\prec y,\;\;
R_\prec(x)(y)=y\prec x,\;\;\forall\;x, y\in A.$$ Moreover, let
$L_\succ, R_\succ, L_\prec, R_\prec: A\rightarrow \frak g\frak l (A)$ be four
linear maps with $x\rightarrow L_\succ(x)$, $x\rightarrow
R_\succ(x), x\rightarrow L_\prec(x)$ and $x\rightarrow R_\prec (x)$
respectively. It is known that the product given by (\cite{Lo1})
$$x*y=x\prec y+x\succ y,\;\;\forall x,y\in A,\eqno (3.1.2)$$
defines an associative algebra. We call $(A,*)$ the associated
associative algebra of $(A,\succ,\prec)$ and $(A,\succ,\prec)$ is
called a compatible dendriform algebra structure on the
associative algebra $(A,*)$. Moreover, $(L_\succ, R_\prec)$ is a
bimodule of the associated associative algebra $(A,*)$.

\begin{theorem} {\rm ([BGN2])}\quad Let $A$ be an associative algebra and $(l,r,V)$ be
a bimodule. Let $T:V\rightarrow A$ be an ${\mathcal O}$-operator
associated to $(l,r,V)$. Then there exists a dendriform algebra
structure on $V$ given by
$$u\succ v=l(T(u))v,\;\;u\prec v=r(T(v))u,\;\;\forall\; u,v\in V.\eqno (3.1.3)$$
So there is an associated associative algebra structure on $V$ given
by equation {\rm (3.1.2)} and $T$ is a homomorphism of associative
algebras. Moreover, $T(V)=\{T(v)|v\in V\}\subset A$ is an
associative subalgebra of $A$ and there is an induced dendriform
algebra structure on $T(V)$ given by
$$T(u)\succ T(v)=T(u\succ v),\;\;T(u)\prec T(v)=T(u\prec v),\;\;\forall\; u,v\in V.\eqno (3.1.4)$$
Its corresponding associated associative algebra structure on $T(V)$
given by equation {\rm (3.1.2)} is just the associative subalgebra
structure of $A$ and $T$ is a homomorphism of dendriform
algebras.\end{theorem}

\begin{coro}
{\rm (\cite{BGN2})}\quad Let $(A,*)$ be an associative algebra. There
is a compatible dendriform algebra structure on $A$ if and only if
there exists an invertible $\mathcal O$-operator of $(A,
*)$.\end{coro}

In fact, if $T$ is an invertible $\mathcal O$-operator associated to
a bimodule $(l,r,V)$, then the compatible dendriform algebra
structure on $A$ is given by
$$x\succ y=T(l(x)T^{-1}(y)),\;\; x\prec y=T(r(y)T^{-1}(x)),\;\;\forall x,y\in A.\eqno (3.1.5)$$
Conversely, let $(A,\succ, \prec)$ be a dendriform algebra and
$(A,*)$ be the associated associative algebra. Then the identity map
$id$ is an $\mathcal O$-operator associated to the bimodule
$(L_\succ, R_\prec)$ of $(A,*)$.

\begin{remark}{\rm If $T$ is an invertible $\mathcal O$-operator associated to
a bimodule $(l,r,V)$, then the linear map $f=T^{-1}:A\rightarrow V$
satisfies
$$f(x*y)=l(x)f(y)+r(y)f(x),\;\;\forall x,y\in A.\eqno (3.1.6)$$
Such a linear map is a 1-cocycle of $(A,*)$ associated to the
bimodule $(l,r,V)$.}
\end{remark}

\begin{coro}
{\rm ([BGN2])}\quad Let $(A,\succ,\prec)$ be a dendriform algebra.
Then
$$r=\sum_{i}^n (e_i\otimes e_i^*-e_i^*\otimes e_i)\eqno (3.1.7)$$
is a solution of the associative Yang-Baxter equation in $A
\ltimes_{R_\prec^*,L_\succ^*} A^*$, where $\{e_1,\cdots, e_n\}$ is a
basis of $A$ and $\{e_1^*,\cdots, e_n^*\}$ is its dual basis.
Moreover there is a natural Connes cocycle $\omega$  on $A
\ltimes_{R_\prec^*,L_\succ^*} A^*$ induced by $r^{-1}: A\oplus
A^*\rightarrow (A\oplus A^*)^*$, which is given by equation {\rm
(1.4.1)}.\end{coro}

\begin{remark}  {\rm It is easy to see that Corollary 2.5.6 is just a special case
of the above conclusion, that is, the former corresponds to the
trivial dendriform algebra structure on an associative algebra
$(A,\cdot)$ given by $\succ=\cdot, \prec =0$ or $\succ=0,
\prec=\cdot$.}\end{remark}

\subsection{Bimodules and matched pairs of dendriform algebras}

\mbox{}

\begin{defn} {\rm (\cite{A4})\quad Let $(A,\succ, \prec)$ be a dendriform
algebra and $V$ be a vector space. Let $l_\succ, r_\succ,
l_\prec,r_\prec:A\rightarrow \frak g\frak l (V)$ be four linear maps. $V$ (or
$(l_\succ,r_\succ,l_\prec,r_\prec)$, or
$(l_\succ,r_\succ,l_\prec,r_\prec,V)$) is called  a {\it bimodule}
of $A$ if the following equations hold (for any $x,y\in A$).
$$l_\prec(x\prec
y)=l_\prec(x)l_*(y);\;\;r_\prec(x)l_\prec(y)=l_\prec(y)r_*(x);\;\;
r_\prec(x)r_\prec(y)=r_\prec(y*x);\eqno (3.2.1)$$
$$l_\prec(x\succ y)=l_\succ(x)l_\prec(y);\;\;r_\prec(x)l_\succ (y)=l_\succ(y)r_\prec(x);\;\;
r_\prec(x)r_\succ(y)=r_\succ(y\prec x);\eqno (3.2.2)$$
$$l_\succ (x*y)=l_\succ(x)l_\succ (y);\;\;r_\succ(x)l_*(y)=l_\succ(y)r_\succ(x);\;\;
r_\succ (x)r_*(y)=r_\succ(y\succ x),\eqno (3.2.3)$$ where
$x*y=x\succ y+x\prec y$, $l_*=l_\succ+l_\prec$,
$r_*=r_\succ+r_\prec$. }\end{defn}

By a direct computation or according to \cite{Sc}, $(l_\succ,
r_\succ, l_\prec, r_\prec, V)$ is a bimodule of a dendriform algebra
$(A,\succ, \prec)$ if and only if there exists a dendriform algebra
structure on  the direct sum $A\oplus V$ of the underlying vector
spaces of $A$ and $V$ given by ($\forall x,y\in A, u,v\in V$)
$$(x+u)\succ (y+v)=x\succ y+l_\succ(x)v+r_\succ (y)u,\;\;
(x+u)\prec (y+v)=x\prec y+l_\prec(x)v+r_\prec(y)u.\eqno (3.2.4)$$ We
denote it by $A\ltimes_{l_\succ,r_\succ, l_\prec,r_\prec}V$.

\begin{prop} Let $(l_\succ,r_\succ,l_\prec,r_\prec,V)$ be a bimodule of a
dendriform algebra $(A,\succ,\prec)$. Let $(A,*)$ be the associated
associative algebra. Then we have the following results.

{\rm (1)} Both $(l_\succ, r_\prec, V)$ and $(l_\succ+l_\prec,
r_\succ+r_\prec, V)$ are bimodules of $(A,*)$.

{\rm (2)} For any bimodule $(l,r, V)$ of $(A,*)$, $(l,0,0,r, V)$ is
a bimodule of $(A,\succ,\prec)$.

{\rm (3)} Both $(l_\succ+l_\prec, 0,0,r_\succ+r_\prec, V)$ and
$(l_\succ, 0,0,r_\prec, V)$ are bimodules of $(A,\succ,\prec)$.

{\rm (4)} The dendriform algebras
$A\ltimes_{l_\succ,r_\succ,l_\prec,r_\prec} V$ and
$A\ltimes_{l_\succ+l_\prec, 0,0,r_\succ+r_\prec}V$ have the same
associated associative algebra
$A\ltimes_{l_\succ+l_\prec,r_\succ+r_\prec} V$.

{\rm (5)} Let $l_\succ^*,r_\succ^*,l_\prec^*,r_\prec^*:A\rightarrow
\frak g\frak l (V^*)$ be the linear maps given by
$$\langle l_\succ^*(x)a^*,y\rangle=\langle l_\succ(x)y,
a^*\rangle,\;\;\langle r_\succ^*(x)a^*,y\rangle=\langle r_\succ(x)y,
a^*\rangle,\eqno (3.2.5)$$
$$\langle l_\prec^*(x)a^*,y\rangle=\langle l_\prec(x)y,
a^*\rangle,\;\;\langle r_\prec^*(x)a^*,y\rangle=\langle r_\prec(x)y,
a^*\rangle.\eqno (3.2.6)$$ Then $(r_\succ^*+r_\prec^*,
-l_\prec^*,-r_\succ^*,l_\succ^*+l_\prec^*, V^*)$ is a bimodule of
$(A,\succ,\prec)$.

{\rm (6)} Both $(r_\succ^*+r_\prec^*, 0,0,l_\succ^*+l_\prec^*, V^*)$
and $(r_\prec^*, 0,0,l_\succ^*, V^*)$ are bimodules of
$(A,\succ,\prec)$.

{\rm (7)} Both $(r_\succ^*+r_\prec^*, l_\succ^*+l_\prec^*,V^*)$ and
$(r_\prec^*, l_\succ^*, V^*)$ are bimodules of $(A,*)$.

{\rm (8)} The dendriform algebras
$A\ltimes_{r_\succ^*+r_\prec^*,-l_\prec^*,
-r_\succ^*,l_\succ^*+l_\prec^*} V^*$ and $A\ltimes_{r_\prec^*,
0,0,l_\succ^*}V^*$ have the same associative algebra
$A\ltimes_{r_\prec^*,l_\succ^*}V^*$.
\end{prop}

\begin{proof}
It is straightforward.\end{proof}

\begin{exam}
{\rm Let $(A,\succ,\prec)$ be a dendriform algebra. Then
$$(L_\succ,R_\succ, L_\prec,R_\prec, A),\;\;\;(L_\succ,0,0,R_\prec, A)\;\;
{\rm and}\;\;(L_\succ+L_\prec,0,0,R_\succ+R_\prec, A)$$ are
bimodules of $(A,\prec, \succ)$. On the other hand,
$$(R_\succ^*+R_\prec^*,-L_\prec^*, -R_\succ^*,L_\succ^*+L_\prec^*,
A^*),\;\;(R_\prec^*,0,0,L_\succ^*, A^*) \;\;{\rm
and}\;\;(R_\succ^*+R_\prec^*,0,0,L_\succ^*+L_\prec^*,A^*)$$ are
bimodules of $(A,\succ, \prec)$, too. There are two compatible
dendriform algebra structures
$$A\ltimes_{R_\succ^*+R_\prec^*,-L_\prec^*,
-R_\succ^*,L_\succ^*+L_\prec^*} A^*\;\; {\rm  and}\;\;
A\ltimes_{R_\prec^*, 0,0,L_\succ^*}A^*$$ on the same associative
algebra $A\ltimes_{R_\prec^*,L_\succ^*}A^*$.}\end{exam}

\begin{theorem}
Let $(A,\succ_A,{\prec_A})$ and $(B,{\succ_B},{\prec_B})$ be two
dendriform algebras. Suppose that there are linear maps
$l_{\succ_A},r_{\succ_A},l_{\prec_A},r_{\prec_A}:A\rightarrow \frak g\frak l (B)$
and $l_{\succ_B},r_{\succ_B},l_{\prec_B},r_{\prec_B}:B\rightarrow
\frak g\frak l (A)$ such that $(l_{\succ_A},r_{\succ_A},l_{\prec_A},r_{\prec_A})$
is a bimodule of $A$ and
$(l_{\succ_B},r_{\succ_B},l_{\prec_B},r_{\prec_B})$ is a bimodule of
$B$ and they satisfy the following 18 equations:
$$r_{\prec_A}(x)(a\prec_Bb)=a\prec_B(r_A(x)b)+r_{\prec_A}(l_B(b)x)a;\eqno (3.2.7)$$
$$l_{\prec_A}(l_{\prec_B}(a)x)b+(r_{\prec_A}(x)a)\prec_B b=a\prec_B
(l_A(x)b)+r_{\prec_A}(r_B(b)x)a;\eqno (3.2.8)$$
$$l_{\prec_A}(x)(a*_Bb)=(l_{\prec_A}(x)a)\prec_B
b+l_{\prec_A}(r_{\prec_B}(a)x)b;\eqno (3.2.9)$$
$$r_{\prec_A}(x)(a\succ_Bb)=r_{\succ_A}(l_{\prec_B}(b)x)a+a\succ_B
(r_{\prec_A}(x)b);\eqno (3.2.10)$$
$$l_{\prec_A}(l_{\succ_B}(a)x)b+(r_{\succ_A}(x)a)\prec_B b=a\succ_B
(l_{\prec_A}(x)b)+r_{\succ_A}(r_{\prec_B}(b)x)a;\eqno (3.2.11)$$
$$l_{\succ_A}(x)(a\prec_Bb)=(l_{\succ_A}(x)a)\prec_B
b+l_{\prec_A}(r_{\succ_B}(a)x)b;\eqno (3.2.12)$$
$$r_{\succ_A}(x)(a*_Bb)=a\succ_B(r_{\succ_A}(x)b)+r_{\succ_A}(l_{\succ_B}(b)x)a;\eqno
(3.2.13)$$
$$a\succ_B
(l_{\succ_A}(x)b)+r_{\succ_A}(r_{\succ_B}(b)x)a=l_{\succ_A}(l_B(a)x)b+(r_A(x)a)\succ_Bb;\eqno
(3.2.14)$$
$$l_{\succ_A}(x)(a\succ_B b)=(l_A(x)a)\succ_B
b+l_{\succ_A}(r_B(a)x)b;\eqno (3.2.15)$$
$$r_{\prec_B}(a)(x\prec_Ay)=x\prec_A(r_B(a)y)+r_{\prec_B}(l_A(y)a)x;\eqno (3.2.16)$$
$$l_{\prec_B}(l_{\prec_A}(x)a)y+(r_{\prec_B}(a)x)\prec_A y=x\prec_A
(l_B(a)y)+r_{\prec_B}(r_A(y)a)x;\eqno (3.2.17)$$
$$l_{\prec_B}(a)(x*_Ay)=(l_{\prec_B}(a)x)\prec_A
y+l_{\prec_B}(r_{\prec_A}(x)a)y;\eqno (3.2.18)$$
$$r_{\prec_B}(a)(x\succ_Ay)=r_{\succ_B}(l_{\prec_A}(y)a)x+x\succ_A
(r_{\prec_B}(a)y);\eqno (3.2.19)$$
$$l_{\prec_B}(l_{\succ_A}(x)a)y+(r_{\succ_B}(a)x)\prec_A y=x\succ_A
(l_{\prec_B}(a)y)+r_{\succ_B}(r_{\prec_A}(y)a)x;\eqno (3.2.20)$$
$$l_{\succ_B}(a)(x\prec_Ay)=(l_{\succ_B}(a)x)\prec_A
y+l_{\prec_B}(r_{\succ_A}(x)a)y;\eqno (3.2.21)$$
$$r_{\succ_B}(a)(x*_Ay)=x\succ_A(r_{\succ_B}(a)y)+r_{\succ_B}(l_{\succ_A}(y)a)x;\eqno
(3.2.22)$$
$$x\succ_A
(l_{\succ_B}(a)y)+r_{\succ_B}(r_{\succ_A}(y)a)x=l_{\succ_B}(l_A(x)a)y+(r_B(a)x)\succ_Ay;\eqno
(3.2.23)$$
$$l_{\succ_B}(a)(x\succ_A y)=(l_B(a)x)\succ_A
y+l_{\succ_B}(r_A(x)a)y,\eqno (3.2.24)$$ for any $x,y\in A,a,b\in B$
and $l_A=l_{\succ_A}+l_{\prec_A},r_A=r_{\succ_A}+r_{\prec_A}$,
$l_B=l_{\succ_B}+l_{\prec_B},r_B=r_{\succ_B}+r_{\prec_B}$. Then
there is a dendriform algebra structure on the direct sum $A\oplus
B$ of the underlying vector spaces of $A$ and $B$ given by
$$(x+a)\succ(y+b)=(x{\succ_A} y+r_{\succ_B}(b)x+l_{\succ_B}(a)y)+(l_{\succ_A}
(x)b+r_{\succ_A}(y)a+a{\succ_B} b),\eqno (3.2.25)$$
$$(x+a)\prec(y+b)=(x{\prec_A} y+r_{\prec_B}(b)x+l_{\prec_B}(a)y)+(l_{\prec_A}
(x)b+r_{\prec_A}(y)a+a{\prec_B} b),\eqno (3.2.26)$$ for any $x,y\in
A,a,b\in B$. We denote this dendriform algebra by
$A\bowtie^{l_{\succ_A},r_{\succ_A},l_{\prec_A},r_{\prec_A}}_{l_{\succ_B},r_{\succ_B},l_{\prec_B},r_{\prec_B}}B$
or simply $A\bowtie B$.  On the other hand, every dendriform algebra
which is the direct sum of the underlying vector spaces of two
subalgebras can be obtained from the above way.
\end{theorem}

\begin{proof} It is straightforward.\end{proof}

\begin{defn}
{\rm Let $(A,\succ_A,{\prec_A})$ and $(B,{\succ_B},{\prec_B})$ be
two dendriform algebras. Suppose that there are linear maps
$l_{\succ_A},r_{\succ_A},l_{\prec_A},r_{\prec_A}:A\rightarrow \frak g\frak l (B)$
and $l_{\succ_B},r_{\succ_B},l_{\prec_B},r_{\prec_B}:B\rightarrow
\frak g\frak l (A)$ such that $(l_{\succ_A},r_{\succ_A},l_{\prec_A},r_{\prec_A})$
is a bimodule of $A$ and
$(l_{\succ_B},r_{\succ_B},l_{\prec_B},r_{\prec_B})$ is a bimodule of
$B$. If equations (3.2.7)-(3.2.24) are satisfied, then
$(A,B,l_{\succ_A},r_{\succ_A},l_{\prec_A},r_{\prec_A},l_{\succ_B},r_{\succ_B},l_{\prec_B},r_{\prec_B})$
is called  a {\it matched pair of dendriform algebras}.}\end{defn}

\begin{remark}
{\rm  Obviously $B$ is an ideal of $A\bowtie B$ if and only if
$l_{\succ_B}=r_{\succ_B}=l_{\prec_B}=r_{\prec_B}=0$. If $B$ is a
trivial ideal, then
$A\bowtie^{l_{\succ_A},r_{\succ_A},l_{\prec_A},r_{\prec_A}}_{0,0,0,0}
B\cong A\ltimes_{l_{\succ_A},r_{\succ_A},l_{\prec_A},r_{\prec_A}}
B.$}\end{remark}

\begin{coro} Let
$(A,B,l_{\succ_A},r_{\succ_A},l_{\prec_A},r_{\prec_A},l_{\succ_B},r_{\succ_B},l_{\prec_B},r_{\prec_B})$
be a matched pair of dendriform algebras. Then $(A,B,
l_{\succ_A}+l_{\prec_A},
r_{\succ_A}+r_{\prec_A},l_{\succ_B}+l_{\prec_B},
r_{\succ_B}+r_{\prec_B})$ is a matched pair of the associated
associative algebras $(A,*_A)$ and $(B,*_B)$.\end{coro}

\begin{proof} In fact, the associated associative algebra $(A\bowtie B, *)$ is exactly the
associative algebra obtained from the matched pair $(A, B,l_A,r_A,
l_B,r_B)$ of associative algebras:
$$(x+a)*(y+b)=x*_Ay+l_B(a)y+r_B(b)x+a*_Bb+l_A(x)b+r_A(y)a,\;\;\forall
x,y\in A,a,b\in B,$$ where $l_A=l_{\succ_A}+l_{\prec_A},
r_A=r_{\succ_A}+r_{\prec_A},l_B=l_{\succ_B}+l_{\prec_B},
r_B=r_{\succ_B}+r_{\prec_B}$.\end{proof}

\section{Double constructions of Connes cocycles and
an analogue of the classical Yang-Baxter equation}

\subsection{Connes cocycles and dendriform algebras}

\begin{theorem}  Let $(A,*)$ be an associative algebra and $\omega$ be a
nondegenerate Connes cocycle. Then there exists a compatible
dendriform algebra structure $\succ,\prec$ on $A$ given by
$$\omega(x\succ y,z)=\omega(y, z*x),\;\; \omega(x\prec
y,z)=\omega(x,y*z),\;\; \forall x,y,z\in A.\eqno (4.1.1)$$
\end{theorem}

\begin{proof} Define a linear map $T:A\rightarrow A^*$ by
$\langle T(x),y\rangle=\omega (x,y),\;\;\forall x,y\in A$. Then $T$
is invertible and $T^{-1}$ is an $\mathcal O$-operator of the
associative algebra $(A,*)$ associated to the bimodule
$(R^*_*,L_*^*)$. By Corollary 3.1.3, there is a compatible
dendriform algebra structure $\succ,\prec$ on $(A,*)$ given by
$$x\succ y=T^{-1}R_*^*(x)T(y),\;\;x\prec y=T^{-1}L_*^*(y)T(x),\;\;\forall
x,y\in A,$$ which gives exactly equation (4.1.1).
\end{proof}

Next, we turn to the double construction of Connes cocycles. Let
$(A,*_A)$ be an associative algebra and suppose that there is a
associative algebra structure $*_{A^*}$ on its dual space $A^*$. We
construct an associative algebra structure on the direct sum
$A\oplus A^*$ of the underlying vector spaces of $A$ and $A^*$ such
that both $A$ and $A^*$ are subalgebras and the antisymmetric
bilinear form on $A\oplus A^*$ given by equation (1.4.1) is a
Connes cocycle on $A\oplus A^*$. Such a construction is called {a
double construction of Connes cocycle associated to $(A,*_A)$ and
$(A^*,*_{A^*})$} and we denote it by $(T(A)=A\bowtie A^*, \omega)$.

\begin{coro}Let $(T(A)=A\bowtie A^*,\omega)$ be a double construction of
Connes cocycle. Then there exists a compatible dendriform algebra
structure $\succ,\prec$ on $T(A)$ defined by equation {\rm (4.1.1)}.
Moreover, $A$ and $A^*$ are dendriform subalgebras with this
product.\end{coro}

\begin{proof} The first half follows from Theorem 4.1.1. Let
$x,y\in A$. Set $x\succ y=a+b^*$, where $a\in A$, $b^*\in A^*$.
Since $A$ is an associative subalgebra of $T(A)$ and $\omega
(A,A)=\omega (A^*,A^*)=0$, we have
$$\omega (b^*,A^*)=\omega (b^*, A)=\omega (x\succ y, A)=\omega (y,
A*x)=0.$$ Therefore $b^*=0$ due to the nondependence of $\omega$.
Hence $x\succ y=a\in A$. Similarly, $x\prec y\in A$. Thus $A$ is a
dendriform subalgebra of $T(A)$ with the product $\succ,\prec$. By
symmetry of $A$ and $A^*$, $A^*$ is also a dendriform
subalgebra.\end{proof}

\begin{defn}{\rm Let $(T(A_1)=A_1\bowtie
A_1^*,\omega_1)$ and $(T(A_2)=A_2\bowtie A_2^*,\omega_2)$ be two
double constructions of Connes cocycles. They are {\it isomorphic}
if there exists an isomorphism of associative algebras
$\varphi:T(A_1) \rightarrow T(A_2)$ satisfying the following
conditions:
$$\varphi(A_1)=A_2,\;\;\varphi(A_1^*)=A_2^*,\;\;\omega_1(x,y)=\varphi^*\omega_2(x,y)=
\omega_2(\varphi(x),\varphi(y)),\forall\; x,y\in A_1.
\eqno(4.1.3)$$}\end{defn}

\begin{prop} Two double constructions of Connes cocycles
$(T(A_1)=A_1\bowtie A_1^*,\omega_1)$ and $(T(A_2)=A_2\bowtie
A_2^*,\omega_2)$ are isomorphic if and only if there exists a
dendriform algebra isomorphism $\varphi:T(A_1)\rightarrow T(A_2)$
satisfying equation {\rm (4.1.3)}, where the dendriform algebra
structures on $T(A_1)$ and $T(A_2)$ are given by equation {\rm
(4.1.1)} respectively.\end{prop}

\begin{proof} It is straightforward.\end{proof}

\begin{theorem} Let $(A,\succ_A,\prec_A)$ be a dendriform algebra and
$(A,*_A)$ be the associated associative algebra. Suppose that
there is a dendriform algebra structure
$``\succ_{A^*},\prec_{A^*}"$ on its dual space $A^*$ and
$(A^*,*_{A^*})$ is the associated associative algebra. Then there
exists a double construction of Connes cocycle associated to
$(A,*_A)$ and $(A,*_{A^*})$ if and only if
$(A,A^*,R^*_{\prec_A},L^*_{\succ_A},R^*_{\prec_{A^*}},L^*_{\succ_{A^*}})$
is a matched pair of the associative algebras. Moreover, every
double construction of Connes cocycle can be obtained from the
above way.\end{theorem}

\begin{proof} The conclusion can be obtained by a similar proof as
of Theorem 2.2.1.
\end{proof}

\begin{coro} Let $(A,\succ,\prec)$ be a dendriform algebra and
$(R^*_\prec, L^*_\succ)$ be the bimodule of the associated
associative algebra $(A,*)$. Then
$(T(A)=A\ltimes_{R^*_\prec,L_\succ^*}A^*,\omega)$ is a double
construction of Connes cocycle. Conversely, let $(T(A)=A\bowtie
A^*,\omega)$ be a double construction of Connes cocycle. If $A^*$ is
an ideal of $T(A)$, then $A^*$ is a trivial associative algebra and
hence $T(A)$ is isomorphic to the semidirect $A\ltimes_{L_{T(A)},
R_{T(A)}}A^*$. Furthermore, this double construction of Connes
cocycle is isomorphic to the double construction of Connes cocycle
$(T(A)=A\ltimes_{R^*_\prec,L_\succ^*}A^*,\omega)$ which the
dendriform algebra structure on $A$ is given by $\omega$ from
equation {\rm (4.1.1)}.\end{coro}

\begin{proof} By Remark 2.1.6, $(A,A^*,R^*_\prec,L_\succ^*,0,0)$ with the
associative algebra structure on $A^*$ being trivial is always a
matched pair of associative algebras, the first half follows
immediately. Conversely, if $A^*$ is an ideal, then for any
$a^*,b^*\in A^*$, we have
$$T(A)*a^*,\;\; b^**T(A)\in A^*\Longrightarrow
\omega(a^**b^*,T(A))=-\omega(T(A)*a^*,b^*)-\omega(b^**T(A),
a^*)=0.$$ Thus $a^**b^*=0$. Then $T(A)$ is isomorphic to
$A\ltimes_{L_{T(A)}, R_{T(A)}}A^*$. By Remark 2.1.6 again, we show
that $(T(A)=A\bowtie A^*,\omega)$ is isomorphic to the double
construction of Connes cocycle
$(T(A)=A\ltimes_{R^*_\prec,L_\succ^*}A^*,\omega)$. \end{proof}

\begin{theorem} Let $(A,\succ_A,\prec_A)$ be a dendriform algebra and
$(A,*_A)$ be the associated associative algebra. Suppose that there
is a dendriform algebra structure $``\succ_{A^*},\prec_{A^*}"$ on
its dual space $A^*$ and $(A^*,*_{A^*})$ is the associated
associative algebra. Then $(A,A^*,R^*_{\prec_A},L^*_{\succ_A}$,
$R^*_{\prec_{A^*}},L^*_{\succ_{A^*}})$ is a matched pair of
associative algebras  if and only if
$$(A,A^*,
R_{\succ_A}^*+R_{\prec_A}^*,-L_{\prec_A}^*,-R_{\succ_A}^*,L_{\succ_A}^*+L_{\prec_A}^*,
R_{\succ_{A^*}}^*+R_{\prec_{A^*}}^*,-L_{\prec_{A^*}}^*,-R_{\succ_{A^*}}^*,L_{\succ_{A^*}}^*+L_{\prec_{A^*}}^*)$$
is a matched pair of dendriform algebras.
\end{theorem}

\begin{proof} The ``if" part follows from Corollary 3.2.7. We need to
prove the ``only if" part. If
$(A,A^*,R^*_{\prec_A},L^*_{\succ_A},R^*_{\prec_{A^*}},L^*_{\succ_{A^*}})$
is a matched pair of associative algebras, then
$(A\bowtie^{R^*_{\prec_A},L^*_{\succ_A}}_{R^*_{\prec_{A^*}},L^*_{\succ_{A^*}}}A^*,\omega)$
is a double construction of Connes cocycle. Hence there exists a
compatible dendriform algebra structure on
$A\bowtie^{R^*_{\prec_A},L^*_{\succ_A}}_{R^*_{\prec_{A^*}},L^*_{\succ_{A^*}}}A^*$
given by equation (4.1.1). By a simple and direct computation, we
show that $A$ and $A^*$ are its subalgebras and the other products
are given by
$$x\succ a^*=(R_{\succ_A}^*+R_{\prec_A}^*)(x)
a^*-L_{\prec_{A^*}}^*(a^*)x,\;\;x\prec
a^*=-R_{\succ_A}^*(x)a^*+(L_{\succ_{A^*}}^*+L_{\prec_{A^*}}^*)(a^*)x,\;\;$$
$$a^*\succ
x=(R_{\succ_{A^*}}^*+R_{\prec_{A^*}}^*)(a^*)x-L_{\prec_A}^*(x)a^*,\;\;
a^*\prec
x=-R_{\succ_{A^*}}^*(a^*)x+(L_{\succ_A}^*+L_{\prec_A}^*)(x)a^*,$$
for any $x\in A,a^* \in A^*$.  Therefore
$$(A,A^*,
R_{\succ_A}^*+R_{\prec_A}^*,-L_{\prec_A}^*,-R_{\succ_A}^*,L_{\succ_A}^*+L_{\prec_A}^*,
R_{\succ_{A^*}}^*+R_{\prec_{A^*}}^*,-L_{\prec_{A^*}}^*,-R_{\succ_{A^*}}^*,L_{\succ_{A^*}}^*+L_{\prec_{A^*}}^*)$$
is a matched pair of dendriform algebras. \end{proof}

\subsection{Dendriform D-bialgebras}

\begin{theorem} Let $(A,\succ_A,\prec_A)$ be a dendriform algebra whose
products are given by two linear maps $\beta_\succ^*, \beta_\prec^*
:A\otimes A\rightarrow A$. Suppose that there is a dendriform
algebra structure $``\succ_{A^*}$,$\prec_{A^*}"$ on its dual space
$A^*$ given by two linear maps $\Delta_\succ^*,\Delta_\prec^*:
A^*\otimes A^*\rightarrow A^*$. Then
$(A,A^*,R^*_{\prec_A},L^*_{\succ_A},R^*_{\prec_{A^*}},L^*_{\succ_{A^*}})$
is a matched pair of associative algebras if and only if the
following equations hold (for any $x,y\in A$ and $a^*,b^*\in A^*$):

{\rm (1)} $\Delta_\prec(x*_Ay)=( id\otimes
L_{\prec_A}(x))\Delta_\prec(y)+(R_A(y) \otimes id  ) \Delta_\prec
(x)$;\hfill {\rm (4.2.1)}

{\rm (2)} $\Delta_\succ (x*_Ay)=( id\otimes
L_A(x))\Delta_\succ(y)+(R_{\prec_A}(y) \otimes id
)\Delta_\succ(x)$;\hfill {\rm (4.2.2)}

{\rm (3)} $\beta_\prec(a^**_{A^*}b^*)=( id\otimes
L_{\prec_{A^*}}(a^*))\beta_\prec(b^*)+(R_{A^*}(b^*) \otimes id  )
\beta_\prec (a^*)$;\hfill {\rm (4.2.3)}

{\rm (4)} $\beta_\succ (a^**_{A^*}b^*)=( id\otimes
L_{A^*}(a^*))\beta_\succ(b^*)+(R_{\prec_{A^*}}(b^*) \otimes id
)\beta_\succ(a^*)$;\hfill {\rm (4.2.4)}

{\rm (5)} $(L_A(x) \otimes id  - id\otimes  R_{\prec_A}
(x))\Delta_\prec(y)+\sigma [(L_{\succ_A}(y)\otimes - id\otimes
R_A(y))\Delta_\succ(x)]=0$;\hfill {\rm (4.2.5)}

{\rm (6)} $(L_{A^*}(a^*) \otimes id  - id\otimes  R_{\prec_{A^*}}
(a^*))\beta_\prec(b^*)+\sigma [(L_{\succ_{A^*}}(b^*)\otimes -
id\otimes R_{A^*}(b^*))\beta_\succ(a^*)]=0$,\hfill {\rm (4.2.6)}

\noindent where
$L_A=L_{\succ_A}+L_{\prec_A},\;\;R_A=R_{\succ_A}+R_{\prec_A},
L_{A^*}=L_{\succ_{A^*}}+L_{\prec_{A^*}},\;\;R_{A^*}=R_{\succ_{A^*}}+R_{\prec_{A^*}}$.\end{theorem}

\begin{proof}Let $\{e_1,\cdots,e_n\}$ be a basis of $A$ and $\{
e_1^*,\cdots, e_n^*\}$ be its dual basis. Set
$$e_i\succ_A e_j=\sum\limits_{k=1}^na_{ij}^ke_k,\;\;
e_i\prec_A e_j=\sum\limits_{k=1}^nb_{ij}^ke_k,\;\; e_i^*\succ_{A^*}
e_j^*=\sum\limits_{k=1}^nc_{ij}^k e_k^*,\;\; e_i^*\prec_{A^*}
e_j^*=\sum\limits_{k=1}^nd_{ij}^k e_k^*.$$ Therefore the coefficient
of $e_l^*$ in
$$R_{\prec_A}^*(e_i)(e_j^**_{A^*}e_k^*)=R_{\prec_A}^*(L_{\succ_{A^*}}^*(e_j^*)e_i)e_k^*+R_{\prec_A}^*(e_i)e_j^**_{A^*}e_k^*
$$
gives the following relation (for any $i,j,k,l$)
$$\sum_{m=1}^n b_{li}^m(c_{jk}^m+d_{jk}^m)=\sum_{m=1}^n
[c_{jm}^ib_{lm}^k+b_{mi}^j(c_{mk}^l+d_{mk}^l)]$$ which is precisely
the relation given by the coefficient of $e_l^*\otimes e_i^*$ in
$$\beta_\prec(e_j^**_{A^*}e_k^*)=(R_{A^*}(e_k^*) \otimes id  )\beta_\prec(e_j^*)+( id\otimes
L_{\succ_{A^*}}(e_j^*))\beta_\prec(e_k^*).$$ So equation (2.1.4) in
the case $l_A=R^*_{\prec_A}, r_A=L^*_{\succ_A}$,
$l_B=l_{A^*}=R^*_{\prec_{A^*}},r_B=r_{A^*}=L^*_{\succ_{A^*}}$ is
equation (4.2.3). Similarly, in this situation, we have the
following correspondences:

equation (2.1.5)$\Longleftrightarrow$ equation (4.2.4);\quad
equation (2.1.6)$\Longleftrightarrow$ equation (4.2.1);

equation (2.1.7)$\Longleftrightarrow$ equation (4.2.2);\quad
equation (2.1.8)$\Longleftrightarrow$ equation (4.2.6);

equation (2.1.9)$\Longleftrightarrow$ equation (4.2.5).

\noindent Therefore the conclusion holds due to Theorem
2.1.4.\end{proof}

\begin{defn}{\rm  Let $A$ be a vector space.  A {\it dendriform D-bialgebra
}structure on $A$ is a set of linear maps
$(\Delta_\prec,\Delta_\succ,\beta_\prec,\beta_\succ)$ such that
$\Delta_\prec,\Delta_\succ:A\rightarrow A\otimes A$,
$\beta_\prec,\beta_\succ: A^*\rightarrow A^*\otimes A^*$ and

(a) $(\Delta_\prec^*,\Delta_\succ^*):A^*\otimes A^*\rightarrow
A^*$ defines a dendriform algebra structure
$(\succ_{A^*},\prec_{A^*})$ on $A^*$;

(b) $(\beta_\prec^*,\beta_\succ^*):A\otimes A\rightarrow A$ defines
a dendriform algebra structure $(\succ_{A},\prec_A)$ on $A$;

(c) Equations (4.2.1-4.2.6) are satisfied.

\noindent We also denote it  by
$(A,A^*,\Delta_\succ,\Delta_\prec,\beta_\succ,\beta_\prec)$ or
simply $(A,A^*)$.}\end{defn}

\begin{remark} {\rm In fact, there have already been the notions
of dendriform bialgebra ([LR1-2], \cite{Ron}, \cite{A4}) and
bidendriform bialgebras (\cite{F2}) which are the special dendriform
bialgebras. We use the terminology ``D-bialgebra" in order to
express its relation with the double construction. All of these
bialgebras are dendriform algebras equipped with coassociative
cooperations verifying some (different) compatibility relations. We
would like to point out that the dendriform D-bialgebras are quite
different from the other types of bialgebras. For example, one of
the differences is that the term $a\otimes b$ appears in both
$\Delta_\prec (a*b)$ and $\Delta_\succ(a*b)$ in a bidendriform
bialgebra, whereas it does not appear in a dendriform
D-bialgebra.}\end{remark}

\begin{theorem} Let $(A,\prec_A,\succ_A)$  and
$(A^*,\prec_{A^*},\succ_{A^*})$ be two dendriform algebras. Let
$(A,*_A)$ and $(A^*,*_{A^*})$ be the associated associative
algebras respectively. Then the following conditions are
equivalent.

{\rm (1)} There is a double construction of Connes cocycle
associated to $(A,*_A)$ and $(A,*_{A^*})$.

{\rm (2)}  $(A, A^*, R_{\prec_A}^*, L_{\succ_A}^*,
R_{\prec_{A^*}}^*, L_{\succ_{A^*}}^*)$ is a matched pair of the
 associative algebras.

{\rm (3)} $(A,A^*,
R_{\succ_A}^*+R_{\prec_A}^*,-L_{\prec_A}^*,-R_{\succ_A}^*,L_{\succ_A}^*+L_{\prec_A}^*,
R_{\succ_{A^*}}^*+R_{\prec_{A^*}}^*,-L_{\prec_{A^*}}^*,-R_{\succ_{A^*}}^*,L_{\succ_{A^*}}^*+L_{\prec_{A^*}}^*)$
is a matched pair of dendriform algebras.

{\rm (4)} $(A,A^*)$ is a dendriform D-bialgebra.\end{theorem}

\begin{proof}
It follows from Theorems 4.1.5, 4.1.7 and 4.2.1. \end{proof}

\begin{defn}{\rm Let
$(A,A^*,\Delta_\succ,\Delta_\prec,\beta_\succ,\beta_\prec)$ and
$(B,B^*,\Delta_\succ,\Delta_\prec,\beta_\succ,\beta_\prec)$ be two
dendriform D-bialgebras. A  {\it homomorphism of dendriform
D-bialgebras} $\varphi:A\rightarrow B$ is a homomorphism of
dendriform algebras such that $\varphi^*:B^*\rightarrow A^*$ is also
a homomorphism of dendriform algebras, that is, $\varphi$ satisfies
$$(\varphi\otimes \varphi) \Delta_\succ (x)=\Delta_\succ(\varphi(x)),\;\;
(\varphi\otimes \varphi) \Delta_\prec
(x)=\Delta_\prec(\varphi(x)),$$
$$(\varphi^*\otimes
\varphi^*)\beta_\succ(a^*)=\beta_\succ(\varphi^*(a^*)),\;\;
(\varphi^*\otimes
\varphi^*)\beta_\prec(a^*)=\beta_\prec(\varphi^*(a^*)),\eqno
(4.2.7)$$ for any $x\in A,a^*\in B^*$.  An {\it isomorphism of
dendriform D-bialgebras} is an invertible homomorphism of dendriform
D-bialgebras.}\end{defn}

\begin{prop} Two double constructions of Connes cocycles are isomorphic
if and only if their corresponding dendriform D-bialgebras are
isomorphic.\end{prop}

\begin{proof}
It follows from a similar proof as of Proposition 2.2.10.
\end{proof}

\begin{exam} {\rm Let
$(A,A^*,\Delta_\succ,\Delta_\prec,\beta_\succ,\beta_\prec)$ be a
dendriform D-bialgebra. Then its dual
$(A^*,A,\beta_\succ,\beta_\prec,\Delta_\succ,\Delta_\prec)$ is also
a dendriform D-bialgebra.}\end{exam}

\begin{exam}
{\rm Let $(A,\succ_A,\prec_A)$ be a dendriform algebra. If the
dendriform algebra structure on $A^*$ is trivial, then
$(A,A^*,0,0,\beta_\succ,\beta_\prec)$ is a dendriform D-bialgebra.
And its corresponding dendriform algebra is
$A\ltimes_{R_\succ^*+R_\prec^*,-L_\prec^*,
-R_\succ^*,L_\succ^*+L_\prec^*} A^*$. Moreover, its corresponding
double construction of Connes cocycle is just the semidirect sum
$A\ltimes_{R_{\prec_A}^*,L_{\succ_A}^*} A^*$ with the bilinear form
$\omega$ given by equation (1.4.1). Dually, if $A$ is a trivial
dendriform algebra, then the dendriform D-bialgebra structures on
$A$ are in one-to-one correspondence with the dendriform algebra
structures on $A^*$.}\end{exam}

\begin{exam} {\rm Let $(A,A^*)$ be a dendriform D-bialgebra.  In the next
subsection, we will prove that there exists a canonical dendriform
D-bialgebra structure on the direct sum $A\oplus A^*$ of the
underlying vector spaces of $A$ and $A^*$.}\end{exam}

\subsection{Coboundary dendriform D-bialgebras}

By Theorem 4.2.1, we have shown that both $\Delta_\succ$ and
$\Delta_\prec$ ($\beta_\succ$ and $\beta_\prec$ respectively ) are
the 1-cocycles of the associated associative algebra $(A,*_A)$
($(A^*,*_{A^*})$ respectively). So it is natural to consider the
special case that they are 1-coboundaries or principal
derivations, as we have done in subsection 2.3.

Let $(A, \succ,\prec)$ be a dendriform algebra and
$r_\succ,r_\prec\in A\otimes A$. Set
$$\Delta_\succ(x)=( id\otimes  L(x)-R_\prec(x) \otimes id  )r_\succ;\eqno (4.3.1)$$
$$\Delta_\prec(x)=( id\otimes  L_\succ(x)-R(x) \otimes id  )r_\prec,\eqno (4.3.2)$$ for any $x\in A$. It is obvious that
$\Delta_\succ$ satisfies equation (4.2.1) and $\Delta_\prec$
satisfies equation (4.2.2). Moreover, by equation (4.2.5), we show
that
$$(L(x) \otimes id  - id\otimes  R_\prec(x))( id\otimes
L_\succ(y)-R(y) \otimes id  )(r_\prec+\sigma
(r_\succ))=0,\;\;\forall x,y\in A.\eqno (4.3.3)$$ Therefore
$(A,\Delta_\succ,\Delta_\prec,\beta_\succ,\beta_\prec)$ is a
dendriform D-bialgebra  if and only if the following conditions are
satisfied:

(1) $\Delta_\succ^*,\Delta_\prec^*:A^*\otimes A^*\rightarrow A^*$
defines a dendriform algebra structure on $A^*$.

(2) $\beta_\succ,\beta_\prec$ satisfy equations (4.2.3-4.2.4) and
(4.2.6), where the dendriform algebra structure on $A^*$ is given
by (1).

\begin{prop} Let $(A,\succ,\prec)$ be a dendriform algebra whose products
are given by two linear maps $\beta_\succ^*,\beta_\prec^*:A\otimes
A\rightarrow A$ and $r_\succ,r_\prec\in A\otimes A$. Suppose there
exists a dendriform algebra structure $``\succ_{A^*},\prec_{A^*}"$
on $A^*$ given by $\Delta_\succ^*,\Delta_\prec^*:A^*\otimes
A^*\rightarrow A^*$, where $\Delta_\succ$ and $\Delta_\prec$ are two
linear maps given by equations {\rm (4.3.1)} and {\rm (4.3.2)}
respectively. Then

{\rm (1)} Equation {\rm (4.2.3)} holds if and only if $r_\succ,
r_\prec$ satisfy
$$[R_\prec(x)\otimes L_\succ (y)- id\otimes  L_\succ (y\prec
x)-R_\prec(y\succ x) \otimes id  ](r_\succ+r_\prec)=0,\;\forall
x,y\in A.\eqno (4.3.4)$$

{\rm (2)} Equation {\rm (4.2.4)} holds if and only if $r_\succ,
r_\prec$ satisfy equation {\rm (4.3.4)}.

{\rm (3)} Equation {\rm (4.2.6)} holds if and only if $r_\succ,
r_\prec$ satisfy (for any $x,y\in A$)
$$[L_\succ(x) \otimes id  - id\otimes  R_\prec(x)][- id\otimes
L_\succ(y)+R_\prec(y) \otimes id  ](r_\prec+r_\succ)$$
$$+[L_\succ(x) \otimes id  - id\otimes  R_\prec(x)][R_\succ(y) \otimes id
(r_\prec+\sigma(r_\succ))- id\otimes  L_\prec(y)
(\sigma(r_\prec)+r_\succ)]=0.\eqno (4.3.5)$$
\end{prop}

\begin{proof} Let $\{ e_1,\cdots, e_n\}$ be a basis of $A$ and $\{
e_1^*,\cdots,e_n^*\}$ be its dual basis. Set
$$r_\prec=\sum\limits_{i,j} a_{ij}e_i\otimes e_j,\;\;
r_\succ=\sum\limits_{i,j} b_{ij}e_i\otimes e_j.$$
$$e_i\succ e_j=\sum\limits_{k=1}^na_{ij}^ke_k,\;\;
e_i\prec e_j=\sum\limits_{k=1}^nb_{ij}^ke_k,\;\; e_i^*\succ
e_j^*=\sum\limits_{k=1}^nc_{ij}^k e_k^*,\;\; e_i^*\prec
e_j^*=\sum\limits_{k=1}^nd_{ij}^k e_k^*.$$ By equations (4.3.1) and
(4.3.2), we have (for any $i,k,l$)
$$c_{kl}^i=\sum_{m=1}^n[b_{km}(a_{im}^l+b_{im}^l)-b_{ml}b_{mi}^k],\;\;
d_{kl}^i=\sum_{m=1}^n[a_{km}a_{im}^l-a_{ml}(a_{mi}^k+b_{mi}^k)].\eqno
(4.3.6)$$

(1) Equation (4.2.3) holds (taking $a^*=e_i^*, b^*=e_j^*$) if and
only if (for any $i,j,m,t$)
$$\sum_{k=1}^n(c_{ij}^k+d_{ij}^k)b_{mt}^k=\sum_{k=1}^n[b_{mk}^jc_{ik}^t+b_{kt}^i(c_{kj}^m+d_{kj}^m)].$$
Substituting equation (4.3.6) into the above equation and after
rearranging the terms suitably, we have
$$(F1)+(F2)+(F3)+(F4)+(F5)+(F6)=0,$$
where
\begin{eqnarray*}
(F1)&=&\sum_{k,l}(a_{kl}+b_{kl})(a_{ml}^jb_{kt}^i);\;\;(F2)=
\sum_{k,l}(b_{kl}b_{kt}^ib_{ml}^j-b_{lk}b_{ln}^ib_{mk}^j);\\
(F3)&=&\sum_{k,l}(a_{il}+b_{il})(-a_{kl}^jb_{mt}^k);\;\; (F4)=\sum_{k,l}b_{il}[b_{mk}^j(a_{tl}^k+b_{tl}^k)-b_{mt}^kb_{kl}^j];\\
(F5)&=&\sum_{k,l}(a_{lj}+b_{lj})(b_{mt}^kb_{lk}^i-b_{kt}^ib_{lm}^k);\;\;(F6)=\sum_{k,l}
a_{lj}(a_{lk}^ib_{mt}^k-a_{lm}^kb_{kt}^i).
\end{eqnarray*}

$(F1)$ is the coefficient of $e_i\otimes e_j$ in
$[R_\prec(e_t)\otimes L_\succ(e_m)](r_\succ+r_\prec)$;

 $(F2)=0$ by interchanging the indices $k$ and $l$;

 $(F3)$ is the coefficient of $e_i\otimes e_j$ in $-[
id\otimes  L_\succ (e_m\prec e_t)](r_\succ+r_\prec)$;

  $(F4)=0$ since the term in the bracket is the coefficient
of $e_j$ in
$$e_m\prec (e_t\succ e_l+e_t\prec e_l)-(e_m\prec
e_n)\prec e_l=0;$$

$(F5)$ is the coefficient of $e_i\otimes e_j$ in $-[R_\prec(e_m\succ
e_t) \otimes id  ](r_\succ+r_\prec)$.

 $(F6)=0$ since the term in the bracket is the coefficient
of $e_i$ in
$$e_l\succ (e_m\prec e_t)-(e_l\succ e_m)\prec e_t=0.$$
Therefore we have
$$[R_\prec(e_t)\otimes L_\succ (e_m)- id\otimes  L_\succ (e_m\prec
e_t)-R_\prec(e_m\succ e_t) \otimes id  ](r_\succ+r_\prec)=0.$$

(2) Similarly, we show that equation (4.2.4) holds if and only if
$r_\succ, r_\prec$ satisfy equation (4.3.4). In fact, comparing with
the proof in (1), the difference appears in $(F2)'$, $(F4)'$ and
$(F6)'$, where

 $(F2)'=\sum\limits_{k,
l}(a_{mk}^ja_{lt}^i-a_{kt}^ia_{ml}^j)=0$ by interchanging the
indices $k$ and $l$;

 $(F4)'=\sum\limits_{k,l} b_{il}
(a_{mt}^kb_{kl}^j-a_{mk}^jb_{tl}^k)=0$ since the term in the bracket
is the coefficient of $e_j$ in
$$(e_m\succ e_t)\prec e_l-e_m\succ (e_t\prec e_l)=0;$$

$(F6)'=\sum\limits_{k,l}a_{lj}[a_{kt}^i(a_{lm}^k+b_{lm}^k)-a_{mt}^ka_{lk}^i]=0$
since the term in the bracket is the coefficient of $e_i$ in
$-e_l\succ (e_m\succ e_t)+ (e_l\succ e_m+e_l\prec e_m)\succ e_t=0$.

(3) Equation (4.2.6) holds (taking $a^*=e_i^*, b^*=e_j^*$) if and
only if (for any $i,j,m,t$)
$$\sum_{l=1}^n[(c_{il}^m+d_{il}^m)b_{lt}^j-b_{ml}^jd_{li}^t+a_{lm}^ic_{jl}^t-a_{tl}^i(c_{lj}^m+d_{lj}^m)]=0.$$
Substituting equation (4.3.6) into the above equation and after
rearranging the terms suitably, we have
$$(F1)+(F2)+(F3)+(F4)+(F5)+(F6)+(F7)+(F8)+(F9)+(F10)=0,$$
where
\begin{eqnarray*}
(F1)&=&\sum_{k,l}(a_{kl}+b_{kl})(-b_{lt}^jb_{km}^i) \Longrightarrow
-R_\prec(e_m)\otimes R_\prec(e_t)(r_\succ+r_\prec);\\
(F2)&=&\sum_{k,l}(a_{lk}+b_{lk})(-a_{mk}^ja_{tl}^i) \Longrightarrow
-L_\succ(e_t)\otimes L_\succ(e_m)(r_\succ+r_\prec);\\
(F3)&=&\sum_{k,l}(a_{kl}+b_{lk})(-a_{km}^ib_{lt}^j) \Longrightarrow
-R_\succ(e_m)\otimes R_\prec(e_t)(\sigma(r_\succ)+r_\prec);\\
(F4)&=&\sum_{k,l}(a_{lk}+b_{kl})(-a_{mk}^jb_{tl}^j) \Longrightarrow
-L_\succ(e_t)\otimes L_\prec(e_m)(r_\succ+\sigma(r_\prec));\\
(F5)&=&\sum_{k,l}(a_{ik}+b_{ik})a_{mk}^lb_{lt}^j \Longrightarrow
 id\otimes  R_\prec(e_t)L_\succ(e_m) (r_\succ+r_\prec);\\
(F6)&=&\sum_{k,l}a_{ki}(a_{kt}^l+b_{kt}^l)b_{ml}^j \Longrightarrow
 id\otimes  R_\prec(e_t)L_\prec(e_m) (\sigma (r_\prec));\\
(F7)&=&\sum_{k,l}b_{ik}b_{lt}^jb_{mk}^l \Longrightarrow  id\otimes
R_\prec(e_t)L_\prec(e_m) (r_\succ);\\
(F8)&=&\sum_{k,l}(a_{kj}+b_{kj})a_{tl}^ib_{km}^l \Longrightarrow
L_\succ(e_t)R_\prec(e_m) \otimes id   (r_\succ+r_\prec);\\
(F9)&=&\sum_{k,l}a_{kj}a_{km}^la_{tl}^i \Longrightarrow
L_\succ(e_t)R_\succ(e_m) \otimes id   (r_\prec);\\
(F10)&=&\sum_{k,l}b_{jk}a_{lm}^i(a_{tk}^l+b_{tk}^l) \Longrightarrow
L_\succ(e_t)R_\succ(e_m) \otimes id (\sigma(r_\succ)).
\end{eqnarray*}
Therefore equation (4.3.5) holds.\end{proof}

By the definition of a dendriform algebra, we have the following
conclusion (cf. \cite{F2}).

\begin{lemma}  Let $A$ be a vector space and $\Delta_\succ,\Delta_\prec:
A\otimes A\rightarrow A$ be two linear maps. Then
$\Delta_\succ^*,\Delta_\prec^*: A^*\otimes A^*\rightarrow A^*$
define a dendriform algebra  structure on $A^*$ if and only if the
following conditions are satisfied:

{\rm (1)} $(\Delta_\prec\otimes {id})\Delta_\prec=({id}\otimes
(\Delta_\succ+\Delta_\prec))\Delta_\prec$;\hfill {\rm (4.3.7)}

{\rm (2)} $({id}\otimes
\Delta_\prec)\Delta_\succ=(\Delta_\succ\otimes {
id})\Delta_\prec$;\hfill {\rm (4.3.8)}

{\rm (3)} $({ id}\otimes
\Delta_\succ)\Delta_\succ=((\Delta_\succ+\Delta_\prec)\otimes
{id})\Delta_\succ$.\hfill {\rm (4.3.9)}
\end{lemma}

\begin{prop} Let $(A,\succ,\prec)$ be a dendriform algebra and
$r_\succ,r_\prec\in A\otimes A$. Define
$\Delta_\succ,\Delta_\prec:A\rightarrow A\otimes A$ by equations
{\rm (4.3.1-4.3.2)}. Then $\Delta_\succ^*,\Delta_\prec^*: A^*\otimes
A^*\rightarrow A^*$ define a dendriform algebra  structure on $A^*$
if and only if the following equations are satisfied (for any $x\in
A$)
\begin{eqnarray*}
&& (R(x)\otimes  id\otimes   id)[(r_{\prec,12}*r_{\prec,
13}+r_{\prec,13}\prec r_{\succ,23}-r_{\prec,23}\succ r_{\prec,
12})\\
&&\hspace{1cm}+r_{\prec,13}\succ (r_{\prec,
23}+r_{\succ,23})-(r_{\prec,
23}+r_{\succ,23})\prec r_{\prec, 12}]\\
&&\hspace{1cm} + (r_{\prec,23}+r_{\succ,23})\prec [( id\otimes
L_\prec(x) \otimes id  )r_{\prec,12}]\\
&&\hspace{1cm}+( id\otimes    id\otimes   L_\succ
(x))(-r_{\prec,12}*r_{\prec, 13}-r_{\prec,13}\prec
r_{\succ,23}+r_{\prec,23}\succ r_{\prec,
12})\\
&&\hspace{1cm}-[( id\otimes    id\otimes   L_\succ(x))r_{13}]\succ
(r_{\succ,23}+r_{\prec,23})=0;\hspace{5cm} (4.3.10)\\
&&(R_\prec(x) \otimes id    \otimes id
)(r_{\prec,23}*r_{\succ,12}-r_{\succ,12}\prec r_{\prec,
13}-r_{\succ,13}\succ r_{\prec, 23})\\
&&\hspace{1cm}-( id\otimes    id\otimes
L_\succ(x))(r_{\prec,23}*r_{\succ,12}-r_{\succ,12}\prec r_{\prec,
13}-r_{\succ,13}\succ r_{\prec, 23})=0;\hspace{1.2cm} (4.3.11)\\
&&(R_\prec(x)\otimes  id\otimes   id)(-r_{\succ,
13}*r_{\succ,23}+r_{\succ,23}\prec r_{\succ, 12}-r_{\prec, 12}\succ
r_{\succ, 13})\\
&&\hspace{1cm} -(r_{\succ,12}+r_{\prec, 12})\prec
[(R_\prec(x)\otimes  id\otimes   1)r_{\succ,13}]\\
&&\hspace{1cm} +[( id\otimes   R_\prec(x) \otimes id
)r_{\succ,23}]\succ
(r_{\succ,12}+r_{\prec,12})\\
&&\hspace{1cm} +( id\otimes    id\otimes   L(x))[r_{\succ,
13}*r_{\succ,23}-r_{\succ,23}\prec r_{\succ, 12}+r_{\prec, 12}\succ
r_{\succ, 13}\\
&&\hspace{1cm}+(r_{\succ,12}+r_{\prec,12})\prec r_{\succ,
13}-r_{\succ,23}\succ (r_{\succ,12}+r_{\prec,12})]=0.\hspace{3.8cm}
(4.3.12)
\end{eqnarray*}The operation between two $r$s is given in an obvious and similar
way as equation {\rm (1.1.3)}.\end{prop}

\begin{proof} We need to prove that equations (4.3.7-4.3.9) are
equivalent to equations (4.3.10-4.3.12) respectively. Here we only
give an explicit proof that equation (4.3.10) holds if and only if
equation (4.3.7) holds since the proof of the other two equations is
similar. Let $x\in A$. After rearranging the terms suitably, we
divide equation (4.3.7) into three parts:
$$(\Delta_\prec\otimes { id})\Delta_\prec(x)-({id}\otimes
(\Delta_\succ+\Delta_\prec))\Delta_\prec(x)=(F1)+(F2)+(F3),$$
where
\begin{eqnarray*}
(F1)&=&\sum_{i,j} \{(a_i\succ x+a_i\prec x)\otimes [a_j\otimes
b_i\succ b_j-(a_j\succ b_i+a_j\prec b_i)\otimes b_j+c_j\otimes
(b_i\succ d_j\\
 &\mbox{}&+b_i\prec d_j)-c_j\prec b_i\otimes
d_j]+[a_j\succ(a_i\succ x+a_i\prec x)+a_j\prec(a_i\succ x+a_i\prec
x)]\\
&\mbox{}&\otimes b_j\otimes b_i\}; \\
(F2)&=& \sum_{i,j}\{a_i\otimes [a_j\succ (x\succ b_i)+a_j\prec
(x\succ b_i)]\otimes b_j+a_i\otimes c_j\prec (x\succ b_i)\otimes
d_j\\
&\mbox{}&-a_j\otimes (a_i\succ x+a_i\prec x)\succ b_j\otimes
b_i\};\\
(F3)&=& \sum_{i,j}\{[a_i\otimes (a_i\succ b_j)-(a_j\succ
a_i+a_j\prec a_i)\otimes b_j]\otimes (x\succ b_i)-a_i\otimes
a_j\otimes [(x\succ b_i)\succ b_j]\\
&\mbox{}&-a_i\otimes c_j\otimes [(x\succ b_i)\succ  d_j+(x\succ
b_i)\prec d_j]\}.
\end{eqnarray*}
On the other hand,
\begin{eqnarray*}
(F1a)&=&(R(x)\otimes  id\otimes   id)(r_{\prec,12}*r_{\prec, 13})
=\sum_{i,j}[(a_i*a_j)*x \otimes b_i\otimes b_j]\\
&=& \sum_{i,j}[a_j\succ(a_i\succ x+a_i\prec x)+a_j\prec(a_i\succ
x+a_i\prec x)]\otimes b_j\otimes b_i];\\
(F1b)&=&(R(x)\otimes  id\otimes   id) (r_{\prec,13}\prec
r_{\succ,23}) =
\sum_{i,j}[(a_i*x) \otimes c_j\otimes (b_i\prec d_j)]\\
&=&\sum_{i,j}[(a_i\succ x+a_i\prec x)\otimes c_j\otimes (b_i\prec
d_j)];\\
(F1c)&=&(R(x)\otimes  id\otimes   id)(-r_{\prec,23}\succ r_{\prec,
12}) =
\sum_{i,j}[-(a_i*x)\otimes (a_j\succ b_i)\otimes b_j]\\
&=&\sum_{i,j}[-(a_i\succ x+a_i\prec x)\otimes (a_j\succ b_i)\otimes
b_j];\\
(F1d)&=& (R(x)\otimes  id\otimes   id)[r_{\prec,13}\succ (r_{\prec,
23}+r_{\succ,23})]\\
&=& \sum_{i,j}\{(a_i\succ x+a_i\prec x)\otimes [a_j\otimes(b_i\succ
b_j)+c_j\otimes (b_i\succ d_j)]\};\\
(F1e)&=& (R(x)\otimes  id\otimes   id)[-(r_{\prec,
23}+r_{\succ,23})\prec
r_{\prec, 12}]\\
&=&-\sum_{i,j}\{(a_i\succ x+a_i\prec x)\otimes [(a_j\prec
b_i)\otimes b_j+(c_j\prec b_i)\otimes d_j]\};\\
(F2')&=&(r_{\prec,23}+r_{\succ,23})\prec [( id\otimes
L_\prec(x) \otimes id  )r_{\prec,12}]\\
&=&\sum_{i,j}a_i\otimes [a_j\prec(x\succ b_i)\otimes b_j+ c_j\prec
(x\succ b_i)] \otimes d_j]\\
(F3a)&=&( id\otimes    id\otimes   L_\succ
(x))(-r_{\prec,12}*r_{\prec, 13})
=\sum_{i,j} -a_i*a_j\otimes b_i\otimes (x\succ b_j)\\
&=&\sum_{i,j}-[(a_i\succ a_j+a_i\prec a_j)\otimes b_i\otimes (x\succ
b_j);\\
(F3b)&=&( id\otimes    id\otimes   L_\succ (x))(-r_{\prec,13}\prec
r_{\succ,23})=\sum_{i,j}-[a_i\otimes c_j\otimes x\succ (b_i\prec d_j)];\\
(F3c)&=&( id\otimes    id\otimes   L_\succ (x))(r_{\prec,23}\succ
r_{\prec, 12})=\sum_{i,j}[ a_i\otimes (a_j\succ b_i)\otimes
(x\succ
b_j)];\\
(F3d)&=&-[( id\otimes    id\otimes   L_\succ(x))r_{13}]\succ
(r_{\succ,23}+r_{\prec,23})\\
&=&-\sum_{i,j}a_i\otimes [a_j\otimes (x\succ b_i)\succ
b_j+c_j\otimes (x\succ b_i)\succ d_j].
\end{eqnarray*}
It is obvious that
$$(F1)=(F1a)+(F1b)+(F1c)+(F1d)+(F1e),\;$$$$(F2)=(F2)',\;
(F3)=(F3a)+(F3b)+(F3c)+(F3d).$$ Therefore equation (4.3.10) holds if
only if equation (4.3.7) holds.\end{proof}

Combining Propositions 4.3.1 and 4.3.3, we obtain the following
conclusion.

\begin{theorem}Let $(A,\succ,\prec)$ be a dendriform algebra and
$r_\succ,r_\prec\in A\otimes A$. Then the linear maps
$\Delta_\succ,\Delta_\prec$ defined by equations {\rm (4.3.1)} and
{\rm (4.3.2)} induce a dendriform algebra structure on $A^*$ such
that $(A,A^*)$ is a dendriform D-bialgebra  if and only if $r_\succ$
and $r_\prec$ satisfy equations {\rm (4.3.3-4.3.5)} and {\rm
(4.3.10-4.3.12)}.\end{theorem}

\begin{defn}{\rm  A dendriform D-bialgebra  $(A,A^*)$ is called {\it coboundary} if
its structure is given by $r_\succ,r_\prec\in A\otimes A$ through
Theorem 4.3.4.}\end{defn}

\begin{theorem}
Let $(A,A^*,\Delta_\succ,\Delta_\prec,\beta_\succ,\beta_\prec)$ be a
dendriform D-bialgebra.  Then there is a canonical dendriform
bialgebra structure on the direct sum $A\oplus A^*$ of the
underlying vector spaces of $A$ and $A^*$ such that both the
inclusions $i_1:A\rightarrow A\oplus A^*$ and $i_2:A^*\rightarrow
A\oplus A^*$ into the two summands are homomorphisms of dendriform
D-bialgebras, where the dendriform D-bialgebra structure on $A^*$ is
given in {\rm Example 4.2.7}.\end{theorem}

\begin{proof} Let $r=\sum_{i} e_i\otimes e_i^*\in A\otimes A^*\subset (A\oplus A^*)\otimes (A\oplus
A^*)$ which corresponds to the identity map $id:A\rightarrow A$,
where $\{ e_1,\cdots,e_n\}$ is a basis of $A$ and $\{ e_1^*,\cdots,
e_n^*\}$ is its dual basis. Suppose that the dendriform D-bialgebra
structure ``$\succ,\prec$'' on $A\oplus A^*$ is given by
$${\mathcal D}{\mathcal D}(A)=A\bowtie^{R_{\succ_A}^*+R_{\prec_A}^*,-L_{\prec_A}^*,-R_{\succ_A}^*,L_{\succ_A}^*+L_{\prec_A}^*}
_{R_{\succ_{A^*}}^*+R_{\prec_{A^*}}^*,-L_{\prec_{A^*}}^*,-R_{\succ_{A^*}}^*,L_{\succ_{A^*}}^*+L_{\prec_{A^*}}^*}A^*.$$
Then we have  (for any $x,y\in A, a,b\in A^*$)
\begin{eqnarray*}
&&x\succ y=x\succ_A y,\;\; x\prec y=x\prec_A y,\;\; x\succ
a=R^*_A(x)a-L_{\prec_{A^*}}^*(a)x\\
&&x\prec a=-R_{\succ_A}^*(x)a+L_{A^*}^*(a)x,\;\; a\succ
x=R^*_{A^*}(a)x-L_{\prec_{A}}^*(x)a,\;\;\\
&&a\prec x=-R_{\succ_{A^*}}^*(a)x+L_{A}^*(x)a,\;\;a\succ
b=a{\succ_{A^*}}b,\;\;a\prec b=a{\prec_{A^*}}b.
\end{eqnarray*}
If $r_\succ=r$ and $r_\prec=-r$ satisfies equations (4.3.3)-(4.3.5)
and (4.3.10)-(4.3.12), then
$$\Delta_{{\mathcal D}{\mathcal D},\succ}(u)=( id\otimes   L(u)-R_\prec(u) \otimes id  )(r_\succ),\;\;\Delta_{{\mathcal D}{\mathcal D},\prec}(u)=(
id\otimes   L_\prec(u)-R(x) \otimes id  )(r_\prec),\;\forall u\in
{\mathcal D}{\mathcal D}(A)$$ can induce a dendriform D-bialgebra
structure on ${\mathcal D}{\mathcal D}(A)$.

In fact, we have
$$r_\prec+r_\succ=0,\;\;r_\prec+\sigma(r_\succ)=\sum_i (-e_i\otimes
e_i^* +e_i^*\otimes e_i)$$ Therefore equation (4.3.4) holds
automatically. By a similar proof as of Theorem 2.3.6, we show that
equations (4.3.3) and (4.3.5)  hold and
\begin{eqnarray*}
r_{12}*r_{13}-r_{13}\prec r_{23}-r_{23}\succ r_{12}&=&
-r_{23}*r_{12}+r_{12}\prec r_{13}+r_{13}\succ
r_{23}\\&=&-r_{13}*r_{23}+r_{23}\prec r_{12}+r_{12}\succ r_{13}=0.
\end{eqnarray*}
So equations (4.3.10)-(4.3.12) are satisfied. Hence ${\mathcal
D}{\mathcal D}(A)$ is a dendriform D-bialgebra. Furthermore, for
$e_k\in A$, we have
\begin{eqnarray*}
\Delta_{{\mathcal D}{\mathcal D},\succ}(e_k)&=&\sum_{i} [e_i\otimes
e_k*e_i^*-(e_i\prec e_k)\otimes e_i^*]=\sum_{i,j}\langle
e_k,e_i^*\succ e_j^*\rangle e_i\otimes e_j
=\Delta_{\succ} (e_k);\\
\Delta_{{\mathcal D}{\mathcal D},\prec}(e_k)&=&\sum_{i} [-e_i\otimes
e_k\prec e_i^*+(e_i*e_k)\otimes e_i^*]=\sum_{i,j}\langle
e_k,e_i^*\prec e_j^*\rangle e_i\otimes e_j = \Delta_{\prec} (e_k).
\end{eqnarray*}
Therefore the inclusion $i_1:A\rightarrow A\oplus A^*$ is a
homomorphism of dendriform D-bialgebras. Similarly, the inclusion
$i_2:A^*\rightarrow A\oplus A^*$ is also a homomorphism of
dendriform D-bialgebras, where the dendriform D-bialgebra structure
on $A^*$ is given in Example {\rm 4.2.7}. \end{proof}

\begin{defn} {\rm Let $(A,A^*)$ be a dendriform D-bialgebra. With the dendriform
D-bialgebra  structure given in Theorem 4.3.6, $A\oplus A^*$ is
called  a {\it dendriform double} of $A$. We denote it by ${\mathcal
D}{\mathcal D}(A)$.}\end{defn}


\begin{coro}  Let $(A,A^*)$ be a dendriform D-bialgebra. Then the
dendriform double ${\mathcal D}{\mathcal D}(A)$ of $A$ is a
dendriform D-bialgebra  and the bilinear form $\omega$ given by
equation {\rm (1.4.1)} is a Connes cocycle.\end{coro}

At the end of this subsection, we would like to point out that,
unlike the symmetry of 1-cocycles of $A$ and $A^*$ appearing in the
definition of a dendriform D-bialgebra  $(A,A^*)$, it is not
necessary that $\beta$ is also a 1-coboundary of $A^*$ for a
coboundary dendriform D-bialgebra
$(A,A^*,\Delta_\succ,\Delta_\prec,\beta_\succ,\beta_\prec)$, where
$\Delta_\succ,\Delta_\prec$ are given by equations (4.3.1-4.3.2).

\subsection{ $D$-equation and its properties}
In this subsection, we consider some simple and special cases to
satisfy the equations (4.3.3-4.3.5) and (4.3.10-4.3.12).

At first, due to equation (4.3.3), we consider the condition
$$r_\prec =r,\;r_\succ=-\sigma (r),\;\;r\in A\otimes A.\eqno
(4.4.1)$$

\begin{coro} Let $(A,\succ,\prec)$ be a dendriform algebra and $r=\sum_i
a_i\otimes b_i\in A\otimes A$. Then the maps
$\Delta_\succ,\Delta_\prec$ defined by equations {\rm (4.3.1)} and
{\rm (4.3.2)} with $r_\succ, r_\prec$ satisfying equation {\rm
(4.4.1)} induce a dendriform algebra structure on $A^*$ such that
$(A,A^*)$ is a dendriform D-bialgebra  if and only if $r$ satisfies
the following equations
$$[P(x\succ y)-( id\otimes   L_\succ (x))P(y)](r-\sigma(r))=0;\eqno
(4.4.2)$$
$$\sigma (P(x))P(y)(r-\sigma(r))=0;\eqno (4.4.3)$$
$$(R(x)\otimes  id\otimes   id- id\otimes    id\otimes   L_\succ (x))[(r_{12}*r_{13}-r_{13}\prec r_{32}-r_{23}\succ
r_{12}) $$
$$+\sum_i (a_i*x)\otimes P(b_i)(r-\sigma(r))-a_i\otimes [P(x\succ
b_i)(r-\sigma(r))]=0;\eqno (4.4.4)$$
$$(R_\prec(x) \otimes id    \otimes id  - id\otimes    id\otimes   L_\prec(x))(-r_{23}*r_{21}+r_{21}\prec
r_{13}+r_{31}\succ r_{23})=0;\eqno (4.4.5)$$
$$(R_\prec(u)\otimes  id\otimes   id- id\otimes    id\otimes   L(u))(-r_{31}*r_{32}+r_{32}\prec r_{21}+r_{12}\succ
r_{31})$$
$$+\sum_i [P(b_i)(r-\sigma(r))\otimes x*a_i-P(b_i\prec
x)(r-\sigma(r))\otimes a_i] =0,\eqno (4.4.6)$$ where $x,y\in A,
P(x)= id\otimes   L_\succ(x)-R_\prec (x) \otimes id  $.\end{coro}

\begin{remark} {\rm Let $\sigma_{123},\sigma_{132}:A\otimes A\otimes
A\rightarrow A\otimes A\otimes A$ be two linear maps given by
$$\sigma_{123}(x\otimes y\otimes z)=z\otimes x\otimes y,\;\;
\sigma_{132}(x\otimes y\otimes z)=y\otimes z\otimes x.\;\;\forall
\;x,y,z\in A.\eqno (4.4.7)$$ Then we have
\begin{eqnarray*}
(r_{23}*r_{21}-r_{21}\prec r_{13}-r_{31}\succ r_{23})
&=&\sigma_{123}(r_{12}*r_{13}-r_{13}\prec r_{32}-r_{23}\succ r_{12});\\
(r_{31}*r_{32}-r_{32}\prec r_{21}-r_{12}\succ r_{31}) &=&
\sigma_{132}(r_{12}*r_{13}-r_{13}\prec r_{32}-r_{23}\succ r_{12}).
\end{eqnarray*}}\end{remark}

\begin{remark}{\rm  We also can consider the case that $r_\succ+r_\prec=0$ as we
have done in the proof of Theorem 4.3.6. Obviously, if in addition,
$r_\prec=r$ is symmetric, then this case is as the same as the case
satisfying equation (4.4.1).}\end{remark}

The simplest way to satisfy equations (4.4.2-4.4.6) is to assume
that $r$ is symmetric and
$$r_{12}*r_{13}=r_{13}\prec r_{23}+r_{23}\succ r_{12}.\eqno
(4.4.8)$$

\begin{coro} Let $(A,\succ,\prec)$ be a dendriform algebra and $r\in
A\otimes A$. Suppose $r$ is symmetric and $r$ satisfies equation
{\rm (4.4.8)}. Then the maps $\Delta_\succ,\Delta_\prec$ defined by
equations {\rm (4.3.1)} and {\rm (4.3.2)} with $r_\succ=-r,
r_\prec=r$ induce a dendriform algebra structure on $A^*$ such that
$(A,A^*)$ is a dendriform D-bialgebra.\end{coro}

\begin{defn} {\rm Let $(A,\succ,\prec)$ be a dendriform algebra  and $r\in
A\otimes A$. Equation (4.4.8) is called {\it ${D}$-equation in
$A$.}}\end{defn}

By Remark 4.4.2, when $r$ is symmetric, the equivalent forms of
$D$-equation  are given as
$$r_{23}*r_{12}=r_{12}\prec r_{13}+r_{13}\succ r_{23};\;\;{\rm
or}\;\;r_{13}*r_{23}=r_{23}\prec r_{12}+r_{12}\succ r_{13}.\eqno
(4.4.9)$$

By a similar proof as of Proposition 2.4.4, we have the following
conclusion.

\begin{prop} Let $(A,\succ,\prec)$ be a dendriform algebra and $r\in
A\otimes A$ be a symmetric solution of $D$-equation in $A$. Then the
dendriform algebra structure and its associated associative algebra
structure on the dendriform double ${\mathcal D}{\mathcal D}(A)$ is
given from the products in $A$ as follows (for any $x\in A,
a^*,b^*\in A^*$).

{\rm (a)} $a^*\prec b^*=-R_\succ^*(r(a^*))b^*+L^*(r(b^*))a^*$,
$a^*\succ b^*=R^*(r(a^*))b^*-L^*_\prec(r(b^*))a^*$;\hfill {\rm
(4.4.10)}

{\rm (b)} $a^**b^*= a^*\succ b^*+a^*\prec b^*=R^*_\prec
(r(a^*))b^*+L_\succ^*(r(b^*))a^*$;\hfill {\rm (4.4.11)}

{\rm (c)} $x\succ a^*=x\succ r(a^*)-r(R^*(x)a^*)+R^*(x)a^*$,\\
\mbox{}\hspace{1cm} $x\prec a^*=x\prec
r(a^*)+r(R_\succ^*(x)a^*)-R_\succ^*(x)a^*$; \hfill {\rm (4.4.12)}

{\rm (d)} $x*a^*=x*r(a^*)-r(R_\prec^*(x)a^*)+R_\prec^*(x)a^*$;\hfill
{\rm (4.4.13)}

{\rm (e)}$a^*\succ x=r(a^*)\succ x+r(L_\prec^*(x)a^*)-L_\prec^*(x)a^*,\\
\mbox{}\hspace{1cm} a^*\prec x=r(a^*)\prec
x-r(L^*(x)a^*)+L^*(x)a^*$;\hfill {\rm (4.4.14)}

{\rm (f)} $a^**
x=r(a^*)*x-r(L_\succ^*(x)a^*)+L_\succ^*(x)a^*$.\hfill {\rm
(4.4.15)}\end{prop}

\begin{theorem} Let $(A,\succ,\prec)$ be a dendriform algebra and $r\in
A\otimes A$. Suppose that $r$ is symmetric and nondegenerate. Then
$r$ is a solution of $D$-equation in $A$ if and only if the inverse
of the isomorphism $A^*\rightarrow A$ induced by $r$, regarded as a
bilinear form ${\mathcal B}$ on $A$ (that is, ${\mathcal
B}(x,y)=\langle r^{-1}x,y\rangle $ for any $x,y\in A$) satisfies
$${\mathcal B}(x*y,z)={\mathcal B}(y,z\prec x)+{\mathcal B}(x,y\succ
z),\;\;\forall\; x,y,z\in A.\eqno (4.4.16)$$\end{theorem}

\begin{proof} Let $r=\sum_i a_i\otimes b_i$. Since $r$ is
symmetric, $r(v^*)=\sum_i\langle v^*,a_i\rangle  b_i=\sum_i\langle
v^*,b_i\rangle  a_i$ for any $v^*\in A^*$. Since $r$ is
nondegenerate, for any $x,y,z\in A$, there exist $u^*,v^*,w^*\in
A^*$ such that $x=r(u^*),y=r(v^*),z=r(w^*)$. Therefore

{\small
\begin{eqnarray*} {\mathcal B}(x*y,z)&=& \langle r(u^*)* r(v^*),
w^*\rangle  = \sum_{i,j}\langle u^*,b_i\rangle  \langle
v^*,b_j\rangle  \langle w^*,a_i* a_j\rangle
=\langle w^*\otimes u^*\otimes v^*, r_{12}* r_{13}\rangle  ;\\
{\mathcal B}(y,z\prec x) &=&\langle v^*, r(w^*)\prec r(u^*)\rangle
=\sum_{i,j}\langle u^*,b_i\rangle  \langle w^*,b_i\rangle  \langle
v^*,a_i\prec a_j\rangle  =\langle w^*\otimes u^*\otimes v^*, r_{13}\prec r_{23}\rangle  ;\\
{\mathcal B}(x,y\succ z)&=& \langle r(v^*)\succ r(w^*), u^*\rangle
=\sum_{i,j}\langle v^*,b_i\rangle  \langle w^*,b_j\rangle  \langle
u^*,a_i\succ a_j\rangle  =\langle w^*\otimes u^*\otimes v^*,
r_{23}\succ r_{12}\rangle.
\end{eqnarray*}}
Therefore ${\mathcal B}$ satisfies equation (4.4.16) if and only if
$r$ is a solution of $D$-equation in $A$ .\end{proof}

\begin{defn}{\rm
Let $(A,\succ,\prec)$ be a dendriform algebra. A bilinear form
${\mathcal B}$ on A is called a 2-cocycle if ${\mathcal B}$
satisfies equation (4.4.16).}
\end{defn}

\begin{remark}
{\rm Let ${\mathcal B}$ be 2-cocycle on a dendriform algebra
$(A,\succ,\prec)$. Then it is easy to show that $\omega
(x,y)={\mathcal B}(x,y)-{\mathcal B}(y,x)$ (for any $x,y\in A$) is a
Connes cocycle of the associated associative algebra $(A,*)$. On the
other hand, ${\mathcal B}$ satisfies
$${\mathcal B}(x\cdot y,z)-{\mathcal B}(x,y\cdot z)
={\mathcal B}(y\cdot x,z)-{\mathcal B}(y,x\cdot z),\;\;\forall\;
x,y,z\in A,\eqno (4.4.17)$$ where $x\cdot y=x\succ y-y\prec x$ for
any $x,y\in A$. Furthermore, $(A,\cdot)$ is a pre-Lie algebra (see
subsections 5.2 and 5.3) and a bilinear form on a pre-Lie algebra
$A$ satisfying equation (4.4.17) is called a 2-cocycle on $A$
(\cite{Ku2}). Moreover, a pre-Lie algebra $A$ over the real number
field ${\bf R}$ is called Hessian if there exists a symmetric and
positive definite 2-cocycle on $A$. In geometry, a Hessian manifold
$M$ is a flat affine manifold provided with a Hessian metric $g$,
that is, $g$ is a Remanning metric such that for any each point
$p\in M$ there exists a $C^\infty$-function $\varphi$ defined on a
neighborhood of $p$ such that
$g_{ij}=\frac{\partial^2\varphi}{\partial x^i\partial x^j}$. A
Hessian pre-Lie algebra corresponds to an affine Lie group $G$ with
a $G$-invariant Hessian metric (\cite{Sh}). Therefore a symmetric
and positive definite 2-cocycle on a real dendriform algebra can
give a Hessian structure.}\end{remark}

\begin{coro}
 Let $(A,\succ_A,\prec_A)$ be a dendriform algebra and $r\in
A\otimes A$ be a nondegenerate symmetric solution of $D$-equation in
$A$. Suppose the dendriform algebra structure
$``\succ_{A^*},\prec_{A^*}"$ on $A^*$ is induced by $r$ through {\rm
Proposition 4.4.6}. Then we have
$$a^*\succ_{A^*} b^*=r^{-1}(r(a^*)\succ_A r(b^*)),
a^*\prec_{A^*} b^*=r^{-1}(r(a^*)\prec_A r(b^*)),\;\;\forall
a^*,b^*\in A^*.\eqno (4.4.18)$$ Therefore $r:A^*\rightarrow A$ is an
isomorphism of dendriform algebras.\end{coro}

\begin{proof}
The conclusion can be obtained by a similar proof as of Corollary
2.4.6.
\end{proof}


\begin{theorem} Let $(A,\succ,\prec)$ be a dendriform algebra and $r\in
A\otimes A$ be symmetric. Then $r$ is a solution of $D$-equation
in $A$ if and only if $r$ satisfies
$$r(a^*)*r(b^*)=r(R_\prec^*(r(a^*))b^*+L_\succ^*(r(b^*))a^*),\;\;\forall\; a^*,b^*\in A^*.\eqno
(4.4.19)$$\end{theorem}

\begin{proof} The conclusion can be obtained by a similar proof as
of Theorem 2.4.7.\end{proof}

Combining Theorem 4.4.11 and Theorem 3.1.2, we have the following
conclusion.

\begin{coro} Let $(A,\succ,\prec)$ be a dendriform algebra and $r\in
A\otimes A$ be symmetric. Then $r$ is a solution of $D$-equation in
$A$ if and only if $r$ is an ${\mathcal O}$-operator of the
associated associative algebra $(A,*)$ associated to
$(R_\prec^*,L_\succ^*)$. Therefore there is a dendriform algebra
structure on $A^*$ given by
$$a^*\succ b^*=R_\prec^*(r(a^*))b^*,\;\;a^*\prec b^*=L_\succ^*(r(b^*))a^*,\;\;\forall a^*,b^*\in A^*.\eqno
(4.4.20)$$ It has the same associated associative algebra of the
dendriform  algebra on $A^*$ given by equation {\rm (4.4.11)}, which
is induced by $r$ in the sense of coboundary dendriform
D-bialgebras. If $r$ is nondegenerate, then there is a new
compatible dendriform algebra structure on $A$ given by
$$x\succ' y= r(R_\prec^*(x)r^{-1}y),\;\;
x\prec'y=r(L_\succ^*(y)r^{-1}x),\;\; \forall\; x,y\in A,\eqno
(4.4.21)$$ which is just the dendriform  algebra structure given by
$${\mathcal B}(x\succ'y,z)={\mathcal B}(y,z*x),
\;\;{\mathcal B}(x\prec'y,z)={\mathcal B}(x,y*z),\;\;\forall\;
x,y,z\in A,\eqno (4.4.22)$$ where ${\mathcal B}$  is the symmetric
2-cocycle on $A$ induced by $r^{-1}$.
\end{coro}

\begin{theorem} Let $(A,*)$ be an associative
algebra and $(l,r, V)$ be a bimodule. Let $(r^*,l^*, V^*)$ be the
bimodule of $A$ given by {\rm Lemma 2.1.2}. Suppose that
$T:V\rightarrow A$ is an ${\mathcal O}$-operator associated to
$(l,r,V)$. Then $r=T+\sigma (T)$ is a symmetric solution of the
$D$-equation in $T(V)\ltimes_{r^*,0,0,l^*}V^*$, where $T(V)\subset
A$ is a dendriform algebra given by equation {\rm (3.1.4)} and
$(r^*,0,0,l^*)$ is a bimodule since its associated associative
algebra $T(V)$ is an associative subalgebra of $A$, and $T$ can be
identified as an element in $T(V)\otimes V^*\subset
(T(V)\ltimes_{r^*,0,0,l^*}V^*)\otimes
(T(V)\ltimes_{r^*,0,0,l^*}V^*)$.\end{theorem}

\begin{proof} Let $\{e_1,\cdots,e_n\}$ be a basis of $A$. Let
$\{v_1,\cdots, v_m\}$ be a basis of $V$ and $\{ v_1^*,\cdots,
v_m^*\}$ be its dual basis. Set
$T(v_i)=\sum\limits_{k=1}^na_{ik}e_k, i=1,\cdots, m$. Then
$$T=\sum_{i=1}^m T(v_i)\otimes v_i^*=\sum_{i=1}^m\sum_{k=1}^n
a_{ik}e_k\otimes v_i^*\in T(V)\otimes V^*\subset
(T(V)\ltimes_{r^*,0,0,l^*}V^*)\otimes
(T(V)\ltimes_{r^*,0,0,l^*}V^*).$$ Therefore we have
\begin{eqnarray*}
r_{12}* r_{13} &=&\sum_{i,j=1}^m\{T(v_i)*T(v_j)\otimes v_i^*\otimes
v_j^*+r^*(T(v_i))v_j^*\otimes v_i^*\otimes
T(v_j)\\
&\mbox{}&+l^*(T(v_j))v_i^*\otimes T(v_i)\otimes v_j^*\};\\
r_{13}\prec r_{23} &=&\sum_{i,j=1}^m\{v_i^*\otimes v_j^*\otimes
T(v_i)\prec T(v_j)+T(v_i)\otimes v_j^*\otimes l^*(T(v_j))v_i^*\}\\
r_{23}\succ r_{12} &=&\sum_{i,j=1}^m \{T(v_j)\otimes
r^*(T(v_i))v_j^*\otimes v_i^*+v_j^*\otimes T(v_i)\succ T(v_j)\otimes
v_i^*\}
\end{eqnarray*}
On the other hand, we have
\begin{eqnarray*}
\sum_{i,j=1}r^*(T(v_i))v_j^*\otimes v_i^*\otimes T(v_j)
&=&\sum_{i,j=1} v_j^*\otimes v_i^*\otimes T(r(T(v_i))v_j);\\
\sum_{i,j=1}l^*(T(v_j))v_i^*\otimes T(v_i)\otimes v_j^*
&=&\sum_{i,j=1} v_i^*\otimes T(l(T(v_j))v_i)\otimes v_j^*;\\
\sum_{i,j=1}T(v_i)\otimes v_j^*\otimes l^*(T(v_j))v_i^*
&=&\sum_{i,j=1} T(l(T(v_j))v_i)\otimes v_j^*\otimes v_i^*;\\
\sum_{i,j=1}T(v_j)\otimes r^*(T(v_i))v_j^*\otimes v_i^* &=&
\sum_{i,j=1}T(r(T(v_i))v_j)\otimes v_j^*\otimes v_i^*.
\end{eqnarray*}
Since $T$ is an ${\mathcal O}$-operator of $A$ associated to
$(l,r,V)$ and
$$T(u)\succ T(v)=T(l(T(u))v),\;\;T(u)\prec T(v)=T(r(T(v))u),\;\;\forall u,v\in V,$$
we show that $r$ is a symmetric solution of the $D$-equation in
$T(V)\ltimes_{r^*,0,0,l^*}V^*$.\end{proof}

\begin{remark}{\rm Roughly speaking, a symmetric
solution of $D$-equation corresponds to the symmetric part of an
${\mathcal O}$-operator, whereas an antisymmetric solution of
associative Yang-Baxter equation corresponds to the antisymmetric
part of an ${\mathcal O}$-operator.}\end{remark}

\begin{coro}
Let $(A,\succ,\prec)$ be a dendriform algebra. Then
$$r=\sum_{i=1}^n (e_i\otimes e_i^*+e_i^*\otimes e_i)\eqno (4.4.23)$$
is a symmetric solution of the $D$-equation in
$A\ltimes_{R_\prec^*,0,0,L_\succ^*} A^*$, where $\{e_1,\cdots,
e_n\}$ is a basis of $A$ and $\{e_1^*,\cdots, e_n^*\}$ is its dual
basis. Moreover, $r$ is nondegenerate and the induced 2-cocycle
${\mathcal B}$ on $A\ltimes_{R_\prec^*,0,0,L_\succ^*} A^*$ is given
by equation {\rm (1.1.1)}.
\end{coro}

\begin{proof}Let $V=A$, $l=L_\succ$, $r=R_\prec$ and $T=id$ in Theorem
4.4.13. Then the conclusion follows immediately. \end{proof}

\begin{remark} {\rm  Comparing with Theorem 4.3.6, we
show that (the non-symmetric) $T=\sum\limits_{i=1}^ne_i\otimes
e_i^*$ induces a dendriform D-bialgebra  structure on
$A\ltimes_{R^*,-L_\prec^*,-R_\succ^*,L^*} A^*$, whereas the above
(symmetric) $r=T+\sigma(T)$ induces a dendriform D-bialgebra
structure on $A\ltimes_{R_\prec^*,0,0,L_\succ^*} A^*$.}\end{remark}

Recall that two Connes cocycles $(A_1,\omega_1)$ and
$(A_2,\omega_2)$ are { isomorphic} if and only if there exists an
isomorphism of associative algebras $\varphi: A_1\rightarrow A_2$
such that
$$\omega_1(x,y)=\varphi^*{\omega}_2(x,y)
=\omega_2(\varphi(x),\varphi (y)),\;\;\forall x,y\in A_1.\eqno
(4.4.24)$$

By a similar proof as of Theorem 2.4.9, we have the following
conclusion.

\begin{theorem}
Let $(A,\succ,\prec)$ be a dendriform algebra. Then as Connes
cocycles of associative algebras, the double construction of Connes
cocycle (or the dendriform D-bialgebra) $(T(A)=A\bowtie A^*,\omega)$
given by a symmetric solution $r$ of $D$-equation in $A$ is
isomorphic to the double construction of Connes cocycle (or the
dendriform D-bialgebra) $(T(A)=A\ltimes_{R_\prec^*,L_\succ^*}A^*,
\omega)$, where $\omega$ is given by equation {\rm (1.4.1)}.
However, in general, they are not isomorphic as double constructions
of Connes cocycles (or dendriform D-bialgebras).\end{theorem}

\begin{coro} Let $(A,\succ, \prec)$ be a
dendriform algebra. Then as Connes cocycles of associative algebras,
the double constructions of Connes cocycles given by all symmetric
solutions of $D$-equation in $A$ are isomorphic to the double
construction of Connes cocycle
$(T(A)=A\ltimes_{R_\prec^*,L_\succ^*}A^*,\omega)$ given by the zero
solution.\end{coro}

\section{Comparison (duality) between bialgebra structures}

\subsection{Comparison (duality) between antisymmetric infinitesimal bialgebras and dendriform
D-bialgebras}

\mbox{}

The results in the previous sections allow us to compare
antisymmetric infinitesimal bialgebras and dendriform D-bialgebras
in terms of the following properties: 1-cocycles of associative
algebras, matched pairs of associative algebras, associative algebra
structures on the direct sum of the associative algebras in the
matched pairs, bilinear forms on the direct sum of the associative
algebras in the matched pairs, double structures on the direct sum
of the associative algebras in the matched pairs, algebraic
equations associated to coboundary cases, nondegenerate solutions,
${\mathcal O}$-operators of associative algebras and constructions
from dendriform algebras. We list the them in Table 1. From this
table, we observe that there is a clear analogy between them and in
particular, double constructions of Frobenius algebras correspond to
double constructions of Connes cocycles in this sense. Moreover, due
to the correspondences between certain symmetries and antisymmetries
appearing in the Table 1, we regard it as a kind of duality.

\begin{table}[t]\caption{Comparison between antisymmetric infinitesimal bialgebras and dendriform D-bialgebras}
\begin{tabular}{|c|c|c|}
\hline Algebras & Antisymmetric &Dendriform D-bialgebras\\
& infinitesimal bialgebras &
\\\hline 1-cocycles of associative algebras &
$( id\otimes   L, R \otimes id  )$ & $( id\otimes   L_\succ, R
\otimes id  )$,\\ && $( id\otimes   L, R_\prec \otimes id )$\\\hline
Matched pairs of associative & $(A,A^*,R_A^*,
L_A^*,R_{A^*}^*,L_{A^*}^*)$ & {\small$(A, A^*, R_{\prec_A}^*,
L_{\succ_A}^*, R_{\prec_{A^*}}^*, L_{\succ_{A^*}}^*)$}\\
 algebras&&\\\hline
Associative  algebra structures on   & double constructions & double constructions\\
the direct sum of the associative  & of Frobenius algebras& of Connes cocycles\\
algebras in the matched pairs&&\\\hline Bilinear forms on  &
symmetric & antisymmetric
\\\cline{2-3} the direct sum of the associative & $\langle x+a^*,y+b^*\rangle  $ & $\langle x+a^*,y+b^*\rangle  $\\algebras in the
matched pairs &$=\langle x,b^*\rangle  +\langle a^*,y\rangle  $ &
$=-\langle x,b^*\rangle  +\langle a^*,y\rangle  $\\\cline{2-3}
&invariant & Connes cocycles\\\hline Double structures on  &
associative doubles & dendriform doubles\\ the direct sum of the
associative&&\\algebras in the matched pairs&&\\\hline Algebraic
equations associated & antisymmetric solutions & symmetric solutions
\\\cline{2-3}
to coboundary cases& associative Yang-Baxter &$D$-equations in
dendriform\\&equations& algebras\\\hline Nondegenerate
solutions & Connes cocycles of   & 2-cocycles of dendriform\\
&associative algebras&algebras\\\hline ${\mathcal O}$-operators of
associative algebras & associated to $(R^*,L^*)$ & associated to
$(R_\prec^*,L_\succ^*)$
\\\cline{2-3} &antisymmetric parts & symmetric parts\\\hline
Constructions from & $r=\sum\limits_{i=1}^n (e_i\otimes
e_i^*-e_i^*\otimes e_i)$ & $r=\sum\limits_{i=1}^n (e_i\otimes
e_i^*+e_i^*\otimes e_i)$\\\cline{2-3} dendriform  algebras &
induced bilinear
forms & induced bilinear forms\\
& $\langle x+a^*,y+b^*\rangle  $ & $\langle x+a^*,y+b^*\rangle  $\\
&$=-\langle x,b^*\rangle  +\langle a^*,y\rangle  $ & $=\langle
x,b^*\rangle  +\langle a^*,y\rangle  $\\\hline
\end{tabular}
\end{table}

Next we consider the case that a dendriform D-bialgebra is also an
antisymmetric infinitesimal bialgebra.

\begin{theorem} Let
$(A,A^*,\Delta_\succ,\Delta_\prec,\beta_\succ,\beta_\prec)$ be a
dendriform D-bialgebra.  Then $(A,A^*)$ is an antisymmetric
infinitesimal bialgebra if and only if the following two equations
hold:
$$\langle L^*_{\prec_{A^*}}(b^*)y, L^*_{\prec_A}(x)a^*\rangle
=\langle R^*_{\succ_{A^*}}(a^*)x, R^*_{\succ_A}(y)b^*\rangle;\eqno
(5.1.1)$$ $$\langle L^*_{\prec_{A^*}}(b^*)y,
R^*_{\succ_A}(x)a^*\rangle+ \langle L^*_{\prec_{A^*}}(a^*)x,
R^*_{\succ_A}(y)b^*\rangle $$$$=\langle R^*_{\succ_{A^*}}(b^*)x,
L^*_{\prec_A}(y)a^*\rangle+ \langle R^*_{\succ_{A^*}}(a^*)y,
L^*_{\prec_A}(x)b^*\rangle,\eqno (5.1.2)$$ for any $x,y\in
A^*,a^*,b^*\in A^*$.\end{theorem}

\begin{proof} The conclusion can be obtained by a similar proof as
of Proposition 2.2.2. \end{proof}

\begin{coro}
Let $(A,\succ,\prec)$ be a dendriform algebra and $r\in A\otimes A$
be a symmetric solution of $D$-equation in $A$. Suppose the
dendriform algebra structure on $A^*$ is induced by $r$ from
equation {\rm (4.4.11)}. Then $(A, A^*)$ is an antisymmetric
infinitesimal bialgebra if and only if the following two equations
hold:
$$\langle y\prec_A(x\succ_A
r(a^*))-y*_Ar(R_{\succ_A}^*(x)a^*), b^*\rangle =\langle
r(L_{\prec_A}^*(y)b^*)*_Ax-(r(b^*)\prec_A y)\succ x,
a^*\rangle;\eqno (5.1.3)$$
$$\langle y\prec_A(r(a^*)\prec_A x)-(y\succ_A r(a^*))\succ_A x
+r(R_{\succ_A}^*(y)a^*)*_Ax-y*_Ar(L_{\prec_A}^*(x)a^*),b^*\rangle$$
$$=\langle -x\prec_A(r(b^*)\prec_A y)+(x\succ_A r(b^*))\succ_A y
-r(R_{\succ_A}^*(x)a^*)*_Ay+x*_Ar(L_{\prec_A}^*(y)a^*),a^*\rangle,\eqno(5.1.4)$$
for any $x,y\in A$ and $a^*\in A^*$.\end{coro}

\begin{coro}
 Let $(A,A^*,\Delta_\succ,\Delta_\prec,\beta_\succ,\beta_\prec)$
be a dendriform D-bialgebra.  If equations {\rm (5.1.1-5.1.2)} are
satisfied, then there are two associative algebra structures
$A\bowtie_{R_{\prec_{A^*}}^*,L_{\succ_{A^*}}^*}^{R_{\prec_A}^*,
L_{\succ_A}^*}A^*$ and
$A\bowtie^{R_{A}^*,L_{A}^*}_{R_{A^*}^*,L_{A^*}^*}A^*$ on the direct
sum $A\oplus A^*$ of the underlying vector spaces of $A$ and $A^*$
such that both $A$ and $A^*$ are associative subalgebras and the
bilinear form given by equation {\rm (1.4.1)} is a Connes cocycle on
$A\bowtie_{R_{\prec_{A^*}}^*,L_{\succ_{A^*}}^*}^{R_{\prec_A}^*,
L_{\succ_A}^*}A^*$ and the bilinear form given by equation {\rm
(1.1.1)} is invariant on
$A\bowtie^{R_{A}^*,L_{A}^*}_{R_{A^*}^*,L_{A^*}^*}A^*$. Moreover,
These two associative algebras are not isomorphic in general.
\end{coro}

\begin{exam} {\rm Let $(A,*_A)$ be an associative algebra and
$\omega$ be a Connes cocycle on $(A,*_A)$. Then there is an
antisymmetric infinitesimal bialgebra whose associative algebra
structure on $A^*$ is given by a nondegenerate solution $r$ of
associative Yang-Baxter equation as follows.
$$\Delta(x)= ( id\otimes   L(x)-R(x) \otimes id  ) r,\;\;\forall x\in A,\eqno (5.1.5)$$
where $r:A^*\rightarrow A$ is given by $\omega(x,y)=\langle
r^{-1}(x),y\rangle$. On the other hand, there exists a compatible
dendriform algebra structure ``$\succ_A,\prec_A$'' on $A$ given
by equation (4.1.1), that is,
$$\omega(x\succ_A y,z)=\omega(y,z*_Ax),\;\;\omega(x\prec_A y,z)=\omega(x,y*_Az),\;\;\forall\; x,y,z\in
A.\eqno (5.1.6)$$ Moreover, there exists a compatible dendriform
algebra structure on the associative algebra $A^*$ given by
$$a^*\succ_{A^*} b^*=r^{-1}(r(a^*)\succ_A r(b^*)),
a^*\prec_{A^*} b^*=r^{-1}(r(a^*)\prec_A r(b^*)),\;\; \forall\;
a^*,b^*\in A.\eqno (5.1.7)$$ Furthermore, it is easy to show that
$$L_{\succ_A}^*(x)a^*=r^{-1}(r(a^*)*_Ax),\;
R_{\succ_A}^*(x)a^*=-r^{-1}(x\prec_A r(a^*)),\;
L_{\prec_A}^*(x)a^*=-r^{-1}(r(a^*)\succ_Ax),$$
$$R_{\prec_A}^*(x)a^*=r^{-1}(x* r(a^*)),\;
L_{\succ_{A^*}}^*(a^*)x=x*_Ar(a^*),\;
R_{\succ_{A^*}}^*(a^*)x=-r(a^*)\prec_A x,$$$$
L_{\prec_{A^*}}^*(a^*)x=-x\succ_Ar(a^*),\;
R_{\prec_{A^*}}^*(a^*)x=r(a^*)*_A x,\;\;\forall\;x\in A, a^*\in A^*.
\eqno (5.1.8)$$ Therefore according to Theorem 4.2.4, $(A,A^*)$ (as
dendriform  algebras) is a dendriform D-bialgebra  if and only if
$(A,A^*,R_{\prec_A}^*,L_{\succ_A}^*,R_{\prec_{A^*}}^*,L_{\succ_{A^*}}^*)$
a matched pair of associative algebras, if and only if $A$ is 2-step
nilpotent, that is, $x*_Ay*_Az=0$ for any $x,y,z\in A$. In this
case, by equation (5.1.6), we show that it is equivalent to
$$x\succ_A(y\succ_Az)=x\prec_A(y\prec_Az)=x\succ_A(y\prec_Az)=0,\;\;\forall\;x,y,z\in
A.\eqno (5.1.9)$$  Therefore, under such conditions, equations
(5.1.1-5.1.2) hold naturally.}\end{exam}

\subsection{Duality in the version of Lie algebras: Lie bialgebras and pre-Lie
bialgebras}

\mbox{}

There is a similar duality in the version of Lie algebras which was
given in \cite{Bai2}. In order to be self-contained, we give a brief
introduction in this subsection. We would like to point out that,
although we give the Lie bialgebras and pre-Lie bialgebras as the
similar structures of antisymmetric infinitesimal bialgebras and
dendriform D-bialgebras here, in fact, it is the Manin triples (Lie
bialgebras) that have been first studied  and then motivate us to
study the other structures.

There are two kinds of important (nondegenerate) bilinear forms on
Lie algebras as follows. A bilinear form ${\mathcal B}(\;,\;)$ on a
Lie algebra $A$ is invariant if
$${\mathcal B}([x,y],z)={\mathcal B}(x,[y,z]),\;\;\forall\; x,y\in A.\eqno (5.2.1)$$
A 2-cocycle (symplectic form) on a Lie algebra $A$ is an
antisymmetric bilinear from $\omega$ satisfying
$$\omega ([x,y],z)+\omega ([y,z],x)+\omega([z,x],y)=0,\;\;\forall\; x,y,z\in
A.\eqno (5.2.2)$$

Moreover, the algebras play a similar role of dendriform algebras in
the double constructions of Frobenius algebras and Connes cocycles
are pre-Lie algebras. In fact, pre-Lie algebras (or under other
names like left-symmetric algebras, quasi-associative algebras,
Vinery algebras and so on) are a class of natural algebraic systems
appearing in many fields in mathematics and mathematical physics
(see a survey article \cite{Bu} and the references therein).

\begin{defn} {\rm Let $A$ be a vector space over a field
${\bf F}$ with a bilinear product $(x,y)\rightarrow xy$. $A$ is
called a {\it  pre-Lie algebra} if $$(xy)z-x(yz)=(yx)z-y(xz),
\;\;\forall x,y,z\in A.\eqno (5.2.3)$$}\end{defn}

Let $A$ be a pre-Lie algebra. For any $x,y\in A$, let $L(x)$ and
$R(x)$ denote the left and right multiplication operator
respectively, that is, $L(x)(y)=xy,\;R(x)(y)=yx$. Let
$L:A\rightarrow \frak g\frak l (A)$ with $x\rightarrow L(x)$ and $R:A\rightarrow
\frak g\frak l (A)$ with $x\rightarrow R(x)$ (for every $x\in A$) be two linear
maps. For a Lie algebra ${\mathcal G}$, we let ${\rm ad}(x)$ denote
the adjoint operator, that is, ${\rm ad}(x)y=[x,y]$, and ${\rm
ad}:{\mathcal G}\rightarrow \frak g\frak l({\mathcal G})$ with $x\rightarrow
{\rm ad}(x)$ be a linear map.

\begin{prop}  Let $A$ be a pre-Lie algebra.

{\rm (1)} The commutator
$$[x,y]=xy-yx,\;\;\forall\; x,y\in A,\eqno (5.2.4)$$
defines a Lie algebra ${\mathcal G}(A)$, which is called the
sub-adjacent Lie algebra of $A$ and $A$ is also called a
compatible pre-Lie algebra structure on the Lie algebra ${\mathcal
G}(A)$.

{\rm (2)} The map $L:A\rightarrow \frak g\frak l (A)$ gives a representation of
the Lie algebra ${\mathcal G}(A)$.\end{prop}

\begin{prop}{\rm (\cite{Chu})}\quad Let ${\mathcal G}$ be a Lie algebra and
$\omega$ be a nondegenerate 2-cocycle on ${\mathcal G}$ (such a Lie
algebra called a symplectic Lie algebra). Then there exists a
compatible pre-Lie algebra structure on ${\mathcal G}$ defined by
$$\omega(x*y,z)=-\omega(y, [x,z]),\;\;\forall\; x,y,z\in {\mathcal G}.\eqno
(5.2.5)$$
\end{prop}

Next we give the ``double constructions" of Lie algebras with
nondegenerate invariant bilinear forms or nondegenerate 2-cocycles.
In fact, both of them have their own (independent) interests in many
fields.

At first, recall that $({\mathcal G},{\mathcal H}, \rho,\mu)$ is { a
matched pair of Lie algebras} if ${\mathcal G}$ and ${\mathcal H}$
are Lie algebras and $\rho:{\mathcal G}\rightarrow \frak g\frak l({\mathcal H})$
and $\mu:{\mathcal H}\rightarrow \frak g\frak l({\mathcal G})$ are
representations satisfying
$$\rho(x)[a,b]-[\rho(x)a,b]-[a,\rho(x)b]+\rho(\mu(a)x)b-\rho(\mu(b)x)a=0;\eqno
(5.2.6)$$
$$\mu(a)[x,y]-[\mu(a)x,y]-[x,\mu(a)y]+\mu(\rho(x)a)y-\mu(\rho(y)a)x=0,\eqno
(5.2.7)$$ for any $x,y\in {\mathcal G}$ and $a,b\in {\mathcal H}$.
In this case, there exists a Lie algebra structure on the direct sum
${\mathcal G}\oplus {\mathcal H}$ of the underlying vector spaces of
${\mathcal G}$ and ${\mathcal H}$ given by
$$[x+a,y+b]=[x,y]+\mu(a)y-\mu(b)x+[a,b]+\rho(x)b-\rho(y)a,\;\;\forall
x,y\in {\mathcal G},a,b\in {\mathcal H}.\eqno (5.2.8)$$ We denote it
by ${\mathcal G}\bowtie^\rho_\mu {\mathcal H}$ or simply ${\mathcal
G}\bowtie {\mathcal H}$. Moreover, every Lie algebra which is the
direct sum of the underlying vector spaces of two subalgebras can be
obtained from a matched pair of Lie algebras as above.

\begin{defn}{\rm Let ${\mathcal G}$ be a Lie algebra. Suppose that there is a
Lie algebra structure on the direct sum of the underlying vector
spaces of ${\mathcal G}$ and its dual space ${\mathcal G}^*$ such
that ${\mathcal G}$ and ${\mathcal G}^*$ are Lie subalgebras.

 (a) If the natural symmetric bilinear form on ${\mathcal G}\oplus {\mathcal G}^*$
given by equation (1.1.1) is invariant, then $({\mathcal G}\bowtie
{\mathcal G}^*, {\mathcal G}, {\mathcal G}^*)$ is called a {\it
(standard) Manin triple}.

(b) If the natural antisymmetric bilinear form on ${\mathcal
G}\oplus {\mathcal G}^*$ given by equation (1.4.1) is a 2-cocycle,
then it is called a {\it phase space of the Lie algebra ${\mathcal
G}$} (\cite{Ku1}). $({\mathcal G}\bowtie {\mathcal G}^*, {\mathcal
G}, {\mathcal G}^*)$ is also called {\it a parak\"ahler structure on
the Lie algebra ${\mathcal G}\bowtie {\mathcal G}^*$}
(\cite{Kan}).}\end{defn}

For a Lie algebra ${\mathcal G}$ and a representation $(\rho, V)$ of
${\mathcal G}$, recall that {a 1-cocycle} $T$ associated to $\rho$
(denoted by $(\rho, T)$) is a linear map from ${\mathcal G}$ to $V$
satisfying
$$T([x,y])=\rho(x)T(y)-\rho(y)T(x),\;\;\forall\; x,y\in {\mathcal G}.\eqno (5.2.9)$$

\begin{defn}{\rm (a) Let ${\mathcal G}$ be a Lie algebra. A {\it  Lie
bialgebra} structure on ${\mathcal G}$ is an antisymmetric linear
map $\delta: {\mathcal G}\rightarrow {\mathcal G}\otimes {\mathcal
G}$ such that $\delta^*:{\mathcal G}^*\otimes {\mathcal
G}^*\rightarrow {\mathcal G}^*$ is a Lie bracket on ${\mathcal G}^*$
and $\delta$ is a 1-cocycle of ${\mathcal G}$ associated to ${\rm
ad} \otimes id  + id\otimes {\rm ad}$ with values in ${\mathcal
G}\otimes {\mathcal G}$. We denote it by $({\mathcal G},{\mathcal
G}^*)$ or $({\mathcal G},\delta)$.

(b) Let $A$ be a vector space. A {\it pre-Lie bialgebra} structure
on $A$ is a pair of linear maps $(\Delta,\beta)$ such that
$\Delta:A\rightarrow A\otimes A,\beta: A^*\rightarrow A^*\otimes
A^*$ and

(1) $\Delta^*:A^*\otimes A^*\rightarrow A^*$ defines a pre-Lie
algebra structure on $A^*$;

(2) $\beta^*:A\otimes A\rightarrow A$ defines a pre-Lie algebra
structure on $A$;

(3) $\Delta$ is a 1-cocycle of ${\mathcal G}(A)$ associated to $L
\otimes id  + id\otimes   {\rm ad}$ with values in $A\otimes A$;

(4) $\beta$ is a 1-cocycle of ${\mathcal G}(A^*)$ associated to $L
\otimes id  + id\otimes   {\rm ad}$ with values in $A^*\otimes
A^*$.

\noindent We denote it by $(A,A^*,\Delta,\beta)$ or simply
$(A,A^*)$.}
\end{defn}

\begin{theorem} {\rm (a)} Let $({\mathcal G},[\;,\;]_{\mathcal G})$
and $({\mathcal G}^*, [\;,\;]_{{\mathcal G}^*})$ be two Lie
algebras. Then the following conditions are equivalent:

{\rm (1)} $({\mathcal G}\bowtie {\mathcal G}^*, {\mathcal
G},{\mathcal G}^*)$ is a standard Manin triple with the bilinear
form {\rm (1.1.1)};

{\rm (2)} $({\mathcal G}, {\mathcal G}^*, {\rm ad}^*_{\mathcal G},
{\rm ad}^*_{{\mathcal G}^*})$ is a matched pair of Lie algebras;

{\rm (3)} $({\mathcal G},{\mathcal G}^*)$ is a Lie bialgebra.

{\rm (b)} Let $(A,\cdot)$ and $(A^*,\circ)$ be two pre-Lie algebras.
Then the following conditions are equivalent:

{\rm (1)} $({\mathcal G}(A)\bowtie {\mathcal G}(A)^*, {\mathcal
G}(A),{\mathcal G}(A^*))$ is a parak\"ahler Lie algebra with the
bilinear form {\rm (1.4.1)}.

{\rm (2)}  $({\mathcal G}(A),{\mathcal G}(A^*),L^*_\cdot,L^*_\circ)$
is a matched pair of Lie algebras;

{\rm (3)} $(A,A^*)$ is a pre-Lie bialgebra.
\end{theorem}

In fact, a Lie bialgebra is the Lie algebra ${\mathcal G}$ of a
Poisson-Lie group $G$ equipped with additional structures induced
from the Poisson structure on $G$ and a Poisson-Lie group is a Lie
group with a Poisson structure compatible with the group operation
in a certain sense. Poisson-Lie groups play an important role in
symplectic geometry and quantum group theory (cf. \cite{D} and the
references therein). On the other hand, in geometry, a parak\"ahler
manifold is a symplectic manifold with a pair of transversal
Lagrangian foliations (\cite{Li}). A parak\"ahler Lie algebra
${\mathcal G}$ is the Lie algebra of a Lie group $G$ with a
$G$-invariant parak\"ahler structure (\cite{Kan}).

We have already obtained many properties of Lie bialgebras and
pre-Lie algebras which are similar to our study in the previous
sections. We put them in the Appendix and we compare pre-Lie
bialgebras and Lie bialgebras in terms of their certain properties
in Table 2. From Table 2, we observe that there is also a clear
analogy between them and in particular, due to the correspondences
between certain symmetries and antisymmetries appearing in the
analogy, we can regard it as a kind of duality again which is
similar to the duality appearing in the Table 1.

\begin{table}[t]\caption{Comparison between Lie bialgebras and pre-Lie bialgebras}
\begin{tabular}{|c|c|c|}
\hline Algebras & Lie bialgebras & Pre-Lie bialgebras\\\hline
Corresponding Lie groups & Poisson-Lie groups & parak\"ahler Lie
groups
\\\hline
1-cocycles of Lie algebras & $ id\otimes   {\rm ad}+{\rm ad}
\otimes id  $ & $L \otimes id  + id\otimes   {\rm ad}$\\\hline
Matched pairs of Lie algebras & $({\mathcal G},{\mathcal G}^*,
{\rm ad}^*_{\mathcal G}, {\rm ad}^*_{{\mathcal G}^*})$ &
$({\mathcal G}(A), {\mathcal G}(A^*), L^*_{A}, L^*_{A^*})$\\\hline
Lie algebra  structures on   & Manin triples & phase spaces\\
the direct sum of the Lie  &&\\ algebras in the matched
pairs&&\\\hline Bilinear forms on  & symmetric & antisymmetric
\\\cline{2-3} the direct sum of the Lie & $\langle x+a^*,y+b^*\rangle  $ & $\langle x+a^*,y+b^*\rangle  $\\algebras in the
matched pairs &$=\langle x,b^*\rangle  +\langle a^*,y\rangle  $ &
$=-\langle x,b^*\rangle  +\langle a^*,y\rangle  $\\\cline{2-3}
&invariant & 2-cocycles\\\hline Double structures on  & Drinfeld
doubles & symplectic doubles\\ the direct sum of the Lie&&\\algebras
in the matched pairs&&\\\hline Algebraic equations associated &
antisymmetric solutions & symmetric solutions
\\\cline{2-3}
to coboundary cases& classical Yang-Baxter &$S$-equations in
pre-Lie\\&equations in Lie algebras&algebras\\\hline Nondegenerate
solutions & 2-cocycles of Lie algebras & 2-cocycles of pre-Lie\\
&&algebras\\\cline{2-3} & symplectic structures & Hessian
structures\\\hline ${\mathcal O}$-operators of Lie algebras &
associated to ${\rm ad}^*$ & associated to $L^*$
\\\cline{2-3} &antisymmetric parts & symmetric parts\\\hline
Constructions from & $r=\sum\limits_{i=1}^n (e_i\otimes
e_i^*-e_i^*\otimes e_i)$ & $r=\sum\limits_{i=1}^n (e_i\otimes
e_i^*+e_i^*\otimes e_i)$\\\cline{2-3} pre-Lie algebras & induced
bilinear
forms & induced bilinear forms\\
& $\langle x+a^*,y+b^*\rangle  $ & $\langle x+a^*,y+b^*\rangle  $\\
&$=-\langle x,b^*\rangle  +\langle a^*,y\rangle  $ & $=\langle
x,b^*\rangle  +\langle a^*,y\rangle  $\\\hline
\end{tabular}
\end{table}


\subsection{Relationships among four bialgebras}

\begin{prop} {\rm (\cite{Cha1}, \cite{A2})} Let $(A,\succ,\prec)$ be a dendriform algebra.
Then there is a pre-Lie algebra structure on $(A,\cdot)$ given by
$$x\cdot y =x\succ y-y\prec x,\;\;\forall x,y\in A.\eqno
(5.3.1)$$\end{prop}

\begin{coro} Let $(A,\succ,\prec)$ be a dendriform algebra. Then the
sub-adjacent Lie algebra of the pre-Lie algebra $(A,\cdot)$ given by
equation {\rm (5.3.1)} is as the same as the commutator Lie algebra
of the associated associative algebra $(A,*)$, that is,
$$[x,y]=x*y-y*x=x\cdot y-y\cdot x=x\succ y+x\prec y-y\succ x-y\prec x,\;\;\forall x,y\in A.\eqno (5.3.2)$$
\end{coro}

Therefore, as Chapoton pointed out in \cite{Cha1} (also see
\cite{A2}, \cite{A4}, \cite{EMP}), there is the following
commutative diagram of categories.
$$\begin{matrix} {\rm dendriform\quad algebras} & \longrightarrow &
\mbox{pre-Lie algebras} \cr \downarrow & &\downarrow\cr {\rm
associative\quad algebras} & \longrightarrow & {\rm Lie\quad
algebras} \cr\end{matrix}$$ In this diagram, the left vertical arrow
is given by equation (3.1.2), the top horizontal arrow is given by
equation (5.3.1), the bottom arrow is given by equation (5.2.4)
since an associative algebra is a special pre-Lie algebra and the
right vertical arrow is given by equation (5.2.4).

Obviously, if a symmetric or antisymmetric bilinear form on an
associative algebra is invariant or a Connes cocycle respectively,
then it is also invariant or a 2-cocycle on the commutator Lie
algebra respectively,.

\begin{theorem} {\rm (1)} A double construction of Frobenius algebra gives a
standard Manin triple (on the commutator Lie algebra) naturally.

{\rm (2)} A double construction of Connes cocycles gives a
parak\"ahler Lie algebra (on the commutator Lie algebra) naturally.
\end{theorem}

\begin{coro} {\rm (1)} Any antisymmetric infinitesimal bialgebra is a Lie
bialgebra (in the sense of its commutator Lie algebra).

{\rm (2)} Any dendriform D-bialgebra  is a pre-Lie bialgebra (in the
sense of equation (5.3.1)).
\end{coro}

\begin{coro} We have the following relationship among the antisymmetric
infinitesimal bialgebras, dendriform algebras, Lie bialgebras and
pre-Lie bialgebras.
$$\begin{matrix} \mbox{\rm dendriform D-bialgebras} & \hookrightarrow & \mbox{\rm pre-Lie bialgebras} \cr
\updownarrow {\rm dual} & &\updownarrow {\rm dual}\cr \mbox{\rm
antisymmetric infinitesimal bialgebras} & \hookrightarrow &
\mbox{\rm Lie bialgebras} \cr\end{matrix}$$

\noindent  The $\updownarrow$ means that the duality given in
subsections {\rm 5.1} and {\rm 5.2} and the $\hookrightarrow$ means
the inclusion in the sense of {\rm Corollary 5.3.4}.
\end{coro}

\begin{remark}{\rm The conclusion (1) in Corollary 5.3.4 and the
relation given by
 the bottom $\hookrightarrow$ in the above diagram were also pointed out in
\cite{A3}.}\end{remark}

\begin{coro} Let
$(A,A^*,\Delta_\succ,\Delta_\prec,\beta_\succ,\beta_\prec)$ be a
dendriform D-bialgebra.  If equations {\rm (5.1.1-5.1.2)} hold, then
$(A,A^*)$ is an antisymmetric infinitesimal bialgebra. $(A,A^*)$ is
also a pre-Lie bialgebra in the sense of equation {\rm (5.3.1)}.
Furthermore, as the commutator Lie algebras, $({\mathcal
G}(A),{\mathcal G}(A)^*)$ is a Lie bialgebra. Therefore, there is an
associative algebra structure and a Lie algebra structure on the
direct sum $A\oplus A^*$ of the underlying space of $A$ and $A^*$
such that the natural symmetric bilinear form given by equation {\rm
(1.1.1)} is invariant on both of them and the natural antisymmetric
bilinear form given by equation {\rm (1.4.1)} is a Connes cocycle on
the associative algebra and a 2-cocycle on the Lie algebra.
Moreover, under such a condition, there is the following commutative
diagram in an obvious way.
$$\begin{matrix} \mbox{\rm dendriform D-bialgebras} & \hookrightarrow & \mbox{\rm pre-Lie bialgebras}
 \cr \downarrow  & &\downarrow \cr
 \mbox{\rm
antisymmetric infinitesimal bialgebras} & \hookrightarrow &
\mbox{\rm Lie bialgebras} \cr\end{matrix}$$
\end{coro}

\section*{Acknowledgments}

The author thanks Professors M. Aguiar, A. Connes, L. Guo and J.-L.
Loday for important suggestion. This work was supported in part by NSFC (10621101, 10920161), NKBRPC (2006CB805905) and SRFDP (200800550015).

\section*{Appendix: Some properties of Lie bialgebras and pre-Lie bialgebras}

In this appendix, we list some properties of Lie bialgebras and
pre-Lie bialgebras. Most of the results can be found in \cite{Bai2}
and the references therein.

\noindent {\bf Proposition A1.}\quad {\it {\rm (a)}\quad Let
$({\mathcal G}, {\mathcal G}^*)$ be a Lie bialgebra. Then there is a
canonical Lie bialgebra structure on ${\mathcal G}\oplus {\mathcal
G}^*$ such that the inclusions $i_1:{\mathcal G}\rightarrow
{\mathcal G}\oplus {\mathcal G}^*$ and $i_2:{\mathcal
G}^*\rightarrow {\mathcal G}\oplus {\mathcal G}^*$ into the two
summands are homomorphisms of Lie bialgebras, where the Lie
bialgebra structure on ${\mathcal G}^*$ is given by
$-\delta_{{\mathcal G}^*}$. Such a structure is called a classical
(Drinfeld) double of ${\mathcal G}$.

{\rm (b)}\quad Let $(A,A^*,\Delta,\beta)$ be a pre-Lie bialgebra.
Then there is a canonical pre-Lie bialgebra structure on $A\oplus
A^*$ such that both the inclusions $i_1:A\rightarrow A\oplus A^*$
and $i_2:A^*\rightarrow A\oplus A^*$ into the two summands are
homomorphisms of pre-Lie bialgebras. Such a structure is called a
symplectic double of $A$.}

\noindent {\bf Definition A2.}\quad (a)\quad A Lie bialgebra
$({\mathcal G}, \delta)$ is called {\it coboundary} if $\delta$ is a
1-coboundary of ${\mathcal G}$ associated to ${\rm ad} \otimes id  +
id\otimes   {\rm ad}$, that is, there exists a $r\in {\mathcal
G}\otimes {\mathcal G}$ such that
$$\delta(x)=({\rm ad}(x)\otimes id+id\otimes {\rm
ad}(x))r,\;\;\forall x\in {\mathcal G}.\eqno ({\rm A1})$$

(b) \quad A pre-Lie bialgebra $(A,A^*,\Delta,\beta)$ is called {\it
coboundary} if $\Delta$ is a 1-coboundary of ${\mathcal G}(A)$
associated to $L \otimes id  + id\otimes   {\rm ad}$, that is, there
exists a $r\in A\otimes A$ such that
$$\Delta (x)=(L(x) \otimes id  + id\otimes   {\rm ad}x)r,\;\;\forall x\in A.\eqno ({\rm A2})$$

\noindent {\bf Theorem A3.}\quad (a) {\it \quad Let ${\mathcal G}$
be a Lie algebra and $r\in {\mathcal G}\otimes {\mathcal G}$. Then
the map $\delta: {\mathcal G}\rightarrow {\mathcal G}\otimes
{\mathcal G}$ defined by equation {\rm (A1)} induces a Lie bialgebra
structure on ${\mathcal G}$ if and only if the following two
conditions are satisfied (for any $x\in {\mathcal G}$):

{\rm (1)} $({\rm ad}(x)\otimes id+id\otimes {\rm
ad}(x))(r+\sigma(r))=0$;

{\rm (2)} $({\rm ad}(x)\otimes id\otimes id +id\otimes {\rm
ad}(x)\otimes id+id\otimes id\otimes {\rm ad}(x))(
[r_{12},r_{13}]+[r_{12},r_{23}]+[r_{13},r_{23}])=0$.

{\rm (b)} \quad Let $A$ be a pre-Lie algebra and $r\in A\otimes A$.
Then the map $\Delta$ defined by equation {\rm (A2)} induces a
pre-Lie algebra structure on $A^*$ such that $(A,A^*)$ is a pre-Lie
bialgebra if and only if the following two conditions are satisfied
(for any $x,y\in A$):

{\rm (1)} $[P(x\cdot y)-P(x)P(y)](r-\sigma(r))=0$;

{\rm (2)} $Q(x)[[r,r]]=0$,

\noindent where  $Q(x)=L(x)\otimes  id\otimes   id+ id\otimes L(x)
\otimes id  + id\otimes    id\otimes   {\rm ad}x$, $P(x)=L(x)
\otimes id  + id\otimes   L(x)$  and
$$[[r,r]]=r_{13}\cdot r_{12}-r_{23}\cdot r_{21}+[r_{23},r_{12}]-
[r_{13},r_{21}]-[r_{13},r_{23}].\eqno ({\rm A3})$$}

\noindent {\bf Corollary A4.}\quad {\it {\rm (a)}\quad Let
${\mathcal G}$ be a Lie algebra and $r\in {\mathcal G}\otimes
{\mathcal G}$. If $r$ is antisymmetric and $r$ satisfies
$$[r_{12},r_{13}]+[r_{12},r_{23}]+[r_{13},r_{23}]=0,\eqno ({\rm A4})$$
then the map $\delta: {\mathcal G}\rightarrow {\mathcal G}\otimes
{\mathcal G}$ defined by equation {\rm (A1)} induces a Lie bialgebra
structure on ${\mathcal G}$.

{\rm (b)}\quad Let $A$ be a pre-Lie algebra and $r\in A\otimes A$.
Suppose that $r$ is symmetric. Then the map $\Delta$ defined by
equation {\rm (A2)} induces a pre-Lie algebra structure on $A^*$
such that $(A,A^*)$ is a pre-Lie bialgebra if
$$-r_{12}\cdot r_{13}+r_{12}\cdot r_{23}+[r_{13},r_{23}]=0.\eqno ({\rm
A5})$$}

\noindent {\bf Definition A5.}\quad (a)\quad Let ${\mathcal G}$ be a
Lie algebra and $r\in {\mathcal G}\otimes {\mathcal G}$. Equation
(A4) is called {\it classical Yang-Baxter equation in ${\mathcal
G}$.}

(b)\quad Let $A$ be a pre-Lie algebra and $r\in A\otimes A$.
Equation (A5) is called {\it ${S}$-equation in $A$.}

Let ${\mathcal G}$ be a Lie algebra and $\rho:{\mathcal
G}\rightarrow \frak g\frak l (V)$ be its representation. Recall that a linear map
$T:V\rightarrow {\mathcal G}$ is called an ${\mathcal O}$-operator
of $\mathcal G$ associated to $\rho$ if $T$ satisfies
$$[T(u), T(v)]=T(\rho(T(u))v-\rho(T(v))u),\forall u,v\in V.\eqno ({\rm A6})$$

\noindent {\bf Proposition A6.}\quad  (a) {\it\quad Let ${\mathcal
G}$ be a Lie algebra and $r\in {\mathcal G}\otimes {\mathcal G}$.

{\rm (1)} Suppose $r$ is antisymmetric and nondegenerate. Then $r$
is a solution of classical Yang-Baxter equation in ${\mathcal G}$ if
and only if the isomorphism ${\mathcal G}^*\rightarrow {\mathcal G}$
induced by $r$, regarded as a bilinear form on ${\mathcal G}$ is a
2-cocycle on ${\mathcal G}$.

{\rm (2)} Suppose $r$ is antisymmetric. Then $r$ is a solution of
classical Yang-Baxter equation in ${\mathcal G}$ if and only if $r$
is an ${\mathcal O}$-operator of $\mathcal G$ associated to ${\rm
ad}^*$, that is, $r$ satisfies
$$[r(a^*),r(b^*)]=r({\rm ad}^*(r(a^*))b^*-{\rm
ad}^*(r(b^*))a^*),\;\;\forall a^*,b^*\in {\mathcal G}^*.\eqno ({\rm
A7})$$

{\rm (b)}\quad Let $A$ be a pre-Lie algebra and $r\in A\otimes A$.

{\rm (1)} Suppose that $r$ is symmetric and nondegenerate. Then $r$
is a solution of $S$-equation in $A$ if and only if the inverse of
the isomorphism $A^*\rightarrow A$ induced by $r$, regarded as a
bilinear form ${\mathcal B}$ on $A$ is a 2-cocycle on $A$ (see
equation {\rm (4.4.17)}).

{\rm (2)} Suppose that $r$ is symmetric. Then $r$ is a solution of
$S$-equation in $A$ if and only if $r$ is an ${\mathcal O}$-operator
of $\mathcal G(A)$ associated to $L^*$, that is, $r$ satisfies
$$[r(a^*),r(b^*)]=r(L_\cdot^*(r(a^*))b^*-L_\cdot^*(r(b^*))a^*),\;\;\forall a^*,b^*\in A^*.
\eqno ({\rm A8})$$}

\noindent {\bf Lemma A7.}\quad {\it Let ${\mathcal G}$ be a Lie
algebra and $\rho:{\mathcal G}\rightarrow \frak g\frak l (V)$ be a
representation.  Let $T:V\rightarrow {\mathcal G}$ be an ${\mathcal
O}$-operator associated to $\rho$. Then the product
$$u\circ v=\rho(T(u))v,\;\;\forall u,v\in V\eqno ({\rm A9})$$
defines a pre-Lie algebra structure on $V$.  Therefore $V$ is a Lie
algebra as the sub-adjacent Lie algebra of this pre-Lie algebra and
$T$ is a homomorphism of Lie algebras. Furthermore,
$T(V)=\{T(v)|v\in V\}\subset {\mathcal G}$ is a Lie subalgebra of
${\mathcal G}$ and there is an induced pre-Lie algebra structure on
$T(V)$ given by
$$T(u)\cdot T(v)=T(u\circ v)=T(\rho (T(u))v),\;\;\forall u,v\in V.\eqno ({\rm A10})$$
Moreover, its sub-adjacent Lie algebra structure is just the Lie
subalgebra structure of ${\mathcal G}$ and $T$ is a homomorphism of
pre-Lie algebras.}

\noindent {\bf Proposition A8.} \quad  {\it Let ${\mathcal G}$ be a
Lie algebra and $\rho:{\mathcal G}\rightarrow \frak g\frak l (V)$ be a
representation. Let $\rho^*:{\mathcal G}\rightarrow \frak g\frak l (V^*)$ be the
dual representation of $\rho$.

{\rm (a)}\quad A linear map $T:V\rightarrow {\mathcal G}$ is an
${\mathcal O}$-operator of $\mathcal G$ associated to $\rho$ if and
only if $r=T-\sigma(T)$ is an antisymmetric solution of the
classical Yang-Baxter equation in ${\mathcal G}\ltimes_{\rho^*}V^*$.

{\rm (b)}\quad Let $T:V\rightarrow {\mathcal G}$ be an ${\mathcal
O}$-operator associated to $\rho$. Then $r=T+\sigma(T)$ is a
symmetric solution of the $S$-equation in
$T(V)\ltimes_{\rho^*,0}V^*$, where $T(V)\subset {\mathcal G}$ is a
pre-Lie algebra given by equation ({\rm A10}) and $(\rho^*,0)$ is a
bimodule since its sub-adjacent Lie algebra ${\mathcal G}(T(V))$ is
a Lie subalgebra of ${\mathcal G}$, and $T$ can be identified as an
element in $T(V)\otimes V^*\subset
(T(V)\ltimes_{\rho^*,0}V^*)\otimes (T(V)\ltimes_{\rho^*,0}V^*)$.}

\noindent {\bf Proposition A9.} \quad {\it Let $(A,\cdot)$ be a
pre-Lie algebra. Let $\{e_1,\cdots, e_n\}$ be a basis of $A$ and
$\{e_1^*,\cdots, e_n^*\}$ be its dual basis.

{\rm (a)}\quad $r$ given by equation {\rm (2.5.3)} is an
antisymmetric solution of the classical Yang-Baxter equation in
${\mathcal G}(A)\ltimes_{L^*} {\mathcal G}(A)^*$.  Moreover, $r$ is
nondegenerate and the induced 2-cocycle ${\mathcal B}$ of ${\mathcal
G}(A)\ltimes_{L^*} {\mathcal G}(A)^*$ is given by equation {\rm
(1.4.1)}.

{\rm (b)}\quad $r$ given by equation {\rm (4.4.23)} is a symmetric
solution of the $S$-equation in $A\ltimes_{L^*,0} A^*$.  Moreover,
$r$ is nondegenerate and the induced 2-cocycle ${\mathcal B}$ of
$A\ltimes_{L^*,0} A^*$ is given by equation {\rm (1.1.1)}}.

\noindent {\bf Theorem A10}\quad {\it Let $(A,A^*,\Delta,\beta)$ be
a pre-Lie bialgebra. Then $({\mathcal G}(A),{\mathcal G}(A^*))$ is a
Lie bialgebra if and only if
$$\langle R_\cdot^*(x)a^*,R_\circ^*(b^*)y\rangle  +\langle R_\cdot^*(x)b^*,R_\circ^*(a^*)y\rangle
=\langle R_\cdot^*(y)b^*,R_\circ^*(a^*)x\rangle  +\langle
R_\cdot^*(y)a^*,R_\circ^*(b^*)x\rangle  , \eqno ({\rm A11})$$ for
any $x,y\in A^*,a^*,b^*\in A^*$.}


\begin{thebibliography}{999}
\bibitem[A1]{A1} M. Aguiar,
Infinitesimal Hopf algebras, Contemporary Mathematics 267, Amer.
Math. Soc., (2000) 1-29.
\bibitem[A2]{A2} M. Aguiar, Pre-Poisson algebras, Lett.
Math. Phys. 54 (2000) 263-277.
\bibitem[A3]{A3} M. Aguiar, On the
associative analog of Lie bialgebras, J. Algebra 244 (2001), no. 2,
492-532.
\bibitem[A4]{A4} M. Aguiar, Infinitesimal bialgebras, pre-Lie
algebras and dendriform algebras, in ``Hopf algebras", Lecture Notes
in Pure and Appl. Math. 237 (2004) 1-33.
\bibitem[Bai1]{Bai1} C. Bai, A unified algebraic approach to the classical
Yang-Baxter equation, J. Phy. A: Math. Theor. 40 (2007) 11073-11082.
\bibitem[Bai2]{Bai2} C. Bai, Left-symmetric bialgebras and an
analogue of the classical Yang-Baxter equation, Comm. Comtemp. Math.
10 (2008) 221-260.

\bibitem[BGN1]{BGN1} C. Bai, L. Guo and X. Ni, $\mathcal O$-operators on associative algebras and associative Yang-Baxter equations, arXiv:0910.3261.
\bibitem[BGN2]{BGN2} C. Bai, L. Guo and X. Ni, $\mathcal O$-operators on associative algebras and dendriform algebras, arXiv:1003.2432.





\bibitem[BaN]{BaN} A.A. Balinskii, S.P. Novikov,
Poisson brackets of hydrodynamic type, Frobenius algebras and Lie
algebras, Soviet Math. Dokl. 32 (1985) 228-231.
\bibitem[Bax]{Bax} G. Baxter, An analytic problem whose solution follows
from a simple algebraic identity, Pacific J. Math. 10 (1960)
731-742.
\bibitem[Bo]{Bo}
M. Bordemann, Nondegenerate invariant bilinear forms on
nonassociative algebras, Acta Math. Univ. Comen. LXVI (1997)
151-201.
\bibitem[BFN]{BFN} M. Bordemann, T. Filk, C. Nowak, Algebraic
classification of actions invariant under generalized flip moves of
2-dimensional graphs, J. Math. Phys. 35 (1994) 4964-4988.
\bibitem[BrN]{BrN} R. Brauer, C. Nesbitt, On the regular
representations of algebras, Proc. Nat. Acad. Sci. USA 23 (1937)
236-240.
\bibitem[Bu]{Bu} D. Burde, Left-symmetric algebras and pre-Lie algebras in
geometry and physics, Cent. Eur. J. Math. 4 (2006) 323-357.
\bibitem[Cha1]{Cha1} F. Chapoton, Un endofoncteur de la
cat\'egorie des op\'erades. Dialgebras and related operads, 105-110,
Lecture Notes in Math. 1763, Springer, Berlin, 2001.
\bibitem[Cha2]{Cha2} F. Chapoton,
Un th\'eor\`eme de Cartier-Milnor-Moore-Quillen pour les big\`ebres
dendriformes et les alg\`ebres braces, J. Pure and Appl. Alg. 168
(2002) 1-18.
\bibitem[CP]{CP} V. Chari, A. Pressley, A
guide to quantum groups, Cambridge University Press, Cambridge
(1994).
\bibitem[Chu]{Chu} B.Y. Chu, Symplectic homogeneous spaces,
Trans. Amer. Math. Soc. 197 (1974) 145-159.
\bibitem[C]{C} A. Connes, Non-commutative differential geometry,
Inst. Hautes \'Etudes Sci. Publ. Math. 62 (1985) 257-360.
\bibitem[D]{D} V. Drinfeld, Hamiltonian structure on the Lie groups,
Lie bialgebras and the geometric sense of the classical Yang-Baxter
equations, Soviet Math. Dokl. 27 (1983) 68-71.
\bibitem[E1]{E1} K.
Ebrahimi-Fard, Loday-type algebras and the Rota-Baxter relation,
Lett. Math. Phys. 61 (2002), no. 2, 139-147. \bibitem[E2]{E2} K.
Ebrahimi-Fard, On the associative Nijenhuis relation, Elect. J.
Comb.,11 (2004), no. 1, Research Paper 38.
\bibitem[EMP]{EMP} K. Ebrahimi-Fard, D. Manchon, F. Patras, New
identities in dendriform algebras, J. Algebra 320 (2008), no. 2,
708-727.
\bibitem[F1]{F1} L. Foissy, Les alg\`ebres de Hopf des arbres
enracin\'es d\'ecor\'es II, Bull. Sci. Math. 126 (2002) 249-288.
\bibitem[F2]{F2} L. Foissy, Bidendriform bialgebras, trees, and
free quasi-symmetric functions, J. Pure. Appl. Algebra 209 (2007)
439-459.
\bibitem[Fra1]{Fra1} A. Frabetti, Dialgebra homology of associative
algebras, C. R. Acad. Sci. Paris 325 (1997) 135-140.
\bibitem[Fra2]{Fra2} A. Frabetti, Leibniz homology of dialgebras of
matrices, J. Pure. Appl. Alg. 129 (1998) 123-141.
\bibitem[Fro]{Fro} G. Von Frobenius, Theorie der hyperkomplexen
Gr\"o$\beta$en, Sitzung der phys.-math. Kl. (1903) 504-538, 634-664.
\bibitem[H1]{H1} R. Holtkamp, Comparison of Hopf algebras on trees,
Arch. Math. (Basel) 80 (2003) 368-383.
\bibitem[H2]{H2} R. Holtkamp, On Hopf algebra structures over free operads,
Adv. Math. 207 (2006) 544-565.
\bibitem[JR]{JR} S.A. Joni, G.C. Rota, Coalgebras and bialgebras in
combinatories, Stud. Appl. Math. 61 (1979) 93-139.
\bibitem[Kac]{Kac} V.G. Kac, Infinite dimensional Lie algebras,
Progress in Math. 44, Boston-Basel-Stuttgart, 1983.
\bibitem[Kan]{Kan} S. Kaneyuki, Homogeneous
symplectic manifolds and dipolarizations in Lie algebras, Tokyo J.
Math. 15 (1992) 313-325.
\bibitem[Kap]{Kap} G. Karpilovsky, Symmetric and $G$-algebras.
With applications to group representations. Mathematics and its
Applications 60, Kluwer Academic Publishers Group, Dordrecht, 1990.
\bibitem[Ko]{Ko} J. Kock, Frobenius algebras and 2D topological quantum
field theories,  London Mathematical Society Student Texts, 59.
Cambridge University Press, Cambridge, 2004.
\bibitem[Ku1]{Ku1} B.A. Kupershmidt, Non-abelian phase spaces, J.
Phys. A: Math. Gen. 27 (1994), no. 8, 2801-2809.
\bibitem[Ku2]{Ku2}B.A. Kupershmidt, On the nature of the Virasoro algebra, J.
Nonlinear Math. Phy. 6 (1999), no. 2, 222-245.
\bibitem[Ku3]{Ku3} B.A. Kupershmidt, What a classical $r$-matrix really is, J.
Nonlinear Math. Phy. 6 (1999), no. 4, 448-488.
\bibitem[LS]{LS} R.G. Larson, M.E. Sweedler, An associative
orthogonal bilinear form for Hopf algebras, Amer. J. Math. 91 (1969)
75-94.
\bibitem[Li]{Li} P. Libermann,
Sur le probl\`eme d'\'equivalence de certaines structures
infinit\'esimales, Ann. Mat. Pura Appl. 36 (1954) 27-120.
\bibitem[Lo1]{Lo1}
J.-L. Loday, Dialgebras. Dialgebras and related operads, 7-66,
Lecture Notes in Math. 1763, Springer, Berlin, 2001.
\bibitem[Lo2]{Lo2} J.-L. Loday, Arithmetree, J. Algebra 258 (2002)
275-309.
\bibitem[Lo3]{Lo3}
J.-L. Loday, Scindement d'associativit\'e et alg\`ebres de Hopf,
Proceedings of the Conference in hounor of Jean Leray, Nantes
(2002), S\'eminaire et Congr\`es (SMF) 9 (2004) 155-172.
\bibitem[Lo4]{Lo4} J.-L. Loday, Generalized bialgebras and triples
of operads, arXiv: math.QA/061188.
\bibitem[LR1]{LR1} J.-L. Loday, M.
Ronco, Hopf algebra of the planar binary trees, Adv. Math. 139
(1998) 293-309.
\bibitem[LR2]{LR2} J.-L. Loday, M.
Ronco, Order structure on the algebra of permutations and of planar
binary trees, J. Alg. Comb. 15 (2002) 253-270.
\bibitem[MR1]{MR1} A. Medina, P. Revoy, Caract\'erisation des
groupes de Lie ayant une pseudom\'etrique bi-invariante,
Applications S\'eminaire Sud-Rhodanien de G\'eometrie III, Lyon,
12.-17.10.1983, Paris, 1984.
\bibitem[MR2]{MR2} A. Medina, P. Revoy, Alg\`ebres de Lie et produit
scalaire invariant, Ann. Sci. Ecole Normale Sup. 18 (1985) 553-561.
\bibitem[Ron]{Ron} M. Ronco, Eulerian idempotents and Milnor-Moore
theorem for certain non-cocommutative Hopf algebras, J. Algebra 254
(2002) 151-172.
\bibitem[Rot]{Rot} G.-C. Rota, Baxter operators, an introduction,
In: ``Gian-Carlo Rota on Combinatorics, Introductory papers and
commentaries", Joseph P.S. Kung, Editor, Birkh\"auser, Boston, 1995.
\bibitem[RFFS]{RFFS} I. Runkel, J. Fjelstad, J. Fuchs, C. Schweigert,
Topological and conformal field theory as Frobenius algebras,
Contemp. Math. 431 (2007) 225-247.
\bibitem[Sc]{Sc} R. Schafer, An introduction to nonassociative
algebras, Dover Publications Inc., New York (1995).
\bibitem[Se]{Se} M.A. Semenov-Tian-Shansky, What is a classical
R-matrix? Funct. Anal. Appl. 17 (1983) 259-272.
\bibitem[Sh]{Sh} H. Shima, Homogeneous
Hessian manifolds, Ann. Inst. Fourier 30 (1980) 91-128.
\bibitem[St]{St} A. Stolin, Frobenius algebras and the Yang-Baxter equation,
New symmetries in the theories of fundamental interactions, 93-97,
PWN, Warsaw, 1997.
\bibitem[U]{U} K. Uchino,
Quantum analogy of Poisson geometry, related dendriform algebras
and Rota-Baxter operators, Lett. Math. Phys. 85 (2008) 91-109.
\bibitem[Y]{Y} K. Yamagata, Frobenius algebras,
Handbook of algebra, Vol. 1, 841-887, North-Holland, Amsterdam,
1996.
\bibitem[Z]{Z} V. N. Zhelyabin, Jordan bialgebras
and their connection with Lie bialgebras, Algebra i Logika 36 (1997)
(1) 3-35; English tranl., Algebra and Logic 36 (1997) (2) 1-15.


\end{thebibliography}
\end{document}